\documentclass[12pt]{article}

\pdfoutput=1

\usepackage{latexsym,amsfonts,bm,epsfig,amsmath,natbib,authblk,amsthm,thmtools,amssymb}

\usepackage{float}
\usepackage{booktabs}
\usepackage{graphicx}
\usepackage{subcaption}
\usepackage[margin=.5cm]{caption}
\usepackage{color}
\usepackage{pgfplots}
\usepackage{tikz}
\newlength\figureheight
\newlength\figurewidth
\usepackage{xcite, xr}

\floatstyle{plain}
\newfloat{Algorithm}{thp}{lop}
\floatname{Algorithm}{Algorithm}

\usepackage[plain,noend]{algorithm2e}

\makeatletter
\renewcommand{\algocf@captiontext}[2]{#1\algocf@typo. \AlCapFnt{}#2} 
\def\@algocf@capt@plain{top}
\renewcommand{\algocf@makecaption}[2]{%
	\addtolength{\hsize}{\algomargin}%
	\sbox\@tempboxa{\algocf@captiontext{#1}{#2}}%
	\ifdim\wd\@tempboxa >\hsize
	\hskip .5\algomargin%
	\parbox[t]{\hsize}{\algocf@captiontext{#1}{#2}}
	\else%
	\global\@minipagefalse%
	\hbox to\hsize{\box\@tempboxa}
	\fi%
	\addtolength{\hsize}{-\algomargin}%
}
\makeatother

\newcommand{\argmin}{\operatornamewithlimits{argmin\,}}

\newcommand{\E}{\mathbb{E}}

\mathchardef\mhyphen="2D

\def \t{\theta}

\def \E{\mathbb{E}}

\def \dt {\mathrm{d}}

\def \MD{\mathcal{D}}

\newtheorem{assumption}{Assumption}
\newtheorem{theorem}{Theorem}

\newtheorem{corollary}{Corollary}
\newtheorem{lemma}{Lemma}
\theoremstyle{remark}
\newtheorem{remark}{Remark}

\makeatletter
\newcommand{\myitem}[1]{%
	\item[#1]\protected@edef\@currentlabel{#1}%
}
\makeatother

\begin{document}
	\def\spacingset#1{\renewcommand{\baselinestretch}%
		{#1}\small\normalsize} \spacingset{1}
	
	\title{Calibrated Generalized Bayesian Inference}
	\date{\empty}
	\author[1]{David T. Frazier\thanks{Corresponding author:  david.frazier@monash.edu}}
	\author[2]{Christopher Drovandi}
	\author[3]{Robert Kohn}
	\affil[1]{Department of Econometrics and Business Statistics, Monash University, Clayton VIC 3800, Australia}
	\affil[2]{School of Mathematical Sciences, Queensland University of Technology, Brisbane 4000 Australia}
	\affil[3]{Australian School of Business, School of Economics, University of New South Wales, Sydney NSW 2052, Australia}
	
	\maketitle


\maketitle

\begin{abstract}
We propose a simple approach that provides accurate uncertainty quantification for Bayesian inference in misspecified or approximate models, and for generalized (Gibbs) posteriors. While existing solutions in this context are based on explicit Gaussian approximations or post-processing procedures, we demonstrate that correct uncertainty quantification can be achieved by substituting the usual posterior with an intuitively appealing alternative that conveys the same information. This solution applies to both likelihood-based and loss-based posteriors, and is formally demonstrated to reliably quantify uncertainty. This new approach is demonstrated through a range of examples, including generalized linear models, and doubly intractable models.
 
\vspace{.5cm}

\noindent \textsc{Keywords}: Generalized Bayesian inference, Gibbs posteriors, Model misspecification, Calibration.
\end{abstract}

	\spacingset{1.8} 

\section{Introduction}

Bayesian methods are lauded for their ability to tackle complicated models, deftly handle latent variables, and for providing a holistic treatment for the uncertainty of model unknowns. While it is well known that Bayesian methods deliver reliable inferences in well-specified models, when the model used to define the posterior is misspecified, the shortcomings of Bayesian posteriors are well known and much recent literature is devoted to correcting these issues; see \cite{nott2023bayesian} for a review. One popular strategy for delivering reliable Bayesian inference in misspecified models is to produce posterior beliefs using ``Gibbs'' posteriors, which updates prior beliefs for a parameter of interest using a general loss function; see, e.g., \citet{zhang2006information}, and \citet{bissiri2016general} for specific architecture, and \citet{syring2020robust}, \cite{matsubara2022robust}, \cite{jewson2021general}, and \cite{loaiza2021focused} for specific examples. 

{\color{black}Since loss functions have arbitrary scale, implementation of Gibbs posteriors requires specifying a learning rate that ``controls the relative weight of loss [...] to prior'' (\citealp{bissiri2016general}). However, for general choices of the learning rate, the probability statements made by Gibbs posteriors are inaccurate: their posterior probability statements \textit{do not have} the asserted coverage in repeated samples, and so are \textit{not calibrated}.} Calibration of Gibbs posteriors is crucial if one wishes to make scientific statements that can be reliably refuted through empirical analysis (\citealp{rubin1984bayesianly}); see, also, \cite{dawid1982well}, and \cite{little2006calibrated} for additional discussion on the importance of calibrated Bayesian inference.

{\color{black} \cite{syring2019calibrating} were the first to propose general methods aimed at tuning the learning rate in Gibbs posteriors to deliver calibrated credible sets. Their proposed approach, which is based on bootstrapping and can be computationally intensive, can be shown to deliver exact calibration for multivariate parameters under a generalized information equality. If the generalized information equality does not hold, exact calibration becomes infeasible but the procedure can be tuned to deliver conservative calibration. We refer to \cite{wu2020comparison} for further discussion, as well as a comparison on alternative ways of tuning the learning rate (e.g., \citealp{holmes2017assigning}, and \citealp{lyddon2019general}). Alternatively, \citet{muller2013risk} and \citet{giummole2019objective} suggest using explicit Gaussian posterior ``corrections'' to deliver calibrated inferences in misspecified models. These corrections amount to replacing the posterior directly by a certain Gaussian density, which may then deliver inaccurate results in small samples, if the original posterior is non-Gaussian, or if the second derivatives of the loss are ill-behaved.
}


{\color{black}Herein, we leverage the general definition of Gibbs posteriors -- as a belief update defined by a loss function and learning rate combination --  to devise an approach that delivers calibrated inferences without requiring learning rate tuning, posterior corrections or bootstrapping. We show that, under regularity conditions, the proposed posterior's credible sets are automatically calibrated asymptotically when the learning rate is set to unity, and we refer to this approach as the asymptotically calibrated posterior (ACP)}. {\color{black}When the loss defining the unknown parameter of interest is sufficiently smooth in expectation, the ACP allows} the ``statistician to be Bayesian in principle and calibrated to the real world in practice'' (\citealp{rubin1984bayesianly})

{\color{black}While the ACP delivers calibrated inferences, it is not coherent in the sense of \citet{bissiri2016general}. This is unsurprising since many of the losses considered for Gibbs posteriors are themselves not coherent, and since all existing methods that seek to deliver calibrated inferences in Gibbs posteriors are not coherent.}

The remainder of the paper proceeds as follows. Section \ref{sec:misspec} discusses the general issue of model misspecification in likelihood-based Bayesian inference, and  Section \ref{sec:cov} demonstrates how a particular Gibbs posterior approach overcomes the known issues with Bayesian inference in this setting. This section also includes two preliminary examples demonstrating the behavior of the proposed method. Section \ref{sec:examp} examines the empirical performance of the proposed approach in exact and generalized Bayesian inference paradigms, including two examples of doubly intractable models, and a situation where a computationally convenient approximating model is used instead of the actual model. In each example, our proposed approach delivers reliable and correctly calibrated Bayesian inferences.  Section \ref{sec:theory} proves that the proposed approach correctly quantifies uncertainty. Section \ref{sec:discuss} compares the new approach with existing approaches that attempt to produce reliable Bayesian inference in these settings and Section \ref{sec:conclusions} concludes. Online supplementary material contains proofs of all stated results.

\section{Bayesian Inference in Misspecified Models}\label{sec:misspec}
\subsection{Setup and Known Issues}
The observed data is $y:=(y_1,\dots,y_n)^\top$, where $y_i\in\mathcal{Y}\subseteq\mathbb{R}^{d_y}$ ($i=1,\dots,n$), is generated from some true unknown probability distribution $P^{(n)}_0$. Generally, in Bayesian inference one approximates the unknown $P^{(n)}_0$ using a class  of models $\{P^{(n)}_\theta:\theta\in\Theta\}$ depending on the unknown parameter $\theta\in\Theta\subseteq\mathbb{R}^{d_\theta}$, for which we have prior beliefs $\pi(\theta)$, and where the number of parameters $d_\theta$ is fixed with $n$.

However, when the model $P^{(n)}_\theta$ is misspecified, it is well-known that uncertainty quantification using the standard Bayesian posterior is unreliable (\citealp{kleijn2012bernstein}), so that the very parameters we are conducting inference on may be ``meaningless -- except perhaps as descriptive statistics'' (\citealp{freedman2006so}). When $P^{(n)}_\theta$ is possibly misspecified, it may instead make sense to eschew the likelihood and consider classes of loss functions that are specific for the inferential task at hand, {\color{black}e.g., such as loss functions that are robust to data contamination (see Section \ref{sec:ksd} for specific examples).} 

In such cases, if one is willing to take a more general approach to updating prior beliefs, we can follow the ideas of \cite{zhang2006information}, \cite{bissiri2016general}, and \cite{knoblauch2022optimization} {\color{black}and produce a \textit{Gibbs   posterior} for $\theta$.} Gibbs posteriors depend on a user-chosen loss function $\mathcal{D}_n:\Theta\rightarrow\mathbb{R}_{}$, associated with the sample $y$, and provide a belief update for $\theta$ by solving the variational optimization problem
\begin{flalign}\label{eq:gibbs}
\pi(\theta\mid \mathcal{D}_n):=\argmin_{\rho\in\mathcal{P}(\Theta)}\left\{\omega \int_{\Theta}\mathcal{D}_n(\theta)\rho(\theta)\dt\theta+\mathcal{K}(\rho\|\pi)\right\}.
\end{flalign}{\color{black}In equation \eqref{eq:gibbs}, $\mathcal{P}(\Theta)$ denotes a set of probability measures over $\Theta$, $\mathcal{K}(\rho\|\pi)$ denote a divergence between $\rho(\theta)$ and the prior $\pi(\theta)$. The parameter $\omega\ge0$, often called the learning rate, plays a fundamental role in the definition of $\pi(\theta\mid\mathcal{D}_n)$ as it scales the loss  $\mathcal{D}_n(\theta)$ relative to the prior $\pi(\theta)$ and determines the relative weight each is given in the belief update. In this way, the selection of a reasonable value for $\omega$ is necessary and crucial since the loss and prior can be on entirely different scales. 
} 

{{\color{black}When the loss is integrable with respect to the prior $\pi(\theta)$, and when we set $\mathcal{K}(\rho\|\pi)=\text{KL}(\rho\|\pi)$ -- the Kullback--Leibler  divergence -- the unique solution to \eqref{eq:gibbs} takes the recognizable form:}
\begin{flalign*}
\pi(\theta\mid \mathcal{D}_n):=\frac{ \pi(\theta)\exp\{-\omega\cdot \mathcal{D}_{n}(\theta)\}}{\int_\Theta \pi(\theta)\exp\{-\omega\cdot \mathcal{D}_{n}(\theta)\}\dt\theta},
\end{flalign*}
{\color{black}	where the notation $\pi(\t\mid \mathcal{D}_n)$ encodes the posterior's dependence on the chosen loss $\mathcal{D}_n(\theta)$. This particular formulation of $\pi(\theta\mid\mathcal{D}_n)$ clarifies the critical role played by $\omega$: since the scale of the loss $\mathcal{D}_n(\theta)$ is arbitrary, $\omega$ is required to ensure that $\exp\{-\omega\cdot\mathcal{D}_n(\theta)\}$ plays the role of a likelihood in a standard Bayesian posterior formulation.}
\begin{remark}
Standard Bayesian inference is recovered from \eqref{eq:gibbs} by taking $\omega=1$ and setting  $\mathcal{D}_n(\theta)=-\log p_\theta^{(n)}(y)$, where $p_\theta^{(n)}(y)$ is the assumed model density. Further, if the loss is additive, so that $\mathcal{D}_n(\theta)=\sum_{i=1}^{n}d(y_i,\theta)$ for some known loss function $d:\mathcal{Y}\times\Theta\rightarrow\mathbb{R}_+$,  we obtain the generalized posterior of \cite{bissiri2016general}. The measure $\pi(\theta\mid \mathcal{D}_n)$ is often referred to as a Gibbs posterior, and sometimes a generalized posterior.
\end{remark}

{\color{black}Gibbs posteriors implicitly conduct inference on the value of $\theta$ that minimizes the limiting expected loss: $\theta_\star:=\argmin_{\theta\in\Theta}\lim_n\E_{}n^{-1}\mathcal{D}_n(\theta)$,  where $\E(\cdot)$ denotes expectation under $P^{(n)}_0$.} However, it is known that the posterior  $\pi(\theta\mid \mathcal{D}_n)$ does not  deliver inference on $\theta_\star$ that is calibrated: {\color{black} a credible set for $\theta_\star$ based on $\pi(\theta\mid \mathcal{D}_n)$ and having posterior probability $(1-\alpha)$  is \textit{calibrated} if the credible set asymptotically contains $\theta_\star$ with $P^{(n)}_0$ - probability at least $(1-\alpha)$.
	}

 {\color{black}While a Gibbs posterior will often satisfy a Bernstein-Von Mises (BvM) theorem, the asymptotic covariance matrix of the Gibbs posterior will not deliver calibrated inferences in general; see \citet{kleijn2012bernstein} and \cite{miller2021asymptotic} for a discussion on BvMs and their required regularity conditions. In large samples, the posterior $\pi(\theta\mid\MD_n)$ behaves like a Gaussian density with mean $\theta_\star$, and variance $[n\omega\mathcal{H}(\theta_\star)]^{-1}$, where $\mathcal{H}_n(\theta):=n^{-1}\nabla_{\theta\theta}^2 \mathcal{D}_n(\theta)$ denotes the  Hessian of $n^{-1}\mathcal{D}_n(\theta)$, and $\mathcal{H}(\theta):=\lim_n\E\{\mathcal{H}_n(\theta)\}$ its limit. Conversely, the posterior mean of $\pi(\theta\mid\MD_n)$ converges towards $\theta_\star$ but has a ``sandwich form'' asymptotic variance: $\Sigma_\star:=\mathcal{H}(\theta_\star)^{-1}\mathcal{I}(\theta_\star)\mathcal{H}(\theta_\star)^{-1}/n$, where $\mathcal{I}(\theta)=\lim_n\text{Cov}\left[n^{-1/2}\mathcal{D}_n(\theta)\right]$. Because the covariance matrix appearing in the BvM for $\pi(\theta\mid \mathcal{D}_n)$ is not the sandwich covariance matrix, its $(1-\alpha)$ credible set will not contain $\theta_\star$ with $P^{(n)}_0$ - probability $(1-\alpha)$ in general.}

Thus, posterior credible sets built from Gibbs posteriors are not calibrated in general, and Bayesian uncertainty quantification may not be reliable in many empirical applications.  The lack of calibrated Bayesian inference in misspecified model  has led researchers to consider approaches that `correct' this issue. However, the suggested approaches of which we are aware are either based on computationally onerous bootstrapping procedures, which must bootstrap the entire posterior distribution multiple times and necessitates running many MCMC samplers - one for each bootstrap replication of the data - as well as ensuring that each MCMC sampler converges to the appropriate distribution, such as in \cite{huggins2019robust}, or \cite{matsubara2022robust}. 

{\color{black}Alternatively,  one could apply ex-post corrections to $\pi(\theta\mid \mathcal{D}_n)$. These approaches substitute the posterior $\pi(\theta\mid \mathcal{D}_n)$ for a Gaussian density with mean $\bar\theta=\int \theta\pi(\theta\mid \mathcal{D}_n)\dt\theta$ and covariance $\widehat\Sigma_n/n$, where $\widehat\Sigma_n=[\mathcal{H}_n(\bar\theta)W_n(\bar\theta)^{-1}\mathcal{H}_{n}(\bar\theta)]^{-1}$ is an estimator of $\Sigma_\star$. These ex-post correction methods amount to assuming that the posterior  is exactly Gaussian, rather than the much weaker requirement that it converges to a Gaussian distribution in the limit.}

\subsection{Automatically Calibrated  Bayesian Uncertainty}\label{sec:cov}

{\color{black}While Gibbs posteriors lack calibration, this does not mean that we should abandon this principled way of updating prior beliefs. Instead, a clear solution to this problem is to develop a Gibbs posterior that satisfies the following two properties.
\begin{itemize}
\myitem{\textbf{P.1}} The posterior  targets the same population loss minimizer as $\pi(\theta\mid \mathcal{D}_n)$, i.e., $\theta_\star$.\label{item:P1}
\myitem{\textbf{P.2}} The posterior reliably quantifies uncertainty, at least asymptotically. \label{item:P2}
\end{itemize}

\noindent Herein, we show that a posterior satisfying  \ref{item:P1}-\ref{item:P2} can be constructed by using the machinery of Gibbs posteriors through a specific \textit{loss and learning rate combination}. If we use the KL divergence in \eqref{eq:gibbs} and replace the original loss $\mathcal{D}_n(\theta)$ in \eqref{eq:gibbs} with 
$$
\mathcal{Q}_n(\theta):=\frac{1}{2}\log|W_n(\theta)|+n\cdot Q_n(\theta),\;\text{ where }Q_n(\theta):=\frac{1}{2}\cdot{\overline{m}_n(\theta)^\top}{}W_{n}(\theta)^{-1}{\overline{m}_n(\theta)}{},
$$ {\color{black}$\overline{m}_n(\theta)=\nabla_\theta \MD_n(\theta)/n$ and $W_n(\theta)$ is a $(d_\theta\times d_\theta)$--dimensional covariance matrix estimator of $\mathcal{I}(\theta)$ -- the covariance of $\overline{m}_n(\theta)$ -- we obtain the Gibbs posterior:} 
\begin{equation}\label{eq:genpost}
	\pi(\theta\mid \mathcal{Q}_n):=\frac{M_n(\theta)^{-\frac{\omega}{2}}\exp\{- \omega\cdot n\cdot Q_n(\theta)\}\pi(\theta)}{\int_\Theta M_n(\theta)^{-\frac{\omega}{2}}\exp\{- \omega\cdot n\cdot Q_n(\theta)\}\pi(\theta)\dt\theta},\quad  M_n(\theta)=|W_{n}(\theta)|.
\end{equation}

In contrast to most Gibbs posteriors, the term $M_n(\theta)^{-\frac{\omega}{2}}\exp\{-\omega\cdot n\cdot{Q}_n(\theta)\}$ in $\pi(\theta\mid\mathcal{Q}_n)$ resembles, and behaves like, a multivariate Gaussian likelihood with variance $W_n(\theta)/n$. As a consequence, when using the loss $\mathcal{Q}_n(\theta)$ to construct a Gibbs posterior, the \textit{natural default choice for the learning rate is to take $\omega=1$}. Such an automatic choice for the learning rate $\omega$ is generally absent for Gibbs posteriors, at least outside of scalar settings. In addition, since $Q_n(\theta)$ is a quadratic function in  $\overline{m}_n(\theta)$, the posterior $\pi(\theta\mid\mathcal{Q}_n)$ will concentrate onto the population risk minimizer $\theta_\star$ and satisfy \ref{item:P1}. Thus, if $\pi(\theta\mid\mathcal{Q}_n)$ reliably quantifies uncertainty, then Gibbs posteriors can satisfy points \ref{item:P1}-\ref{item:P2} through a particular loss and learning rate combination.

Section \ref{sec:theory} rigorously shows that, if the risk minimizer is unique, $\pi(\theta\mid \mathcal{Q}_n)$ automatically satisfies \ref{item:P2} when $W_n(\theta)$ is a consistent estimator of $\mathcal{I}(\theta)=\lim_n\text{Cov}\{\sqrt{n}\overline{m}_n(\theta)\}$. In this case, $(1-\alpha)$ credible sets calculated from $\pi(\theta\mid \mathcal{Q}_n)$ have $P^{(n)}_0$- probability $(1-\alpha)$ as desired (under regularity conditions).

Traditional Gibbs posteriors must tune the learning rate $\omega$ to deliver reliable credible sets. Our approach differs fundamentally: the posterior $\pi(\theta\mid \mathcal{Q}_n)$ asymptotically delivers calibrated credible sets under a default choice for $\omega$, without any tuning. Throughout the remainder of the paper we refer to $\pi(\theta\mid \mathcal{Q}_n)$ as the 
\textit{asymptotically calibrated posterior} (ACP) to differentiate it from other Gibbs posterior formulations. The large sample calibration of $\pi(\theta\mid\mathcal{Q}_n)$ can be established heuristically using a second-order expansion of $Q_n(\theta)$ around $\theta_n$; however, due to space restrictions we give such a discussion in Appendix \ref{app:soe}, where we also give two concrete examples that illustrate the general structure of $\pi(\theta\mid \mathcal{Q}_n)$.


  
  {\color{black} While the ACP is formally defined via the variational optimization problem in \eqref{eq:gibbs}, with $\mathcal{D}_n(\theta)=\mathcal{Q}_n(\theta)$, its equivalent representation in \eqref{eq:genpost} ensures that standard MCMC methods can be used to produce samples from $\pi(\theta\mid\mathcal{Q}_n)$. Throughout the numerical experiments, we use adaptive MCMC or Hamiltonian Monte Carlo (HMC), implemented via {RSTAN}, to produce samples from the desired target.}

 {\color{black}
 \begin{remark}\label{rem:variance_estimator}
As discussed above, for $\pi(\theta\mid \mathcal{Q}_n)$ to correctly quantify uncertainty asymptotically, $W_n(\theta)$ must be a consistent estimator of $\mathcal{I}(\theta)=\lim_n\text{Cov}\{\sqrt{n}\overline{m}_n(\theta)\}$. {\color{black}Critically, specifying  $W_n(\theta)$ requires no inference on additional parameters. Further, reliable inference can be achieved by setting $\omega=1$ and taking  $W_n(\theta)$ to be the sample variance of $\overline{m}_n(\theta)$: when $\overline{m}_n(\theta)=n^{-1}\sum_{i=1}^{n}m_i(\theta)$, for some known functions $m_i(\theta)$,
\begin{equation}
{W}_n(\theta)=n^{-1}\sum_{i=1}^{n}\left\{m_i(\theta)-\overline{m}_n(\theta)\right\}\left\{m_i(\theta)-\overline{m}_n(\theta)\right\}^\top.
\label{eq:sampvar}	
\end{equation}If the observed data is dependent then the simple sample covariance may be inappropriate and a sample covariance estimator that allows for dependence should instead be used.}
 \end{remark}	
 	
}

{\color{black}
\begin{remark}
	Our restriction to fixed $d_\theta$ rules out cases where $d_\theta$ is allowed to increase as the sample size $n$ increases and limits our analysis. However, due to their generality and lack of an underlying probabilistic modeling structure, existing applications of Gibbs posteriors are often limited to cases where $d_\theta$ is not too large. Given this, we leave the extension of our analysis to cases where $d_\theta$ can increase with $n$ for future research. However, in the following experiments the ACP is applied in problems where $d_\theta$ is up to twenty using simple sample variance estimators and delivers reliable results with reasonable computing times.
\end{remark}
}

{\color{black}
	\begin{remark}
		The form of $\pi(\theta\mid \mathcal{Q}_n)$ in \eqref{eq:genpost} resembles certain quasi-Bayesian posteriors that are constructed from exponentiated quadratic forms in sample moments; see, e.g., \cite{chernozhukov2003mcmc}, \cite{yin2009} and \cite{chib2018bayesian} for examples. In these approaches, a vector of correctly specified and over-identified sample moments or estimating equations is combined with a weighting matrix to form a quadratic form, which is then exponentiated and used as a kernel in MCMC. In contrast, the ACP is not built from a particular class of estimating equations, but through a transformation of a loss function that encodes the quantity of interest and explicitly acknowledges model misspecification. The ACP targets this quantity of interest through a general variational principle: it is the optimal solution to a variational optimization problem that depends on a specific loss, learning rate and divergence used in its definition. The generality of this framework means that the ACP can be directly applied to the wide-variety of loss functions considered in generalized Bayesian inference; see Section \ref{sec:discuss} for a more in-depth comparison. 
	\end{remark}
} 
} 
 
\section{Examples}\label{sec:examp}
{\color{black} Before formally demonstrating that the ACP correctly quantifies uncertainty, we present several examples which empirically illustrate the methods reliable uncertainty quantification across different choices of loss functions, and assumed models.
}

\subsection{Linear Regression}\label{sec:reg}

	Consider the standard linear regression model
	$$
	y_{i}=x_i^\top\beta+\epsilon_i,\;(i=1,\dots,n),
	$$where the error distribution is assumed to be $\epsilon_i\stackrel{iid}{\sim} N(0,\sigma^2)$, with $\sigma>0$ unknown, and $x_i$ and $\beta$ are $3\times1$-dimensional vectors, where $x_{i,1}=1$ for all $i=1,\dots,n$. For simplicity, we assign flat priors to $\beta$ and $\sigma$. While the mean component of the regression is correctly specified, the true error term is given by
	$
	\epsilon_i\stackrel{iid}{\sim}N(0, 1/3  +  1/3 \cdot |x_{2,i}|^\gamma + 1/3 \cdot |x_{3,i}|^\gamma),
	$ where $\gamma\ge0$, and $x_{j,i}$ denotes the $j$th element of the vector $x_i$. When $\gamma=0$ the assumed model with homoskedastic errors is correct, whereas if $\gamma\ne0$, the assumed model is misspecified.  Under correct specification, i.e.\ $\gamma = 0$, the true value of the standard deviation is $\sigma = 1$. 
	
The example compares the accuracy of the standard Bayes posterior and ACP in correctly specified and misspecified regimes. As a competitor we also consider {\color{black} an oracle method} that correctly models the form of the heteroskedasticity parametrically. In particular, we consider posterior inference based on the assumed model 
\begin{equation}\label{eq:bayes_lin} 
	y_{i}=x_i^\top\beta+\epsilon_i,\quad 
	\epsilon_i\stackrel{\mathrm{iid}}{\sim}N\left(0,\xi_1+\xi_2|x_{2,i}|^\gamma+\xi_3|x_{i,3}|^\gamma\right),\quad (i=1,\dots,n),
\end{equation}
where $\xi_1,\xi_2,\xi_3 \geq 0 $ are unknown coefficients, each with flat priors, and $\gamma$ is known. {\color{black}Such an approach is infeasible in practice since the form of heteroskedasticity is unknown}; however, since the form is known in this context, we consider the posterior for $\beta$ obtained from the model in \eqref{eq:bayes_lin} {\color{black}as an oracle benchmark.}  This  heteroskedastic robust Bayes (HrB) approach assumes that $\gamma$ is set to its true value.

We generate $n=100$ observations from the model under $\gamma=0$, and $\gamma=2$;  $x_{1,i}=1$ for all $i$, so that $\beta_1$ is the intercept, and $(x_{2,i},x_{3,i})^\top$ is generated for each  replication as bivariate independent standard Gaussian, and we set the true value of $\beta$ to be $\beta =(1,1,1)^\top$. When the model is misspecified in this manner,  the pseudo-true remains equal to the true value of $\beta$, while the pseudo-true value of $\sigma$ is no longer unity, and so we focus in this example on inferences for $\beta$. We replicate this design 1000 times to create 1000 observed datasets of size $n=100$. For each dataset we sample the Bayes posterior, ACP, and HrB-posterior using random walk Metropolis-Hastings (RWMH) with a Gaussian proposal kernel using the \texttt{AdaptiveMCMC.jl} library \citep{vihola2014ergonomic} in \texttt{Julia}.  All posteriors are approximated using 20000 samples with an initial 1000 iterations for burn-in. {\color{black}We also implement the Gaussian posterior correction method of \citet{muller2013risk} (referred to as PostCorr)} with a sandwich covariance matrix based on the heteroskedastic-consistent covariance estimator of \citet{white1980heteroskedasticity}. For each replication, the matrix $W_n(\theta)$ in the ACP is the sample variance of the score equations (based on the data in that replication); {\color{black}see appendix \ref{sec:gml_app} for further details on the specification of this matrix.} Across each dataset and method, the posterior bias, variance and marginal coverage for the regression coefficients are compared. 

Table \ref{tab:reg} summarizes the results and shows that the ACP is similarly located to standard Bayes, but has larger posterior variances under both regimes. In the homoskedastic regime ($\gamma=0$), the Bayes posterior and the ACP both have approximately 95\% coverage. In the heteroskedastic regime ($\gamma=2$), the coverage of standard Bayes is further away from the nominal level than the ACP, with the lowest coverage level around 87\%. 

The table also shows that when $\gamma = 2$  the coverage of the HrB-posterior (HrB), PostCorr and the ACP are similar, with the HrB-posterior having tighter credible intervals  because it correctly models the heteroskedasticity. Critically, however, the ACP obtains similar results to the infeasible HrB-posterior without modeling the heteroskedastic variance. This feature is extremely useful in practice since the only way to reliably model heteroskedasticity is to use nonparametric Bayesian methods, which introduces 
%
a substantial level of complexity into an otherwise simple inference problem. {\color{black}The PostCorr method produces reasonably accurate results under misspecification, which is not surprising as the sandwich matrix we use in this example is specifically designed for heteroskedastic linear regression models, but does deliver some under-coverage. In the Poisson regression example in the next subsection, the posterior correction also produces credible sets that display under-coverage, demonstrating that posterior corrections may be inaccurate even in regression models.}
	
	\begin{table}[H]
\hspace*{-1.25cm} 			
		\centering
		{\footnotesize
			\begin{tabular}{rrrrrrrrrrrrr}
				\hline\hline
				 & \multicolumn{3}{c}{ACP} & \multicolumn{3}{c}{SB} & \multicolumn{3}{c}{HrB} &  \multicolumn{3}{c}{PostCorr} \\\hline\hline
			$\beta_1$ & $-0$.0022 &  0.0109 &  0.956 &  $-0$.0022 &  0.0105 &  0.959 &  $-0$.002 &   0.0110 &   0.957 &  $-0$.0021 &  0.0099 &  0.951   \\
			$\beta_2$ & 0.0041 &  0.0123 &  0.951 &   0.004 &   0.0107 &  0.945 &   0.004 &  0.0112 &  0.948 &   0.0040 &   0.0098 &  0.927   \\
			$\beta_3$ & 0.0018 &  0.0125 &  0.958 &   0.0019 &  0.0108 &  0.946 &   0.002 &  0.0113 &  0.951 &  0.0019 &  0.0100   & 0.932  \\\hline
			\textbf{$\gamma=2$} & Bias & Var &Cover& Bias & Var &Cover&Bias & Var &Cover & Bias & Var &Cover \\ \hline
			$\beta_1$ & 	$-0$.0017 &  0.011 &  0.966 &  $-0$.0017 &  0.0103 &  0.957 &  $$-0$$.0007 &  0.008 &   0.953 &  $-0$.0017 &  0.0096 &  0.947  \\
			 $\beta_2$ &	0.0051 & 0.020 &   0.963 &   0.0049 &  0.0104 &  0.871 &   0.0046 &  0.0139 &  0.945 &   0.0049 &  0.0153 &  0.913  \\
			$\beta_3$ &	0.0032 & 0.020 &   0.970 &    0.0036 &  0.0105 &  0.872 &   0.0012 &  0.0141 &  0.954  &  0.0036 &  0.0154 & 0.928 \\\hline\hline
		\end{tabular}}
		 \captionsetup{width=1.1\linewidth}
		\caption{Results in the normal linear regression model for standard Bayes (SB), HrB-posterior (HrB), ACP and the Gaussian posterior correction (PostCorr). Bias is the bias of the posterior mean across the replications. Var is the average posterior variance across the replications. Cover is the actual posterior coverage, where the nominal level is set to 95\% for the experiments. {\color{black} The true value of the coefficients is $\beta=(1,1,1)^\top$.}}
		\label{tab:reg}
	\end{table}

We also explore how the ACP scales to higher dimensional problems by considering the same data generating process (DGP) with $d=20$ regressors and $n=1000$. In the DGP, the regressors $x_4$--$x_{20}$ do not appear in the true model for the variance, and the regression coefficients corresponding to $x_4$--$x_{20}$ are also set to zero.  We consider the correctly specified ($\gamma = 0$) and the misspecified ($\gamma = 2$) cases.  Here, the HrB model is defined as 
\begin{equation*}
	y_{i}=x_i^\top\beta+\epsilon_i,\quad 
	\epsilon_i\stackrel{\mathrm{iid}}{\sim}N\left(0,   \sum_{k=1}^{20} \xi_k |x_{k,i}|^\gamma \right),\quad (i=1,\dots,n).
\end{equation*}
As the model is higher dimensional, we use 50000 MCMC iterations after 1000 burn-in iterations, and repeat the analysis for 100 independent datasets.  The results are shown in Table \ref{tab:reg_d20_all} for three specific covariates. Results for all covariates are given in Appendix \ref{app:sup_results_reg} of the supplementary material.  In the correctly specified case, the ACP delivers inferences that are close to the exact, and HrB produces average variances that are significantly larger.  The inflated variances can be attributed to having to estimate twice the number of parameters where many of the true $\xi$ parameter values are 0.  {\color{black}For the misspecified model, exact inference results in some under-coverage for $\beta_2$ and $\beta_3$, whereas the ACP is more accurate.  The results for the ACP are similar to the Bayes posterior for other parameters.  Again, the average posterior variances are larger for HrB. For this larger sample size, the PostCorr method seems to produce more reliable credible sets.}
	\begin{table}[H]
	\hspace*{-1.25cm} 			
	\centering
	{\footnotesize
		\begin{tabular}{rrrrrrrrrrrrr}
			\hline\hline
			& \multicolumn{3}{c}{ACP} & \multicolumn{3}{c}{SB} & \multicolumn{3}{c}{HrB}& \multicolumn{3}{c}{PostCorr}\\\hline\hline
			\textbf{$\gamma=0$} & Bias & Var &Cover& Bias & Var &Cover&Bias & Var &Cover&Bias & Var &Cover\\ \hline
$\beta_2$ & $-0$.0025 &  0.0011 &  0.97 &  $-0$.0024 &  0.0010 &  0.97 &  $-0$.0029 &  0.0014 &  0.98 &  $-0$.0025 &  0.0010 &  0.97 \\
$\beta_3$ & 0.0057 &  0.0010 &  0.94 &  0.0056 &  0.0010 &  0.94 &  0.0049 &  0.0015 &  0.94 &  0.0056 &  0.0010 & 0.93 \\
$\beta_{20}$ & 0.0032&  0.0010&  0.97&  0.0032&  0.0010&  0.97&  0.003&  0.0014&  0.97&  0.0032&  0.0010&  0.97 \\ \hline \hline
			\textbf{$\gamma=2$} & Bias & Var &Cover& Bias & Var &Cover&Bias & Var &Cover&Bias & Var &Cover  \\ \hline
$\beta_2$ & $-0$.0026&  0.0017&  0.98&  $-0$.0028&  0.0010&  0.91&  $-0$.0028&  0.0019&  0.94&  $-0$.0028&  0.0016&  0.97 \\
			$\beta_3$ & 0.0070&  0.0017&  0.94&  0.0071&  0.0010&  0.87&  0.0056&  0.0018&  0.90&  0.0071&  0.0017&  0.94 \\
$\beta_{20}$ &  0.0037&  0.0010&  0.96&  0.0039&  0.0010&  0.96&  0.0019&  0.0014&  0.97&  0.0039&  0.0010&  0.97 \\ \hline \hline
	\end{tabular}}
	\captionsetup{width=1.1\linewidth}
	\caption{Results in the higher dimensional normal linear regression model for standard Bayes, HrB-posterior, ACP and the Gaussian posterior correction (PostCorr). {\color{black} The true value of the coefficients is $\beta=(1,1,1,0,\dots,0)^\top$.} {\color{black} See Table \ref{tab:reg} for a detailed description of the table entries.} 
    }
	\label{tab:reg_d20_all} 
\end{table}

\subsection{Poisson Regression}\label{sec:poiss}
For $i=1,\dots,n$, we observe count response data $y_i\in\mathbb{N}$, and covariates  $x_i$, a $d_\theta\times1$-dimensional vector with $x_{i,1}=1$, and our goal is inference on the unknown regression parameter $\theta$ in the generalized linear model (GLM):
\begin{equation*}
	\E(y_i\mid x_i)=\mu_i=g^{-1}(x_i^\top\theta),\quad \text{var}(y_i\mid x_i)=V(\mu_i;\psi),
\end{equation*}where $g(\cdot)$ is a strictly monotone and differentiable link function, and $V(\cdot)$ is a positive and continuous variance function  with dispersion parameter $\psi$. We again use a flat prior for $\theta$, and we treat $\psi$ as a hyperparameter. A common choice for observed counts is the Poisson distribution with link function $g^{}(\cdot)=\exp(\cdot)$.

A useful alternative to conducting standard Bayesian inference in this setting, is to instead consider generalized Bayes using the quasi-likelihood of \cite{wedderburn1974quasi} for the Poisson model with variance function $V(\mu_i;\psi)=\psi\mu_i$. Such an approach is equivalent to implementing \eqref{eq:gibbs} under the loss function:
 $
\mathcal{D}_n(\theta)=-\psi^{-1}\sum_{i=1}^{n}\left\{y_i\log(\mu_i)-\mu_i\right\},
$ with $\psi>0$ attempting to account for over-dispersion. Applying this choice within $\mathcal{D}_n(\theta)$, and taking $\omega=1$, produces the Gibbs posterior
\begin{equation*}\label{eq:gibbspois}
	\pi(\theta\mid \mathcal{D}_n,\psi)\propto \pi(\theta)\exp\left[\psi^{-1}\sum_{i=1}^{n}\left\{y_ix_i^\top\theta-\exp(x_i^\top\theta)\right\}\right].	
\end{equation*}
This approach was suggested in several studies, such as \cite{ventura2016pseudo}, and was shown by \cite{agnoletto2023bayesian} to produce asymptotically correct calibration levels when the true variance function takes the form $ \text{var}(y_i\mid x_i)=V(\mu_i;\psi)$. Following \cite{agnoletto2023bayesian}, estimation of $\psi$ is obtained by first fitting an overdispersed Poisson GLM to obtain the point estimator $\hat\psi$, with MCMC then used to sample $\pi(\theta\mid \mathcal{D}_n,\hat\psi)$.

While overdispersion is a common argument for using classes of count distributions other than the Poisson, specifying the dispersion functions is usually difficult, and requires joint inference on $\theta$ and $\psi$, even though only inference on $\theta$ is the goal. In this section, we compare standard Bayesian inference using the Poisson model, the approach of \cite{agnoletto2023bayesian} based on $\pi(\theta\mid \mathcal{D}_n,\hat\psi)$, and the ACP based on the assumed Poisson approximating model,  which is equivalent to fixing $V(\mu_i;\psi)=\mu_i$. {\color{black}{Similarly to the linear regression model, we take $W_n(\theta)$ as the sample covariance of the score equations.}} We follow the simulation design of \cite{agnoletto2023bayesian} and compare the case of  overdispersed counts. The results show that the ACP and the approach of \cite{agnoletto2023bayesian} are very similar, but the ACP does not need to model the dispersion function $V(\mu_i;\psi)$ or estimate the hyperparameter $\psi$. 

Following the simulation design of \cite{agnoletto2023bayesian}, the data is generated as follows: for $\psi_0$ some true dispersion parameter (fixed at $\psi_0=1.5$ in our experiments), first generate $\tilde{y}_i\stackrel{iid}{\sim}\Gamma(\mu_i/\psi_0,1/\psi_0)$, where $\Gamma(\alpha,\beta)$ denotes the Gamma distribution with shape $\alpha$ and rate $\beta$, and define the observed $y_i$ as the integer floor of $\tilde{y}_i$. This process yields counts that are overdispersed relative to the Poisson distribution. We again take $x_i$ and $\theta$ as $d_\theta\times1$-dimensional vectors, where $x_{i,1}=1$ for all $i=1,\dots,n$, and $(x_{2,i},\ldots,x_{d_\theta,i})^\top$ is generated as a multivariate independent Gaussian. We consider two simulation designs corresponding to  $d_\theta=10$ and $d_\theta=20$: when $d_\theta=10$ we set $\theta=(3.5,0.5,-0.5,0.5,0,\dots,0)^\top$, and for $d_\theta=20$ the parameter $\theta$ has the same structure but the remaining ten entries in $\theta$ are zero.

{\color{black}For $n=1000$ we generate 100 replicated datasets from this DGP. Posterior sampling was carried out using STAN along with the default choices for the NUTS algorithm. Table \ref{tab:possreg} presents the posterior mean, posterior variance and the marginal coverage of each method and displays the results for the first three, and tenth covariates. The results suggest  that the ACP and generalized Bayes  (GB) approach of \cite{agnoletto2023bayesian} behave similarly with coverages that are close to the nominal level.} In contrast, standard Bayesian inference based on the Poisson model is overly precise, and has  poorer coverage. Critically, unlike the approach of \cite{agnoletto2023bayesian}, the ACP obtains reliable coverage without having to model the dispersion or estimating the over-dispersion parameter $\psi$. {\color{black}In comparison with these methods, we see that the Gaussian posterior correction approach (PostCorr) performs well with $d_\theta=10$, but the resulting approximation is less reliable when $d_\theta=20$; please see Appendix \ref{app:sup_results_poiss} for further evidence. Even when $d_\theta=20$, the Gaussian approximation on which PostCorr is based becomes inaccurate, and results in variances that are too small, which ultimately delivers poorer coverage. This behavior is not observed in the ACP and GB, which both deliver reliable coverage for $d_\theta=10$ and $d_\theta=20$. }

	\begin{table}[H]
	\hspace*{-1.25cm} 			
	\centering
	{\footnotesize
		\begin{tabular}{rrrrrrrrrrrrr}
			\hline\hline
			 & \multicolumn{3}{c}{ACP} & \multicolumn{3}{c}{SB} & \multicolumn{3}{c}{GB} & \multicolumn{3}{c}{PostCorr} 
		 \\	 \hline\hline
			\textbf{$d_\theta=10$}& Bias & Var &Cover& Bias & Var &Cover&Bias & Var &Cover&Bias & Var &Cover\\ \hline
$\theta_1$          &       $-0$.014    &           0.055   &  0.92    &            $-0$.009    &            0.037 &     0.88       &     $-0$.010  &          0.055  &    0.93&     $-0$.0009  &          0.0002  &    0.92\\
$\theta_2$        &          0.014     &          0.032   &  0.96      &           0.013    &            0.022   &   0.97      &       0.013      &      0.032    &  0.97&     $-0$.0003  &          0.0002  &    0.94\\
$\theta_3$       &         $-0$.005     &          0.032 &    0.95       &         $-0$.004 &               0.021  &    0.88  &          $-0$.004   &         0.032 &     0.95&     0.0002  &          0.0002  &    0.93\\
			$\theta_{10}$    &             0.030   &            0.032&     0.95      &           0.037    &            0.021   &   0.86     &        0.028        &    0.032  &    0.94& 0.0001&  0.0056& 0.91\\\hline\hline
			\textbf{$d_\theta=20$}& Bias & Var &Cover& Bias & Var &Cover & Bias & Var &Cover&Bias & Var &Cover\\ \hline
		$\theta_1$      &          $-0$.0157           &    0.056&     0.95       &         $-0$.0124  &              0.037  &    0.90  &          $-0$.0112         &   0.056     &   0.97& $-0$.0014& 0.0076& 0.87 \\
$\theta_{2}$     &           $-0$.0050         &      0.033  &   0.94     &           $-0$.0048    &            0.022 &     0.88    &        $-0$.0048      &      0.034   &   0.95& 0.0004 & 0.0060& 0.94    \\
$\theta_{3}$      &           0.0020        &       0.033&     0.96   &              0.0025        &        0.022   &   0.93        &     0.0025      &      0.033 &     0.96&  0.0003& 0.0060&  0.93   \\
			$\theta_{10}$       &        $-0$.0060          &     0.033&     0.96   &             $-0$.0061  &              0.022   &   0.92 &           $-0$.0059           & 0.033 &     0.95&  $-0$.0003&    0.0060&   0.89  \\
\hline
	\end{tabular}}
	\captionsetup{width=1.1\linewidth}
	\caption{Results in the Poisson regression model for standard Bayes (SB), ACP, {\color{black}the generalized Bayes approach of \cite{agnoletto2023bayesian} (GB) based on the quasi-likelihood,  and the Gaussian posterior correction (PostCorr). The true value of the displayed parameter is $(\theta_1,\theta_2,\theta_3)=(3.5,0.5,-0.5,0,\dots,0)$, where for both $d_\theta=10$ and $d_\theta=20$, only the first three coefficients are non-zero. See Table \ref{tab:reg} for a detailed description of the table entries.}}
	\label{tab:possreg}
\end{table}

\subsection{Doubly Intractable Models}
In many interesting settings the likelihood for the model $P_\theta^{(n)}$ is of the form $p_\theta^{(n)}(y)=Z_\theta^{-1}\tilde{p}(y\mid\theta)$, where $\tilde{p}(y\mid\theta)$ is an analytically tractable density kernel, and $Z_\theta$ is an intractable normalizing constant. Bayesian inference in such settings, often called doubly intractable models since both $Z_\theta$ and the marginal likelihood are intractable, is challenging and often requires approximating $p_\theta^{(n)}(y)$. Examples include spatial models, exponential random graphs, and certain discrete data settings (\citealp{matsubara2022robust}). 

 \cite{matsubara2022robust, matsubara2023generalised} demonstrate how generalized Bayesian methods based on certain classes of discrepancies can be used to deliver posterior inferences in such settings. The key insight of these approaches is that by replacing the likelihood with a well-chosen discrepancy function, computing the intractable normalizing constant can be avoided. For continuous variables, \cite{matsubara2022robust} produce Bayesian inference using the kernel Stein discrepancy (KSD-Bayes); for discrete data, \cite{matsubara2023generalised} use the discrete Fisher divergence (DFD) to produce generalized Bayes posteriors (DFD-Bayes). 

Regardless of which divergence is used to produce the generalized Bayes posterior for $\theta$, the resulting posterior is not calibrated in general (\citealp{miller2021asymptotic}). To circumvent this issue, \cite{matsubara2023generalised} propose a computationally onerous bootstrapping procedure to deliver calibrated inferences, while \cite{matsubara2022robust} propose an approximate calibration procedure based on a particular choice for the posterior learning rate. 

Theorem \ref{thm:two} demonstrates that so long as $p_\theta^{(n)}$ is continuously differentiable, the ACP based on the KSD (or DFD) delivers calibrated inferences without bootstrapping, or the need to choose the learning rate.

\subsubsection{DFD-Bayes: Conway-Maxwell-Poisson Model}

This section performs approximate Bayesian inference for the Conway-Maxwell-Poisson model \citep{conway1962queuing}, which is a flexible model for discrete data $x$, with $x\ge0$, that can capture both under- and over-dispersion.  The probability mass function for a single observation $x$ conditional on parameter $\theta = (\theta_1, \theta_2)^\top$ is
\begin{align*}
	p(y\mid\theta) &= \frac{\tilde{p}(y|\theta)}{Z_\theta},
\end{align*}
where $\tilde{p}(y|\theta) = (\theta_1)^y(y!)^{-\theta_2}$, with $\theta_1 > 0$ and $\theta_2 \in [0,1]$.  The normalizing constant $Z_\theta = \sum_{y=0}^\infty \tilde{p}(y|\theta)$ does not have a closed form expression, except for special cases; for example, when $\theta_2=1$ we recover the Poisson distribution with mean $\theta_1$.  However, the normalizing constant $Z_\theta$ can be accurately approximated and facilitates a comparison with a precise approximation of the true posterior.

To avoid the intractable normalizing constant associated with discrete distributions, \citet{matsubara2023generalised} {\color{black}conduct generalized Bayesian inference on $\theta$ using the DFD as the loss function in \eqref{eq:gibbs}.}  In the one-dimensional data setting the DFD between the statistical model conditioned on $\theta$, $P_\theta$, {\color{black}and the empirical distribution of the data, given by $P_n$,} is defined as
\begin{align*}
	\mbox{DFD}(P_\theta||P_n) &=  \frac{1}{n} \sum_{i=1}^n \left( \frac{p(y_i^-|\theta)}{p(y_i|\theta)}  \right)^2 - 2 \left( \frac{p(y_i|\theta)}{p(y_i^+|\theta)}  \right), 
\end{align*}
and it is evident that the intractable normalizing constants cancel in the ratios.  For this example, $y^+ = y+1$ and $y^- = y-1$, unless $y=0$ in which case we set $y^{-} = \mathrm{max} \{y_i\}_{i=1}^n$, i.e.\ the maximum value of the dataset.

\citet{matsubara2023generalised} embed the DFD within a generalized Bayes framework to conduct approximate Bayesian inferences without needing $Z_\theta$.  To calibrate the scaling parameter $\omega$ in the generalized posterior, \citet{matsubara2023generalised} propose the following steps: first, create $B$ bootstrap replications of the observed data, and for each of the bootstrapped dataset, produce an estimator of $\theta$, say $\{\theta_n^{(b)} : b=1,\dots,B \}$, by minimizing the DFD between the assumed model, and the $b$-th bootstrapped dataset; the value of $\omega$ is then chosen to minimize the Fisher divergence between the generalized posterior and the empirical bootstrap sample $\{\theta_n^{(b)} : b=1,\dots,B \}$.  This value of $\omega$ can then be used in an MCMC algorithm to generate samples from the posterior.

In contrast, the ACP approach avoids a calibration process, and will deliver calibrated Bayesian inference.  All that is required is to compute the score of each component of the DFD, which in this case can be done analytically.  
The implementation can be accelerated by computing the score for each unique value of the dataset, and weighting by the number of replicates of each unique value. Again, the \texttt{AdaptiveMCMC.jl} library in \texttt{Julia} is used to sample the ACP.  We do not consider the posterior correction method of \citet{muller2013risk} in this example as the model is correctly specified. 

From this model we generate 100 independent datasets of size $n=2000$ using the true parameter value $\theta = (4,0.75)^\top$. Table \ref{tab:CMP_results} shows the results for the ACP, which is compared with an accurate approximation of the true posterior and the GB approach of \citet{matsubara2023generalised}. {\color{black}Although not formally proved in \citet{matsubara2023generalised}, we conjecture that if the model is correctly specified, then the limit minimizer of the DFD as $n \rightarrow \infty$ - the point onto which the ACP will concentrate - is equal to the true parameter value. Using a large dataset of 10 million observations simulated from the model, we find that the resulting minimizer of the DFD coincides with the true parameter up to two decimal places, and so we use $\theta = (4,0.75)^\top$ to calculate posterior coverage.} The ACP produces posteriors with a larger standard deviation than the true posterior, but still achieves reasonable coverage rates. The ACP results are slightly less accurate than the GB approach of \citet{matsubara2023generalised}; however, the approach of \citet{matsubara2023generalised} focuses on intractable discrete models whereas the ACP is generally applicable and avoids bootstrapping the distribution of the point estimator that minimizes the loss.  

Figure \ref{fig:cmp_posteriors} displays posterior approximations for a single dataset: the ACP approach produces an approximation that is inflated relative to the true posterior, and is similar to the GB approach of \citet{matsubara2023generalised}.  We find that for some of the datasets, the ACP has a substantially heavier tail (see Figure \ref{fig:cmp_posterior_sds}) than the GB approach, which for the parameter $\theta_1$ seems to lead to some  inflation of the MSE and bias shown in Table \ref{tab:CMP_results}.

\begin{table}[H]
	\hspace*{-1.25cm} 			
	\centering
	{\footnotesize
		\begin{tabular}{rrrrrrrrrr}
			\hline\hline
			& \multicolumn{3}{c}{ACP} & \multicolumn{3}{c}{SB} & \multicolumn{3}{c}{GB}\\\hline\hline
			 & Bias & Var &Cover& Bias & Var &Cover&Bias & Var &Cover\\ \hline
$\theta_1$ & 0.268  & 0.3840   & 0.96   & 0.076 & 0.0580  & 0.96  & 0.137  & 0.161 & 0.98  \\
$\theta_2$ & 0.040   & 0.0064   & 0.95  & 0.012  &  0.0021  &  0.97  & 0.019 & 0.0038 & 0.96  \\ \hline \hline
	\end{tabular}}
	\captionsetup{width=1.1\linewidth}
	\caption{Results for the Conway-Maxwell-Poisson example using 100 independent datasets simulated with true parameter value $\theta = (\theta_1, \theta_2)^\top = (4,0.75)^\top$. {\color{black}Here, GB refers to the approach of \cite{matsubara2023generalised}, and SB to standard Bayes. See Table \ref{tab:reg} for a detailed description of the table entries.} 
    }
	\label{tab:CMP_results} 
\end{table}
\begin{figure}[!htp]
	\centering
	\begin{subfigure}[t]{0.45\textwidth}
		\centering
		\includegraphics[width=\linewidth,height=50mm]{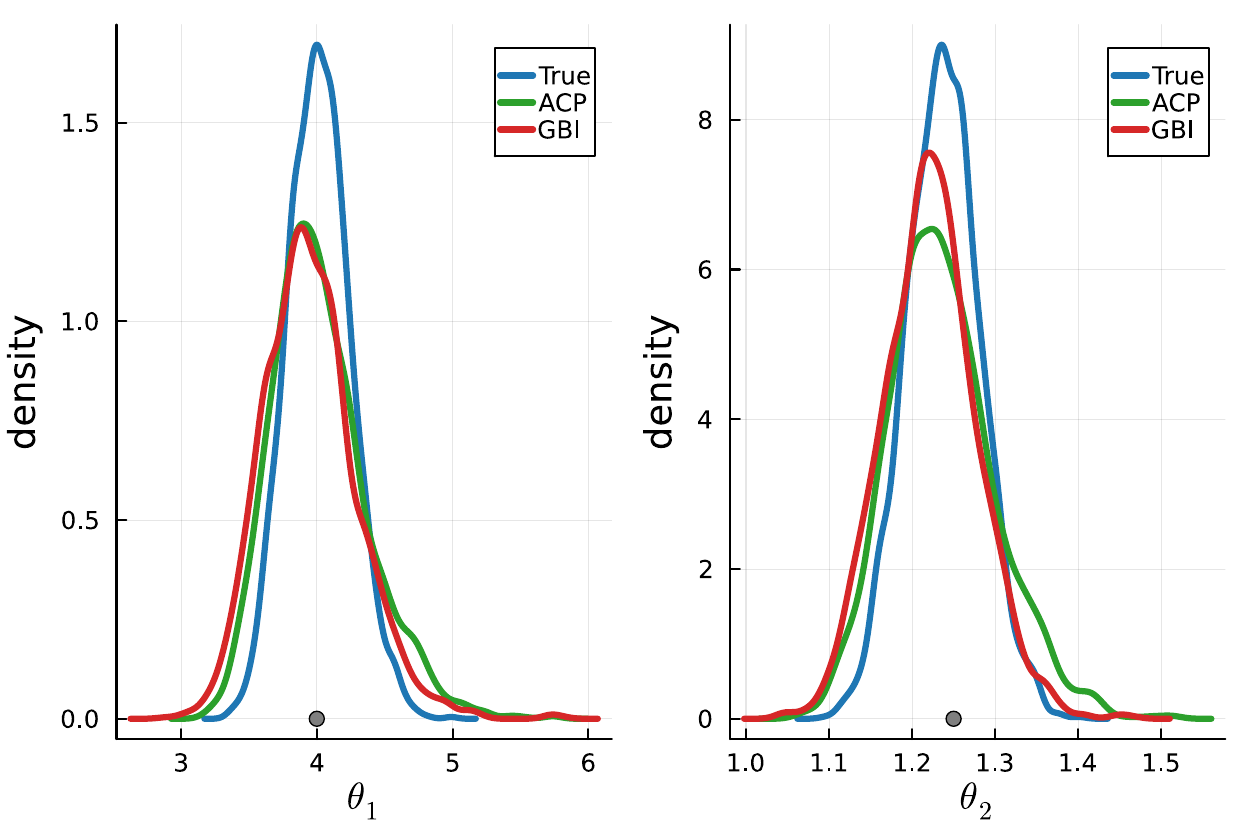}
		\caption{Univariate density estimates of approximations to the posterior distribution for a single dataset generated from a CMP model with true parameter value $\theta=(4,0.75)^\top$.}
		\label{fig:cmp_posteriors}
	\end{subfigure}
	\hfill
	\begin{subfigure}[t]{0.45\textwidth}
		\centering
		\includegraphics[width=\linewidth]{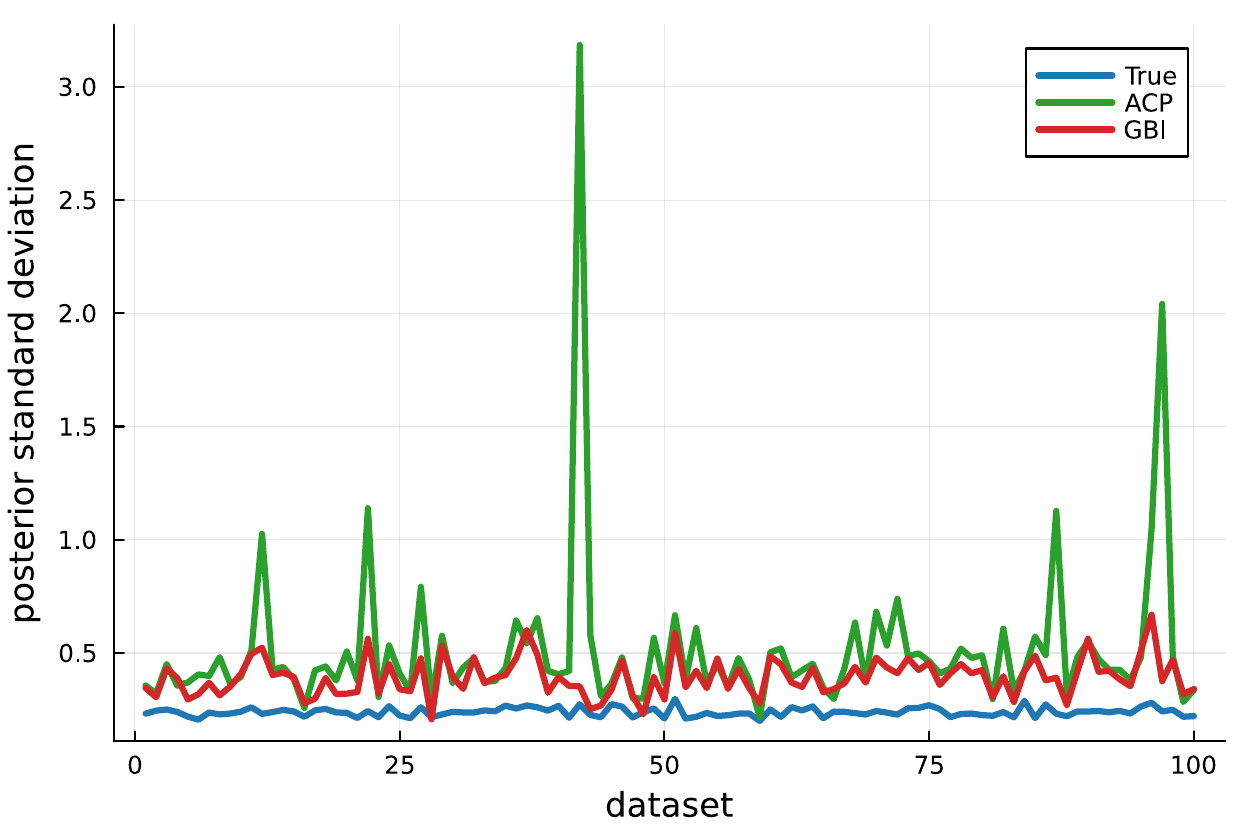}
		\caption{Univariate posterior standard deviations of $\theta_1$ across 100 datasets generated from a CMP model with true parameter value $\theta=(4,0.75)^\top$.}
		\label{fig:cmp_posterior_sds}
	\end{subfigure}
	
	\captionsetup{width=1.0\linewidth}
	\caption{Posterior summaries for data generated from a CMP model.  Results are shown for ACP (green), the GB approach of \citet{matsubara2022robust} (red), and an accurate approximation of the exact likelihood (blue).}
	\label{fig:cmp_side_by_side}
\end{figure}

%
%

\subsubsection{KSD-Bayes: Contaminated Normal}\label{sec:ksd}
While discrete Fisher divergence-based Bayes (DFD-Bayes) can be used to conduct posterior inference in models where the normalizing constant for the mass function is intractable, \cite{matsubara2022robust} suggest using the kernel Stein discrepancy (KSD) within a generalized Bayes framework when the variables observed variables are continuous. In particular, to avoid the intractable normalizing constant $Z_\theta$ in $p^{(n)}_\theta(y)=p(y\mid\theta)/Z_\theta$, \cite{matsubara2022robust} conduct generalized Bayesian inference using the KSD, which is defined as follows: for some positive-definite kernel function $K:\mathcal{Y}\times\mathcal{Y}\rightarrow\mathbb{R}^{d\times d}$, $P,Q\in\mathcal{P}(\mathcal{Y})$, where $\mathcal{P}(\mathcal{Y})$ denotes the space of distributions over $\mathcal{Y}$, and $\mathcal{S}_{Q}$ the Stein operator (see \citealp{matsubara2022robust} for a formal definition), the KSD takes the form
$
\operatorname{KSD}^2(Q \| P):=\mathbb{E}_{Y, Y^{\prime} \sim P}\left[\mathcal{S}_{Q} \mathcal{S}_{Q} K\left(Y, Y^{\prime}\right)\right],
$ 
which can be explicitly estimated using 
$$
\mathrm{KSD}^2 ({P}^{}_\theta \| {P}_n)=\frac{1}{n^2} \sum_{i=1}^n \sum_{j=1}^n \mathcal{S}_{P_\theta} \mathcal{S}_{P_\theta} K\left(y_i, y_j\right),
$$where ${P}_n$ is the empirical measure of the sample $y_1,\dots,y_n$. \cite{matsubara2022robust} show that $\mathrm{KSD}^2\left(P_\theta \| {P}_n\right)$ can be evaluated in closed form when $\mathcal{Y}=\mathbb{R}^d$ and when $\mathcal{S}_{\mathcal{Q}}$ is the Langevin-Stein operator; see \citealp{matsubara2022robust} for complete details, and, in particular, the discussion surrounding their equation (5).

\cite{matsubara2022robust} embed the KSD within a generalized Bayesian framework, and demonstrate that in exponential family models with conjugate priors, the KSD-Bayes posterior has a closed form. The main motivation behind the use of KSD-Bayes comes from its robustness properties, which allows the authors to show that the KSD-Bayes posterior displays a global bias-robustness property under data contamination.

While KSD-Bayes displays meaningful robustness, the resulting posterior is not calibrated. The authors suggest a two-stage approach for setting the learning rate to overcome this issue; however, the resulting KSD-Bayes posterior will not be calibrated unless $\theta$ is a scalar. In contrast, the ACP based on $\mathrm{KSD}^2 ({P}^{(n)}_\theta \| {P}_n)$ inherits the bias robustness properties of the KSD-Bayes posterior, and also delivers calibrated inference. 

{\color{black}To demonstrate this, we compare the ACP based on the KSD, KSD-Bayes and standard Bayes in the contaminated normal location model analyzed in Section 6.1 of \cite{matsubara2022robust}.} For the ACP and KSD-Bayes, we set the kernel function $K(\cdot,\cdot)$  as the inverse multi-quadratic kernel suggested as the default choice by \cite{matsubara2022robust}, and for KSD-Bayes we use their optimal choice of learning rate in an attempt to correct the coverage of the KSD-Bayes posterior; see Section 5 of \cite{matsubara2022robust} for details on these choices. 

The assumed model is $P_\theta=\mathcal{N}(\theta,1)$, with prior belief $\theta\sim \mathcal{N}(0,1)$, but the actual generating process is an $\epsilon$-contamination of $P_\theta$ with contaminating distribution $\mathcal{N}(5,1)$. We take the true value of $\theta=-1$, so that the true DGP is $(1-\epsilon) \mathcal{N}(-1,1) + \epsilon\mathcal{N}(5,3)$. We generate 100 data points from this model and 100 replicated datasets under $\epsilon=0$ (correct specification), and $\epsilon=0.10$ (model misspecification). {\color{black}For each dataset, we calculate the bias and average posterior variance and Monte Carlo coverage of the 95\% credible set.} Table \ref{tab:KSD-bayes}  demonstrates that the ACP has a similar location as the KSD-Bayes posterior, {\color{black}but its credible sets are closer to the nominal level under contamination than the KSD-Bayes posterior, even though the KSD-Bayes posterior utilizes the two-stage learning rate described in \cite{matsubara2022robust}.} Similar findings hold under both designs.

\begin{table}[H]
	\hspace*{-1.25cm} 			
	\centering
	{\footnotesize
		\begin{tabular}{rrrrrrrrrr}
			\hline\hline
			& \multicolumn{3}{c}{ACP} & \multicolumn{3}{c}{SB} & \multicolumn{3}{c}{KSD-B}\\\hline\hline
			 & Bias & Var &Cover& Bias & Var &Cover&Bias & Var &Cover\\ \hline
$\epsilon=0,\;\;\theta$ & $-$0.011 & 0.021   & 0.98   & 0.010 & 0.010  & 0.95  & 0.003  & 0.014 & 0.96  \\
$\epsilon=0.10,\;\;\theta$ & 0.087   & 0.021   & 0.94  & 0.614  &  0.010  &  0.02  & 0.100 & 0.014 & 0.83  \\ \hline \hline
	\end{tabular}}
	\captionsetup{width=1.1\linewidth}
	\caption{Results for the contaminated normal location for the ACP, KSD-Bayes (KSD-B) and standard Bayes (SB) with true parameter value $\theta = 1$. The case of $\epsilon=0$ corresponds to correct specification, while $\epsilon=0.10$ corresponds to 10\% data contamination. Table \ref{tab:reg} describes the column entries. }
	\label{tab:KSD-bayes}
\end{table}
\section{Theoretical Results}\label{sec:theory}
\subsection{Assumptions}
Using standard regularity conditions commonly encountered in frequentist inference, we formally demonstrate that  $\pi(\theta\mid \mathcal{Q}_n)$ delivers calibrated credible sets. We first fix notation. For a positive sequence $a_n\rightarrow\infty$ as $n\rightarrow\infty$, $X_n=o_p(a_n^{-1})$ denotes convergence of the sequence $a_n X_n$ to zero in probability; $X_n=O_p(a_n^{-1})$ denotes that $a_nX_n$ is bounded in probability. For a set $A\subseteq\mathbb{R}^d$, $\mathrm{Int}(A)$ denotes the interior of $A$. The notation $\Rightarrow$ denotes weak convergence under $P^{(n)}_0$.

\begin{assumption}\label{ass:infeasible}There exist $m:\Theta\rightarrow\mathbb{R}^{d_\theta}$ satisfying the following conditions. 
	\begin{enumerate}
		\item[(i)] $\sup_{\theta\in\Theta}\|\overline{m}_n(\theta)-m(\theta)\|=o_p(1)$. 
		\item[(ii)] There exists some $\theta_\star\in\mathrm{Int}(\Theta)$, such that $m(\theta_\star)=0$.
		\item[(iii)] For some $\delta>0$, $m(\theta)$ is continuously differentiable over $\|\theta-\theta_\star\|\le\delta$, and $\mathcal{H}(\theta_\star):=-\nabla_{\theta}m(\theta_\star)$ is invertible.
		\item[(iv)] There is a positive-definite matrix $\mathcal{I}(\theta_\star)$, such that $\sqrt{n}\overline{m}_n(\theta_\star)\Rightarrow N\{0,\mathcal{I}(\theta_\star)\}$.
		\item[(v)] For any $\epsilon>0$, there exists $\delta>0$ such that $$\limsup_{n\rightarrow\infty}P^{(n)}_0\left(\sup _{\left\|\theta-\theta_\star\right\|<\delta}\frac{\sqrt{n}\left\|\{\overline{m}_{n}(\theta)-\overline{m}_{n}(\theta_\star)\}-\{m(\theta)-m(\theta_\star)\}\right\|}{1+\sqrt{n}\|\theta-\theta_\star\|}>\epsilon\right)<\epsilon.$$
	\end{enumerate}
	
\end{assumption}

\begin{assumption}\label{ass:weight}The following conditions
	are satisfied: {\color{black}(i) for $n$ large enough, with $P^{(n)}_0$-probability one, the matrix ${W}_n(\theta)$ is positive semi-definite and symmetric;} (ii) there exist a matrix $W(\theta)$ that is positive semi-definite and symmetric for all $\theta\in\Theta$, such that $\sup_{\theta\in\Theta}\|{W}^{}_n(\theta)-W(\theta)\|=o_{p}(1)$ and $W(\theta)$ is continuous and positive-definite for all $\|\theta-\theta_\star\|\le\delta$ and some $\delta>0$; (iii) {\color{black}for any $\epsilon>0$, and some $\theta_\star\in\Theta$,} $\sup_{\|\theta-\theta_\star\|\ge\epsilon}m(\theta)^\top W(\theta)^{-1}m(\theta)>0$.
\end{assumption}
\begin{remark}{\normalfont 
		Assumptions \ref{ass:infeasible} and \ref{ass:weight} together enforce smoothness and identification conditions on  the loss used within the ACP. These conditions permit the existence of a quadratic expansion for the loss that is smooth in $\theta$ near $\theta_\star$, and with a remainder term that can be suitably controlled. Assumption \ref{ass:infeasible}(i) is a uniform law of large numbers and is satisfied for a large class of functions under many different data generating processes, and can often be verified using empirical processes methods (see, e.g., \citealp{vaart2023empirical}, Chapter 3). {\color{black}Assumption \ref{ass:infeasible}(iii) requires that the limit counterpart $m(\theta)$ is differentiable near $\theta_\star$.  Assumption \ref{ass:infeasible}(v) is a smoothness condition that only requires $\overline{m}_n(\theta)$ to be smooth ``in expectation'' and also requires that the deviations between $\overline{m}_n(\theta)$ and $m(\theta)$ can be appropriately controlled. This smoothness condition allows for loss functions that are only weakly differentiable, e.g., loss functions based on absolute values or check functions, which do not satisfy a standard differentiability condition, see, e.g., Chapter 3 of  \cite{vaart2023empirical} for a similar  assumption.} We also note that  Assumption \ref{ass:infeasible} does not require that $\theta_\star$ is unique. This is important as we will later treat the case of multi-modal posteriors caused by score equations that admit multiple roots. } 
\end{remark}

\begin{remark}{\normalfont 
Since the behavior of the ACP depends on the matrix $W_n(\theta)$, we require certain regularity conditions on $W_n(\theta)$. Assumption \ref{ass:weight}(ii) requires that $W_n(\theta)$ converges uniformly to $W(\theta)$, which is continuous and positive-definite for all $\theta$ sufficiently close to $\theta_\star$. Assumption \ref{ass:weight}(ii) does not require $W(\theta)$ to be invertible uniformly over $\Theta$, but only sufficiently close to $\theta_\star$. {\color{black}In such cases the limiting quadratic form $Q(\theta)=2^{-1}m(\theta)^\top W(\theta)^{-1}m(\theta)$ need not be continuous over $\Theta$, and there is no guarantee that a minimum exists. If this is the case, an additional identification condition given by Assumption \ref{ass:weight}(iii) is required; this condition is equivalent to the standard condition that the population criterion $Q(\theta)$ admits a (possibly non-unique) minimizer. This condition is not onerous, and ensures the existence of a point onto which the posterior will concentrate. Such a condition is well known in related fields such as the generalized method of moments (see, e.g., Assumption 3 in \citealp{hall2003large}). If $W(\theta)^{-1}$ is positive-definite and bounded over $\Theta$, Assumption \ref{ass:weight}(iii) is automatically satisfied.
}}
\end{remark}

The next assumption requires the  existence of certain prior moments. 
\begin{assumption}\label{ass:prior}
	For $\theta_\star$ as defined in Assumption \ref{ass:infeasible}, $\pi(\theta_\star)>0$ and $\pi(\theta)$ is continuous on $\Theta$. For some $p\ge1$, $\int_\Theta\|\theta\|^p\pi(\theta)\dt\theta<+\infty$, and for all $n$ large enough, {\color{black} with $P^{(n)}_0$-probability one}, $\int_\Theta |W_n(\theta)|^{-1/2}\|\theta-\theta_\star\|^{p}\pi\left(\theta\right)\dt\theta<+\infty$.
\end{assumption}
\begin{remark}
	The condition $\int_\Theta |W_n(\theta)|^{-1/2}\|\theta-\theta_\star\|^{p}\pi\left(\theta\right)\dt\theta<+\infty$, {\color{black} with $P^{(n)}_0$-probability one}, is not commonly encountered and is required since the matrix $W_n(\theta)$ may be singular far away from $\theta_\star$. This condition {\color{black}allows $W_n(\theta)$ to have points of singularity, but restricts them to have negligible mass under $\pi(\theta)$ so that the expectation of the $\|W_n(\theta)|^{-1/2}\|\theta-\theta_\star\|^p$ under $\pi(\theta)$ exists. For example, if $W_n(\theta)$ is singular at $\theta'$, then so long as $\theta'$ is far enough out in the tails of $\pi(\theta)$, i.e., $\pi(\theta')\approx0$, then the expectation will exist. For a given dataset this condition can be verified using Monte Carlo simulation from the prior.}
	
\end{remark}

\subsection{Results}
\subsubsection{Unique identification}
If $\theta_\star$ in Assumption \ref{ass:infeasible} is unique, then the ACP correctly quantifies uncertainty. {\color{black}To formally state this result, recall that $\theta_n$ satisfies $0=\overline{m}_n(\theta_n)$, define  $t:=\sqrt{n}(\theta-\theta_n)$, let its support be $\mathcal{T}_n:=\{t:=\sqrt{n}(\theta-\theta_n):\theta\in\Theta\}$, and define $\Delta(\theta):=\mathcal{H}(\theta)W(\theta)^{-1}\mathcal{H}(\theta)^\top$.}

\begin{theorem}\label{thm:two}If Assumptions \ref{ass:infeasible}-\ref{ass:prior} hold and $\theta_\star$ is unique, then, for $p\in\{0,1\}$,
	$$
		\int_{\mathcal{T}_n}\|t\|^p|\pi(t\mid \mathcal{Q}_{n})-{\color{black}N\{t;0,[\omega\cdot\Delta(\theta_\star)]^{-1}\}}|\dt t=o_{p}(1).
	$$
\end{theorem}

\begin{corollary}\label{corr:one}Define $\overline\theta=\int_{\Theta}\theta \pi(\theta\mid \mathcal{Q}_n)\dt\theta$. If the assumptions of Theorem \ref{thm:two} hold, then    
	$$
	\sqrt{n}(\overline\theta-\theta_\star)\Rightarrow N\{0,\Delta(\theta_\star)^{-1}\mathcal{H}(\theta_\star) W(\theta_\star)^{-1}\mathcal{I}(\theta_\star) W(\theta_\star)^{-1}\mathcal{H}(\theta_\star)^\top\Delta(\theta_\star)^{-1}\}.
	$$
\end{corollary} 

\begin{remark}
{\color{black}Theorem \ref{thm:two} and Corollary \ref{corr:one} together imply that credible sets obtained from $\pi(\theta\mid \mathcal{Q}_n)$ are calibrated so long as the learning rate satisfies $\omega\rightarrow1$ and $ W_n(\theta_\star)\xrightarrow{p}\mathcal{I}(\theta_\star)=\lim_{n\rightarrow+\infty}\text{Cov}\left\{\sqrt{n}\overline{m}_n(\theta_\star)\right\}.
$ To see this, note that when $\omega=1$ and $W(\theta_\star)=\mathcal{I}(\theta_\star)$, the asymptotic posterior variance in Theorem \ref{thm:two} is $\Delta(\theta_\star)^{-1}=[\mathcal{H}(\theta_\star)\mathcal{I}(\theta_\star)^{-1}\mathcal{H}(\theta_\star)^\top]^{-1}$, and the asymptotic variance of the posterior mean in Corollary \ref{corr:one} reduces to 
\begin{flalign*}
\Delta(\theta_\star)^{-1}\mathcal{H}(\theta_\star)W(\theta_\star)^{-1}\mathcal{I}(\theta_\star) W(\theta_\star)^{-1}\mathcal{H}(\theta_\star)^\top\Delta(\theta_\star)^{-1}&=\Delta(\theta_\star)^{-1}\mathcal{H}(\theta_\star)\mathcal{I}(\theta_\star)^{-1}\mathcal{H}(\theta_\star)^\top\Delta(\theta_\star)^{-1}\\&=\Delta(\theta_\star)^{-1},
\end{flalign*}so that the two coincide. As discussed in Remark \ref{rem:variance_estimator}, we can often ensure that $W(\theta_\star)=\mathcal{I}(\theta_\star)$ by taking $W_n(\theta)$ as the sample variance of $\overline{m}_n(\theta)$.}%
\end{remark}
The uniqueness of the solution $\theta_\star$ to the population estimating equations $0=m(\theta)$ is necessary to ensure that the posterior is asymptotically Gaussian; it is standard in the theoretical analysis of Bayesian methods, but does not always hold.

\subsubsection{Non-unique Identification}\label{sec:nonunique}
{\color{black}Since we are operating with arbitrary losses and 
possibly misspecified models, 
the equation $0=m(\theta)$ may admit multiple solutions; in this case a result similar to Theorem \ref{thm:two} should still be satisfied around each solution.}
%
%
To derive such a result, recall that $\Delta(\theta):=\mathcal{H}(\theta)W(\theta)^{-1}\mathcal{H}(\theta)^\top$,  let $\lambda_{\text{min}}\{\Delta(\theta)\}$ be the smallest eigenvalue of $\Delta(\theta)$, and define $\Theta_\star:=\{\theta\in\Theta:0=m(\theta)\}$. 



\begin{assumption}\label{ass:set}{\color{black}The set $\Theta_\star$ contains a finite number of distinct elements, and for some $c>0$, and each $\theta_\star\in\Theta_\star$, 
$0<c\le\lambda_{\text{min}}\{\Delta(\theta_\star)\}	\le1/c$. The elements in $\Theta_\star$ are well-separated, i.e., for any $\theta_{1,\star},\theta_{2,\star}\in\Theta_\star$, there exists an $\epsilon>0$ such that $\|\theta_{1,\star}-\theta_{2,\star}\|>\epsilon$.} 	
\end{assumption}

Similarly to the case where $\theta_\star$ is unique, we must choose a centering sequence to obtain the posterior shape. {\color{black}In general, a useful choice for this centering sequence would be the roots of $\overline{m}_n(\theta)$, i.e., $\theta_n\in \arg_{\theta\in\Theta}\{0=\overline{m}_n(\theta)\}$}. Assumption \ref{ass:infeasible} by itself does not imply that $\theta_n$ is consistent for any $\theta_\star\in\Theta_\star$, and without further assumptions $\theta_n$ cannot be used as a centering sequence to determine the posterior shape. 
While it is straightforward to show the existence and consistency {\color{black}of $\theta_n\in\arg_{\theta\in\Theta}\{0=\overline{m}_n(\theta)\}$} under additional smoothness conditions on $\overline{m}_n(\theta)$, to maintain generality we impose the following high-level condition. 
\begin{assumption}\label{ass:cons} {\color{black}For each $\theta_\star\in\Theta_\star$, there exists a sequence $\theta_n\in\arg_{\theta\in\Theta}\{0=\overline{m}_n(\theta)\}$ such that $\|\theta_n-\theta_\star\|=o_p(1)$. }	
\end{assumption}

{\color{black}To ascertain the behavior of $\pi(\theta\mid \mathcal{Q}_n)$ when $\Theta_\star$ contains multiple elements, we analyze the local asymptotic behavior of $\pi(\theta\mid\mathcal{Q}_n)$ around a given root sequence $\theta_{n}\in\{0=\arg_{\theta\in\Theta}\overline{m}_n(\theta)\}$. Redefine $t:=\sqrt{n}(\theta-\theta_n)$ as the local parameter around the root $\theta_n$, and let $\pi(t\mid \mathcal{Q}_n)=\pi(\theta_n+t/\sqrt{n}\mid \mathcal{Q}_n)/{n}^{d_\theta/2}$ denote the posterior for $t$  in a shrinking neighborhood of $\theta_n$}. 
\begin{theorem}\label{thm:multi}
	If Assumptions \ref{ass:infeasible}-\ref{ass:cons} are satisfied at $\theta_\star$, {\color{black}then for  $\gamma_n>0$ such that $\gamma_n/\sqrt{n}=o(1)$, for some $C_\star\in[0,1]$,
	$
\left|\int_{\|t\|\le\gamma_n}{\pi}(t\mid \mathcal{Q}_n)\dt t-C_\star\int_{\|t\|\le\gamma_n}N\{t;0,[\omega\cdot\Delta(\theta_\star)]^{-1}\} \dt t\right|=o_{p}(1)
	$.}
\end{theorem}

Theorem \ref{thm:multi} demonstrates that if the score equations $0=m(\theta)$ admit multiple roots, say $\theta_{1,\star}$ and $\theta_{2,\star}$, {\color{black}then in a shrinking neighborhood of the root $\theta_{1,\star}$ (resp., $\theta_{2,\star}$) the ACP resembles -- \textit{but is not equal to} -- a  Gaussian density around $\theta_{1,\star}$ (resp., $\theta_{2,\star}$) with variance $\Delta(\theta_{1,\star})^{-1}$ (resp., $\Delta(\theta_{2,\star})^{-1}$). Hence, Theorem \ref{thm:multi} implies that if there are multiple roots then the ACP will asymptotically converge to a Gaussian mixture, with the specific value of $C_\star$ in Theorem \ref{thm:multi} accounting for this mixture formulation.  
} 


{\color{black}A consequence of the mixed-Gaussian limit implied by Theorem \ref{thm:multi} is that ACP posterior credible sets can only be represented as a highest probability density region (HPDR), and therefore do not admit a closed-form representation even asymptotically. 
This difficulty is due to the fact that in multimodal settings the shape of the HPDR depends on: 1) the number of elements in $\Theta_\star$ and their proximity to the boundary of $\Theta$; 2) the value of the prior at these elements; and 3) the value of $\Delta(\theta)^{-1/2}$ at the elements. Thus, in line with the analysis and recommendations of \citet{gustafson2012behaviour,gustafson2014behaviour}, in the multi-modal case the exact frequentist coverage of generic ACP credible sets can only be established on a case-by-case basis. }

{\color{black}
While obtaining a direct calibration guarantee for a generic ACP credible set of $\Theta_\star$ similar to that implied by Theorem \ref{thm:two} is not feasible without more structure on the problem under analysis, this \textit{does not imply that} the ACP is not calibrated in multi-modal settings: Theorem \ref{thm:multi} only implies that ACP credible sets do not admit a closed-form representation, as was the case in Theorem \ref{thm:two}. That being said, it is possible to construct ACP credible regions which account for the multimodal nature of the loss function and are calibrated. We demonstrate this using one construction, but note here that other constructions may deliver similar results. Let $t_k=\sqrt{n}(\theta-\theta_{k,n})$ denote a local parameter around the root $\theta_{k,n}$  (an estimator of $\theta_{k,\star}\in\Theta_\star$), and let $r_n\ge0$,  $r_n\rightarrow\infty$, be such that $r_n/\sqrt{n}=o(1)$; define a (local) quantile around the $k$-th root via $\gamma_{k,n}(\alpha)=\inf\{\gamma>0:\frac{\Pi\left\{\|t_{k}\|\le \gamma \mid\mathcal{Q}_n\right\}}{\Pi\left\{\|t_{k}\|\le r_n \mid\mathcal{Q}_n\right\}}\ge 1-\alpha\}$, and define the credible region $\mathcal{C}^{(n)}_{\Pi}(1-\alpha):=\bigcup_{k=1}^{K}\{t_k:\|t_k\|\le \gamma_{k,n}(\alpha/K)\}$. 
\begin{theorem}\label{thm:new}
Assume that $\Theta_\star=\{\theta_{1,\star},\dots,\theta_{K,\star}\}$. If Assumptions  \ref{ass:infeasible}--\ref{ass:cons} are satisfied at each  $\theta_{k,\star}\in\Theta_\star$,   then 
$
\lim_{n}P^{(n)}_0\left\{\Theta_\star\subseteq \mathcal{C}^{(n)}_{\Pi}(1-\alpha)\right\}\ge1-\alpha.
$
\end{theorem}
Theorem \ref{thm:new} show that reliable inference for $\Theta_\star$ is achievable when the ACP credible region accounts for the multimodal nature of the loss function, which is helpful as it is infeasible to deduce the exact level of coverage without further structure or assumptions. In addition, we note that $\mathcal{C}^{(n)}_{\Pi}(1-\alpha)$ coincides with the standard credible set in the case of unique identification. 
}

\subsubsection{Monte Carlo example: non-unique identification}\label{example:non_unique}
{\color{black}This section illustrates a prominent example where a standard ACP credible set delivers calibrated inferences, while Bayes does not; an additional example with similar conclusions is given in Appendix \ref{example:non_unique_tukey}. We assume that the observed data $y$ is generated from the two component Gaussian mixture model (GMM):
$$
y_i\stackrel{iid}{\sim}\eta \cdot N(\mu_1,\sigma_1^2)+(1-\eta) N(\mu_2,\sigma_2^2),
$$
where the means, $\mu=(\mu_1,\mu_2)^\top$, standard deviations, $\sigma=(\sigma_1,\sigma_2)^\top$, and mixture weight, $\eta$, are all unknown. We consider Gaussian $N(0,1)$ priors for both means, inverse Gaussian priors for the standard deviations, and a uniform prior on the mixture weight.

In GMMs the  parameter $\theta=(\mu^\top,\sigma^{\top},\eta)^\top$ is not uniquely identified since the corresponding mixture components can be switched without altering the value of the likelihood. In this way, non-unique identification arises from label switching: two maxima exist that correspond to permutations of component labels. Fundamentally, this is not a problem for Bayesian inference nor for the ACP as both methods deliver posteriors that concentrate onto the set of identified parameters.

While concentration of both the ACP and Bayes posterior will be robust in this multi-modal setting, it is also known that mixture models are sensitive to model deviations (\citealp{Gray1994}, \citealp{10.3150/18-BEJ1087}). Therefore, in this section we gauge the ability of the Bayes posterior and the ACP to deliver calibrated inferences on $\theta$ when the model is well-specified and misspecified. In the well-specified case observed data is generated from the GMM with $\theta=(1,-1,1,0.8,0.5)^\top$; in the misspecified case data is generated from a mixture model with one Gaussian component, $N(1,1)$, and one Student-t component, with mean $-1$, scale $0.8$, and six degrees of freedom. In both designs, we generate an equal proportion from both mixture components and observe $n=250$ data points.

To assess the coverage of standard Bayes and the ACP in this setting, we generated 100 replications from each DGP and recorded the frequentist coverage of the methods' 95\% HPD regions for $\theta=(1,-1,1,0.8,0.5)^\top$ in both settings. Since the posteriors can be multi-modal, we only report the methods coverage in Table \ref{tab:cover-comp}. Under correct model specification, both posteriors exhibit similar coverage, and are conservative in this small sample case. In contrast, under model misspecification the Bayes posterior exhibits significant under-coverage, whereas the ACP again delivers conservative inferences.
}
\begin{table}[H]
	\hspace*{-1.25cm} 			
	\centering
	{\footnotesize
		\begin{tabular}{rrrrrrrrrrr}
			\hline\hline
			 \multicolumn{4}{c}{Correct} & \multicolumn{1}{c}{} & \multicolumn{4}{c}{Misspec.}\\\hline\hline
			 & $\sigma_1$ & $\mu_1 $&$\sigma_2$& $\mu_2$ & $\eta $&$\sigma_1$&$\mu_1$ & $\sigma_2$ &$\mu_2$ &$\eta$\\ \hline
ACP& 1.00 & 0.98   & 1.00   & 1.00 & 1.00  & 1.00  & 0.98  & 0.98 & 0.98 & 1.00\\
SB & 1.00   & 0.98   & 0.98  & 1.00  &  1.00  &  0.79  & 0.91 & 0.88 & 0.96 &0.87 \\ \hline \hline
	\end{tabular}}
	\captionsetup{width=1.1\linewidth}
	\caption{Coverage results for the mixture model example: Correct refers to correct model specification, while Misspec. refers to the misspecified model setting. Results report are the actual frequentist coverage of the posteriors 95\% credible sets for the unknown parameter $\theta=(\mu^\top,\sigma^{\top},\eta)^\top$, where the true value is $\theta=(1,-1,1,0.8,0.5)^\top$. }
	\label{tab:cover-comp}
\end{table}


\section{Discussion: Similar Approaches and Alternatives}\label{sec:discuss}

\subsection{Similarities with Other Gibbs Posteriors}

{\color{black}The ACP, $\pi(\theta\mid\mathcal{Q}_n)$, is the solution to the variational optimization problem in \eqref{eq:gibbs} under the specific loss function $\mathcal{Q}_n(\theta)$; so long as this loss satisfies the regularity conditions in Theorem \ref{thm:two} the ACP delivers asymptotically calibrated inferences. Although the ACP is formulated as a general belief update, posteriors with a similar functional form do appear elsewhere in the Bayesian literature.}

{\color{black}
When the KL divergence is used as the regularizer in \eqref{eq:gibbs}, the kernel of the ACP can be written in terms of a quadratic form in empirical scores, 
which resembles the exponentiated moment-criterion used in some `quasi-Bayesian' posteriors. When faced with parameters defined via a set of \textit{over-identified estimation equations for $\theta$}, i.e., more equations than unknown parameters, \cite{chernozhukov2003mcmc} propose to conduct a type of posterior inference on $\theta$ by considering the exponentiated negative quadratic form based on a vector of sample moments as a kernel in MCMC; for a related approach see \citet{chib2018bayesian}. When the moments take the specific form of estimating equations derived from a GLM, \cite{yin2009} proposes a related quasi-Bayesian posterior that can be applied ``as long as the moments are correctly specified''.}


{\color{black}
The underlying philosophy and scope of quasi-Bayesian posteriors also differ from the ACP. Quasi-Bayesian posteriors begin with a set of correctly specified  moment conditions or estimating equations,  and treat the exponentiated (typically quadratic) criterion as a kernel in MCMC. In contrast, the ACP starts from an arbitrary loss and learning rate, which explicitly acknowledges possible model misspecification, and then uses the variational optimization problem in \eqref{eq:gibbs} to provide an asymptotically calibrated posterior update. Unlike quasi-posteriors, the general definition of the ACP as a variational optimizer based on a given loss and learning rate, ensures that this method can deliver asymptotically calibrated inferences across a wide range of losses arising from quasi-likelihoods, divergence-based discrepancies, and robust contamination losses. 
}

{\color{black}Philosophically, Gibbs posteriors are based on conducting a form of Bayesian inference \textit{on a population loss minimizer of interest that is defined by a user-chosen loss function.} Our approach takes this idea one step further: by using the general belief update defined in \eqref{eq:gibbs}, we can produce a posterior that conducts inference on the same population loss minimizer of interest  but in such a way that the resulting posterior appropriately quantifies uncertainty. Further, unlike existing generalized Bayes methods, where the choice of the learning rate requires careful thought and can greatly impact uncertainty quantification, there exists a default choice for the learning rate in the ACP: any sequence of learning rates $\omega_n\rightarrow1$ will deliver asymptotically calibrated inferences.}


\subsection{Alternatives}\label{sec:discuss1}
In a likelihood context, when the model is correctly specified $\Sigma^{-1}=W(\theta_\star)^{-1}=\mathcal{I}^{-1}_\star$, and credible sets built from the ACP and the Bayes posterior will asymptotically coincide. If the model is misspecified the ACP still yields credible sets that are calibrated so long as $W_n(\theta_\star)$ is a consistent estimator of $\mathcal{I}_\star$. 
%
%
Two alternative approaches that seek to deliver similar coverage guarantees and which have received meaningful attention are the `sandwich' correction suggested in \cite{muller2013risk}, and the BayesBag approach of \cite{huggins2019robust}. 

{\color{black}\citeauthor{muller2013risk}'s approach amounts to transforming the draws from the Bayes posterior using the explicit Gaussian approximation: $\theta\sim N\{\bar\theta,n^{-1}\mathcal{H}_n(\bar\theta)^{-1}W_n(\bar\theta)\mathcal{H}_n(\bar\theta)^{-1}\}$, where $\bar\theta$ is the posterior mean; see \cite{giummole2019objective} for a related approach in the case of generalized posteriors built using scoring rules.} We argue that this \textit{ex-post correction} is suboptimal for several reasons: first, philosophically, it amounts to the subsequent application of additional inference methods to the output of a Bayesian learning algorithm, and has no representation as a belief update in the sense of \eqref{eq:gibbs}; second, it requires the explicit calculation of second-derivatives, which can be difficult and may be ill-behaved; third, this Gaussian approximation is poor when posteriors are not Gaussian, e.g., when the parameters have restricted support, in small sample sizes, or when the posterior is multi-modal{\color{black}, see Sections \ref{sec:poiss} and \ref{example:non_unique} for such examples}; fourth, without additional constraints, this correction can easily produce values of $\theta$ that lie outside the support of $\pi(\theta)$, for instance, when $\Theta$ is a bounded subset of $\mathbb{R}^{d_\theta}$.

Posterior bagging, as suggested in the BayesBag approach (\citealp{huggins2019robust}) is an alternative method that attempts to correct posterior coverage through bagging. Let $b=1,\dots,B$ denote bootstrap indices, and $y^{(b)}=(y_1^{(b)},\dots,y_n^{(b)})$ the $b$-th bootstrap sample, where $y_i^{(b)}$ is sampled with replacement from the original dataset. The BayesBag posterior is then given by 
$$
\pi^\star(\theta \mid y)\approx B^{-1}\sum_{b=1}^{B}\pi(\theta\mid y^{(b)}).
$$BayesBag is easy to use, but \cite{huggins2019robust} demonstrate that it is not calibrated in general; {\color{black}however, it does seem to deliver reasonable results in many empirical situations and can be tuned so that inferences are conservative}. Further, we note that it is unclear how to easily apply BayesBag to weakly dependent data, which are easily handled by the ACP (see Section \ref{sec:whittle} for one such example). 

Similarly to BayesBag, the bootstrapping approach to {\color{black}tuning the learning rate originally proposed by \cite{syring2019calibrating} may not deliver calibrated inferences unless the generalized matrix information equality is satisfied. However, in most applications, it is possible to tune the learning rate so that the corresponding credible sets are conservative  (\citealp{martin2022direct})}.

\section{Conclusion}\label{sec:conclusions}

{\color{black}We propose a new generalized (Gibbs) posterior which ensures that the resulting inferences are calibrated asymptotically.} All existing approaches of which we are aware attempt to correct the coverage of generalized posteriors using either ex-post correction of the posterior draws, which are ultimately based on some (implicit) normality assumption on the resulting posterior, or  bootstrapping approaches. In contrast, our approach delivers correct uncertainty quantification without any tuning or ex-post corrections.

When the likelihood is intractable and must be estimated, a version of the ACP can still be implemented; however, we must then resort to importance sampling, or sequential importance sampling, along with  Fisher's identity to estimate the gradients used in the ACP, which can then be used in MCMC schemes to produce posterior inference. In such cases, analyzing the behavior of the posterior becomes more difficult than the case where the scores are not estimated, and obtaining theoretical results similar to those in Theorem \ref{thm:two} is more onerous. Given the additional technicalities that are required to implement such an approach, we leave this topic for future research.

\subsubsection*{Supplementary Material}
The supplementary material for this submission contains additional empirical results, examples, and proofs of all technical results. 
\subsubsection*{Acknowledgments}
David Frazier was supported by the Australian Research Council's Discovery Early Career Researcher Award funding scheme (DE200101070).  Christopher Drovandi was supported by the Australian Research Council Future Fellowships Scheme (FT210100260).  The authors thank Jeremias Knoblauch, and Jonathan Huggins for comments on a previous draft. The authors would also like to thank three anonymous referee's for their very helpful comments on a previous version of the draft.

 {
	\spacingset{1.0} 
	\footnotesize
\bibliographystyle{apalike}
\bibliography{library}
}

	\begin{abstract}
	This supplementary material contains additional empirical results and proofs of all technical results. This includes additional empirical results for the linear and Poisson regression model examples, and a range of additional examples, including robust models of location, and time series models. 

\end{abstract}

\spacingset{1.8} 

\appendix
\section{Intuition for Calibrated Inferences}\label{sec:soe}
In this section, we give an intuitive argument that explains why the posterior $\pi(\theta\mid \mathcal{Q}_n)$, which we refer to as the ACP, delivers calibrated inferences. To this end, we recall that 
$$
\mathcal{Q}_n(\theta):=\frac{1}{2}\log|W_n(\theta)|+n\cdot Q_n(\theta),\;\text{ where }Q_n(\theta):=\frac{1}{2}\cdot{\overline{m}_n(\theta)^\top}{}W_{n}(\theta)^{-1}{\overline{m}_n(\theta)}{}.
$$ Replacing $\MD_n(\theta)$ in the optimization problem defined in \eqref{eq:gibbs} with $\mathcal{Q}_n(\theta)$ produces the posterior 
\begin{equation}\label{eq:genpost}
	\pi(\theta\mid \mathcal{Q}_n):=\frac{M_n(\theta)^{-\frac{\omega}{2}}\exp\{- \omega\cdot n\cdot Q_n(\theta)\}\pi(\theta)}{\int_\Theta M_n(\theta)^{-\frac{\omega}{2}}\exp\{- \omega\cdot n\cdot Q_n(\theta)\}\pi(\theta)\dt\theta},
\end{equation}	
where $M_n(\theta)=|W_{n}(\theta)|$. We note here that we do not require $\mathcal{D}_n(\theta)$ to be differentiable everywhere, and that $\overline{m}_n(\theta)=n^{-1}\nabla_\theta \mathcal{D}_n(\theta)$ must only satisfy Assumption \ref{ass:infeasible}. This flexibility is useful if one wishes to handle losses such as the check function loss, which is not differentiable at specific points. As our goal in this section is to deliver intuition on why $\pi(\theta\mid\mathcal{Q}_n)$ is calibrated, in what follows we dispense with this technical requirement and assume that all derivatives which are written exist.

\subsection{Second-Order Expansions}\label{app:soe}

To understand why $\pi(\t\mid \mathcal{Q}_n)$  accurately quantifies uncertainty, first let us consider a second-order expansion of $Q_n(\theta)$ around the loss minimizer $\theta_n$, the solution to $0=\overline{m}_n(\theta)$. Expanding $Q_n(\theta)$ around $\theta_n$ up to the second order we obtain
\begin{flalign*}
	Q_n(\theta)=Q_n(\theta_n)+(\theta-\theta_n)^\top \nabla_\theta Q_n(\theta_n)&+\frac{1}{2}(\theta-\theta_n)^\top \nabla_\theta^2 Q_n(\theta_n) (\theta-\theta_n)\\&+\frac{1}{2}(\theta-\theta_n)^\top [\nabla_\theta^2 Q_n(\overline\theta_n)-\nabla_\theta^2 Q_n(\theta_n)] (\theta-\theta_n)+R_n(\theta),    
\end{flalign*}where $\overline\theta_n$ is a line-by-line intermediate value between $\theta$ and $\theta_n$. 

The derivatives in this expansion are relatively simple, but rather tedious to write, so we will instead consider the scalar case $\theta\in\mathbb{R}$ and use this to generalize to the case of vector $\theta\in\mathbb{R}^{d_{\theta}}$. In the scalar case, the first derivative in the expansion is given by 
$$
\nabla_\theta Q_n(\theta)=[\nabla_\theta\overline{m}_n(\theta)]W_n^{-1}(\theta)\overline{m}_n(\theta)-\frac{1}{2}\overline{m}_n(\theta)^\top W_n(\theta)^{-1}\nabla_\theta W_n(\theta) W_{n}^{-1}(\theta)\overline{m}_n(\theta),
$$
while the second derivative is more complicated, and given by 
\begin{flalign*}
	\nabla_\theta^2 Q_n(\theta)=[\nabla_\theta^2\overline{m}_n(\theta)]W_n^{-1}(\theta)\overline{m}_n(\theta)+&[\nabla_\theta\overline{m}_n(\theta)]W_n^{-1}(\theta)[\nabla_\theta\overline{m}_n(\theta)]\\&-2\overline{m}_n(\theta)^\top W_n(\theta)^{-1}\nabla_\theta W_n(\theta) W_{n}^{-1}(\theta)m_n(\theta)\\&+\overline{m}_n(\theta)^\top W_n(\theta)^{-1}\nabla_\theta W_n(\theta)W_n(\theta)^{-1}\nabla_\theta W_n(\theta) W_{n}^{-1}(\theta)m_n(\theta)\\&-\frac{1}{2}\overline{m}_n(\theta)^\top W_n(\theta)^{-1}\nabla_\theta^2 W_n(\theta)W_n(\theta)^{-1}\overline{m}_n(\theta).
\end{flalign*} 

Now, since $0=\overline{m}_n(\theta)$, when evaluated at $\theta_n$, the first derivative in the expansion becomes
\begin{flalign*}
	\nabla_\theta Q_n(\theta_n)&=[\nabla_\theta\overline{m}_n(\theta_n)]^{\top}W_n^{-1}(\theta_n)\overline{m}_n(\theta_n)-\frac{1}{2}\overline{m}_n(\theta_n)^\top W_n(\theta_n)^{-1}\nabla_\theta W_n(\theta_n) W_{n}^{-1}(\theta_n)\overline{m}_n(\theta_n)\\&=0.
\end{flalign*}To obtain the form of the second term in the expansion, note that if any $\overline{m}_n(\theta_n)$ appears in any term, then this term is zero. Consequently, the only term that remains in the second derivative when evaluated at $\theta_n$ is $[\nabla_\theta\overline{m}_n(\theta)]W_n^{-1}(\theta)[\nabla_\theta\overline{m}_n(\theta)]$, so that 
\begin{flalign*}
	\nabla_\theta^2 Q_n(\theta)=    [\nabla_\theta\overline{m}_n(\theta)]W_n^{-1}(\theta)[\nabla_\theta\overline{m}_n(\theta)]=\mathcal{H}_n(\theta)W_n(\theta)^{-1}\mathcal{H}_n(\theta).
\end{flalign*}
Generalizing the above to the vector case delivers the second-order expansion
\begin{flalign*}
	Q_n(\theta)&=Q_n(\theta_n)+(\theta-\theta_n)^\top \nabla_\theta^2 Q_n(\theta_n) (\theta-\theta_n)+\frac{1}{2}(\theta-\theta_n)^\top [\nabla_\theta^2 Q_n(\overline\theta_n)-\nabla_\theta^2 Q_n(\theta_n)] (\theta-\theta_n)\\&=   \frac{1}{2} (\theta-\theta_n)^\top [\nabla_\theta\overline{m}_n(\theta)]^\top W_n^{-1}(\theta)[\nabla_\theta\overline{m}_n(\theta)](\theta-\theta_n)+R_n(\theta)\\&=\frac{1}{2}(\theta-\theta_n)^\top\mathcal{H}_n(\theta_n)^\top W_n(\theta_n)^{-1}\mathcal{H}_n(\theta_n)(\theta-\theta_n)+R_n(\theta).
\end{flalign*}

To apply the above within the Gibbs posterior $\pi(\theta\mid \mathcal{Q}_n)$, note that 
\begin{flalign*}
	\pi(\theta\mid \mathcal{Q}_n)=&\frac{M_n(\theta)^{-\frac{\omega}{2}}\exp\{-\omega\cdot n\cdot Q_n(\theta)\}\pi(\theta)}{\int M_n(\theta)^{-\frac{\omega}{2}}\exp\{-\omega\cdot n\cdot Q_n(\theta)\}\pi(\theta)\dt\theta}\\=&\frac{M_n(\theta)^{-\frac{\omega}{2}}\exp\{-\omega\cdot n\cdot \left[Q_n(\theta)-Q_n(\theta_n)\right]\}\pi(\theta)}{\int M_n(\theta)^{-\frac{\omega}{2}}\exp\{-\omega\cdot n\cdot \left[Q_n(\theta)-Q_n(\theta_n)\right]\}\pi(\theta)\dt\theta}\\=&\frac{M_n(\theta)^{-\frac{\omega}{2}}\exp\{-\omega\cdot n\cdot \left[\frac{1}{2}(\theta-\theta_n)^\top\mathcal{H}_n(\theta_n)^\top W_n(\theta_n)^{-1}\mathcal{H}_n(\theta_n)(\theta-\theta_n)-R_n(\theta)\right]\}\pi(\theta)}{\int M_n(\theta)^{-\frac{\omega}{2}}\exp\{-\omega\cdot n\cdot \left[\frac{1}{2}(\theta-\theta_n)^\top\mathcal{H}_n(\theta_n)^\top W_n(\theta_n)^{-1}\mathcal{H}_n(\theta_n)(\theta-\theta_n)-R_n(\theta)\right]\}\pi(\theta)\dt\theta}
	\\\propto & M_n(\theta)^{-\frac{\omega}{2}}\exp\left\{-\frac{\omega}{2}\sqrt{n}(\theta-\theta_n)^\top\mathcal{H}_n(\theta_n)^\top W_n(\theta_n)^{-1}\mathcal{H}_n(\theta_n)\sqrt{n}(\theta-\theta_n)-\omega\cdot n\cdot R_n(\theta)\right\}\pi(\theta)
\end{flalign*}

Defining $\Delta_n(\theta)=\mathcal{H}_n(\theta)^\top W_n(\theta)^{-1}\mathcal{H}_n(\theta)$, we can state this more compactly as 
$$
\pi(\theta\mid \mathcal{Q}_n)\propto \underbrace{e^{-\frac{\omega n}{2}(\theta-\theta_n)^\top{\Delta_n(\theta)}(\theta-\theta_n)}}_{\mathrm{Term (1)}}\underbrace{\pi(\theta)M_n(\theta)^{-\frac{\omega}{2}}e^{-\omega\cdot n R_n(\theta)}}_{\mathrm{Term (2)}}
$$We can immediately recognize that Term (1) is the kernel of a Gaussian density with mean $\theta_n$ and variance $[(n\omega)\Delta_n(\theta_n)]^{-1}$. Therefore, if we can show that Term (2) can be suitably controlled, then in large samples the first term will dominate and $\pi(\theta\mid \mathcal{Q}_n)$ will behave as if it were a Gaussian density with mean $\theta_n$ and variance $[(n\omega)\Delta_n(\theta_n)]^{-1}$. Further, the variance term  $[(n\omega)\Delta_n(\theta_n)]^{-1}$ is precisely the sandwich form that delivers calibrated uncertainty quantification. The bulk of the technical arguments in Theorem \ref{thm:two} are centered around demonstrating that these additional terms do not meaningfully alter the behavior of $\pi(\theta\mid\mathcal{Q}_n)$ in large samples. 

Now, if we have unique identification, then we can show that $\theta_n=\theta_\star+o_p(1)$, and $W_n(\theta_\star)=\mathcal{I}(\theta_\star)+o_p(1)$, so that the following convergence follows: $$\Delta_n^{-1}(\theta_n)=[\mathcal{H}_n(\theta_n) W_n(\theta_n)^{-1}\mathcal{H}_n(\theta_n)]^{-1}=\mathcal{H}(\theta_\star)^{-1}\mathcal{I}(\theta_\star)\mathcal{H}(\theta_\star)^{-1}+o_p(1)=\Sigma_\star+o_p(1).$$ In this case, the $1-\alpha$ credible sets formed from $\pi(\theta\mid \mathcal{Q}_n)$ are elliptical and asymptotically equivalent to the set:
$$
\left\{\theta\in\Theta:(\theta-\theta_n)^\top \frac{\Sigma_\star}{(n\omega)}(\theta-\theta_n)\le c_{1-\alpha}\right\},
$$ for some threshold $c_{1-\alpha}$.  The above set is the usual ``Wald-type'' frequentist confidence set for testing the null hypothesis $H_0:\theta=\theta_\star$. As a consequence, in large samples,  credible sets for $\theta_\star$ based on $\pi(\theta\mid \mathcal{Q}_n)$ will be calibrated. 
\subsection{Specific Example: Generalized Linear Models}\label{sec:gml_app}
Given a specific class of loss functions for $\theta$, we can derive explicit formulas for the ACP $\pi(\theta\mid \mathcal{Q}_n)$ that can be used to analyze the asymptotic shape of its confidence sets. In what follows, we carry out such a discussion for the class of quasi-likelihood loss functions. 

For $i=1,\dots,n$ we observe responses $y_i\in\mathbb{R}$, and covariates  $x_i\in\mathbb{R}^{d_\theta}$, with $x_{i,1}=1$. Our goal is inference on the unknown regression parameter $\theta$ in the generalized linear model (GLM):
\begin{equation*}
	\E(y_i\mid x_i)=\mu_i=g^{-1}(x_i^\top\theta),\quad \text{var}(y_i\mid x_i)=V(\mu_i;\psi),
\end{equation*}where $g(\cdot)$ is a strictly monotone and differentiable link function, and $V(\cdot)$ is a positive and continuous variance function  with dispersion parameter $\psi$. We have prior beliefs $\pi(\theta)\propto 1$, while we treat $\psi$ as a hyper-parameter. In this case, inference on $\theta$ can often be carried out using a quasi-log-likelihood loss, which delivers the cumulative loss  function
$$
\mathcal{D}_n(\theta) = -\sum_{i=1}^{n}\int^{\mu_i}_{-\infty} \frac{y_i - t}{V(t;\psi)} , \dt t.
$$
From here, it is straightforward to give the first and second order derivatives of the loss and to deduce the shape of the resulting posterior. 

In this case, the average gradient of $\mathcal{D}_n(\theta)$, i.e., $\overline{m}_n(\theta)$, is given by
$$
\overline{m}_n(\theta)
= -\frac{1}{n}\sum_{i=1}^n x_i\frac{\{y_i - \mu_i(\theta)\}}{V[\mu_i(\theta);\psi]} \cdot \mu_i'(x_i^\top\theta),
$$ where $\mu_i'(x_i^\top\theta)=\nabla_{\eta}\mu_i(\eta){|}_{\eta=x_i^\top\theta}$. The Hessian is given by 
$$
\mathcal{H}_n(\theta)= -\sum_{i=1}^n \left\{
\left[ -\frac{[\mu_i'(\eta_i)]^2}{V(\mu_i)} - \frac{(y_i - \mu_i) V'(\mu_i)}{V(\mu_i)^2} [\mu_i'(\eta_i)]^2 + \frac{y_i - \mu_i}{V(\mu_i)} \mu_i''(\eta_i) \right] x_i x_i^\top
\right\},
$$where $V'(\mu_i)=\nabla_\mu V(\mu)|_{\mu=g^{-1}(x_i^\top\theta)}$. It is also possible to give an explicit formula for $W_n(\theta)$ in this setting: for $\varepsilon_i(\theta)=\mu_i'(x^\top_i\theta)\{y_i-\mu_i(\theta)\}/V\{\mu_i;\psi\}$, we have 
$$
W_n(\theta)=\frac{1}{n}\sum_{i=1}^{n}x_ix_i^\top\epsilon_i(\theta)^2.
$$

Using these components, for $\omega=1$, the ACP has the straightforward expression 
$$
\pi(\theta\mid \mathcal{Q}_n)\propto |W_n(\theta)|^{-\frac{\omega}{2}}e^{-\frac{\omega}{2}\left\{\frac{1}{\sqrt{n}}\sum_{i=1}^n x_i\varepsilon_i(\theta)\right\}^\top \left\{\frac{1}{n}\sum_{i=1}^{n}x_i\varepsilon_i(\theta)^2x_i^\top\right\}^{-1}\left\{\frac{1}{\sqrt{n}}\sum_{i=1}^n x_i\varepsilon_i(\theta)\right\}}\times \pi(\theta).
$$

Two primary cases explored in the main text are the linear and Poisson regression. In such cases, the posterior given above can be simplified using the following. 
\begin{enumerate}
	\item Linear Regression: Take $\mu(\eta)=\eta=x^\top\theta$, and $V[\mu;\psi]=1$, so that $\varepsilon_i(\theta)=(y_i-x_i^\top\theta)$.
	\item Poisson Regression: Take $\mu(\eta)=e^\eta$, and $V[\mu;\psi]=\mu=e^\eta=e^{x^\top\theta}$, so that $\varepsilon_i(\theta)=(y_i-\mu(x_i^\top\theta))$.
\end{enumerate}

As pointed out to us by an anonymous referee, in the special case of GLMs estimated under quasi-likelihood loss, the ACP resembles the Bayesian generalized method of moments (GMM) approach proposed by \cite{yin2009}. In GLMs, \cite{yin2009} proposed a quasi-posterior based on the GMM criterion function constructed using the quasi-likelihood estimating equations $\overline{m}_n(\theta)$. While the two posteriors are indeed similar in this context, they will differ even when $\omega=1$ as the two posterior kernels differ by the determinant term $|W_n(\theta)|^{-\omega/2}$. More generally, since all we require for the ACP to deliver well-calibrated inferences is that the tuning parameter satisfies $\omega_n\rightarrow 1$, the two posteriors will be distinct in finite samples. We refer to Section 2.2 and 5.1 in the main paper for a more detailed comparison.

\section{Additional Results: Regression models}
\subsection{Normal Linear Regression}\label{app:sup_results_reg}
This section contains complete numerical results for the case of $d=20$ covariates for the normal linear regression model considered in Section \ref{sec:reg} of the main text. Please refer to Section \ref{sec:reg} for full details on the numerical experiments. 
\begin{table}[H]
	\hspace*{-1.25cm} 			
	\centering
	{\footnotesize
		\begin{tabular}{rrrrrrrrrrrrr}
			\hline\hline
			& \multicolumn{3}{c}{ACP} & \multicolumn{3}{c}{SB} & \multicolumn{3}{c}{HrB}& \multicolumn{3}{c}{PostCorr}\\\hline\hline
			& Bias & Var &Cover& Bias & Var &Cover&Bias & Var &Cover&Bias & Var &Cover\\ \hline
			$\beta_1$ & 0.0009 & 0.001 &  0.97 &  0.0009  & 0.001 &  0.97 &  0.0016 &  0.0014 &  0.97 &  0.0009 &  0.001 &  0.96 \\
			$\beta_2$ & -0.0025 &  0.0011 &  0.97 &  -0.0024 &  0.001 &  0.97 &  -0.0029 &  0.0014 &  0.98 &  -0.0025 &  0.001 &  0.97 \\
			$\beta_3$ & 0.0057 &  0.001 &  0.94 &  0.0056 &  0.001 &  0.94 &  0.0049 &  0.0015 &  0.94 &  0.0056 &  0.001 & 0.93 \\
			$\beta_4$ & 0.0059 &  0.0011 &  0.93 &  0.0059 &  0.001 &  0.93 &  0.0057 &  0.0014 &  0.94 &  0.0059 &  0.001 &  0.91 \\
			$\beta_5$ & -0.0024 &  0.001 & 0.95 &  -0.0023 &  0.001 &  0.93 &  -0.0022 &  0.0013 &  0.97 &  -0.0023 &  0.001 &  0.92 \\
			$\beta_6$ & 0.0009 &  0.001 &  0.97 &  0.0008 & 0.001 & 0.96 & 0.0003 & 0.0015 & 0.98 & 0.0008 & 0.001 & 0.95 \\
			$\beta_7$ & -0.0033 & 0.001 & 0.93 & -0.0032&  0.001&  0.93&  -0.0029&  0.0013&  0.94&  -0.0032&  0.001&  0.93 \\
			$\beta_8$ & -0.0004&  0.001&  0.96&  -0.0007&  0.001&  0.96&  -0.0006&  0.0014&  0.97&  -0.0006&  0.001&  0.96 \\
			$\beta_9$ & -0.003&  0.001&  0.97&  -0.0029&  0.001&  0.96&  -0.0033&  0.0014&  0.97&  -0.0029&  0.001&  0.97 \\
			$\beta_{10}$ & 0.0 &  0.001&  0.93&  0.0002&  0.001&  0.93&  -0.0003&  0.0014&  0.94&  0.0002&  0.001&  0.93 \\
			$\beta_{11}$ & 0.0047&  0.001&  0.96&  0.0046&  0.001&  0.96&  0.0056&  0.0013&  0.97&  0.0046&  0.001&  0.96 \\
			$\beta_{12}$ & 0.0029&  0.001&  0.95&  0.0026&  0.001&  0.92&  0.0029&  0.0013&  0.95&  0.0026&  0.001&  0.94 \\
			$\beta_{13}$ & 0.0026&  0.001&  0.96&  0.0025&  0.001&  0.95&  0.0037&  0.0014&  0.98&  0.0025&  0.001&  0.95 \\
			$\beta_{14}$ & 0.0073&  0.001&  0.93&  0.0075&  0.001&  0.91&  0.0076&  0.0013&  0.94&  0.0075&  0.001&  0.93 \\
			$\beta_{15}$ & 0.0056&  0.001&  0.95&  0.0056&  0.001&  0.95&  0.0056&  0.0014&  0.95&  0.0056&  0.001&  0.95 \\
			$\beta_{16}$ & -0.0035&  0.001&  0.96&  -0.0033&  0.001&  0.96&  -0.0027&  0.0014&  0.95&  -0.0033&  0.001&  0.96 \\
			$\beta_{17}$ & 0.0002&  0.001&  0.96&  0.0006&  0.001&  0.96&  0.0008&  0.0014&  0.96&  0.0006&  0.001&  0.96 \\
			$\beta_{18}$ & -0.0017&  0.001&  0.98&  -0.0018&  0.001&  0.97&  -0.0018&  0.0014&  0.97&  -0.0018&  0.001&  0.97 \\
			$\beta_{19}$ & 0.0055&  0.001&  0.95&  0.0056&  0.001&  0.95&  0.0056&  0.0014&  0.96&  0.0056&  0.001&  0.95 \\
			$\beta_{20}$ & 0.0032&  0.001&  0.97&  0.0032&  0.001&  0.97&  0.003&  0.0014&  0.97&  0.0032&  0.001&  0.97 \\ \hline \hline
		\end{tabular}
		\caption{Accuracy results in the normal linear regression model for standard Bayes (SB), HrB-posterior, ACP and the Gaussian posterior correction method (PostCorr) when $d=20$ and $\gamma = 0$ (i.e.\ correctly specified).  Bias is the bias of the posterior mean across the replications. Var is the average posterior variance deviation across the replications. Cover is the actual posterior coverage, where the nominal level is set to 95\% for the experiments.}}
	\label{tab:reg_d20_correct}
\end{table}

\begin{table}[H]
	\hspace*{-1.25cm} 			
	\centering
	{\footnotesize
		\begin{tabular}{rrrrrrrrrrrrr}
			\hline\hline
			& \multicolumn{3}{c}{ACP} & \multicolumn{3}{c}{SB} & \multicolumn{3}{c}{HrB} & \multicolumn{3}{c}{PostCorr}\\\hline\hline
			& Bias & Var &Cover& Bias & Var &Cover&Bias & Var &Cover&Bias & Var &Cover\\ \hline
			$\beta_1$ & 0.0011 &  0.001 &  0.96 &  0.001 &  0.001 &  0.95 &  -0.0006 &  0.0012&  0.97&  0.0011&  0.001&  0.97 \\
			$\beta_2$ & -0.0026&  0.0017&  0.99&  -0.0028&  0.001&  0.91&  -0.0028&  0.0019&  0.94&  -0.0028&  0.0016&  0.97 \\
			$\beta_3$ & 0.007&  0.0017&  0.94&  0.0071&  0.001&  0.87&  0.0056&  0.0018&  0.9&  0.0071&  0.0017&  0.94 \\
			$\beta_4$ & 0.0062&  0.0011&  0.94&  0.006&  0.001&  0.93&  0.0056&  0.0014&  0.94&  0.006&  0.001&  0.93 \\
			$\beta_5$ & -0.0008&  0.001&  0.97&  -0.0007&  0.001&  0.96&  -0.0024&  0.0013&  0.99&  -0.0006&  0.001&  0.96 \\
			$\beta_6$ & 0.0017&  0.001&  0.97&  0.0016&  0.001&  0.96&  0.0019&  0.0017&  0.98&  0.0016&  0.001&  0.96 \\
			$\beta_7$ & -0.0025&  0.001&  0.95&  -0.0023&  0.001&  0.95&  -0.0027&  0.0014&  0.95&  -0.0023&  0.001&  0.95 \\
			$\beta_8$ & -0.0008&  0.001&  0.95&  -0.0009&  0.001&  0.95&  -0.0007&  0.0014&  0.97&  -0.0009&  0.001&  0.95 \\
			$\beta_9$ & -0.0033&  0.001&  0.98&  -0.0035&  0.001&  0.96&  -0.002&  0.0014&  0.99&  -0.0035&  0.001&  0.98 \\
			$\beta_{10}$ & -0.0008&  0.001&  0.92&  -0.0005&  0.001&  0.91&  0.0005&  0.0014&  0.94&  -0.0005&  0.001&  0.91 \\
			$\beta_{11}$ & 0.0045&  0.0011&  0.96&  0.0045&  0.001&  0.96&  0.0053&  0.0014&  0.98&  0.0045&  0.001&  0.96 \\
			$\beta_{12}$ & 0.0024&  0.001&  0.94&  0.0021&  0.001&  0.92&  0.0033&  0.0013&  0.97&  0.0021&  0.001&  0.93 \\
			$\beta_{13}$ & 0.0025&  0.001&  0.96&  0.0026&  0.001&  0.96&  0.0005&  0.0014&  0.94&  0.0026&  0.001&  0.95    \\
			$\beta_{14}$ &  0.0063&  0.001&  0.95&  0.0067&  0.001&  0.95&  0.0078&  0.0013&  0.94&  0.0067&  0.001&  0.94    \\
			$\beta_{15}$ &  0.0043&  0.001&  0.94&  0.0046&  0.001&  0.94&  0.0038&  0.0014&  0.95&  0.0046&  0.001&  0.94 \\
			$\beta_{16}$ &  -0.0042&  0.0011&  0.97&  -0.0038&  0.001&  0.97&  -0.0029&  0.0014&  0.98&  -0.0038&  0.001&  0.97 \\
			$\beta_{17}$ &  0.0002&  0.001&  0.95&  0.0004&  0.001&  0.95&  0.0005&  0.0013&  0.95&  0.0004&  0.001&  0.95 \\
			$\beta_{18}$ &  -0.0011&  0.001&  0.99&  -0.0012&  0.001&  0.98&  -0.0037&  0.0013&  0.96&  -0.0012&  0.001&  0.97 \\
			$\beta_{19}$ &  0.0059&  0.001&  0.98&  0.006&  0.001&  0.93&  0.0054&  0.0013&  0.97&  0.006&  0.001&  0.94 \\
			$\beta_{20}$ &  0.0037&  0.001&  0.96&  0.0039&  0.001&  0.96&  0.0019&  0.0014&  0.97&  0.0039&  0.001&  0.97 \\ \hline \hline
		\end{tabular}
		\caption{Accuracy results in the normal linear regression model for standard Bayes (SB), HrB-posterior and ACP when $d=20$ and $\gamma = 2$ (i.e.\ misspecified).  Bias is the bias of the posterior mean across the replications. Var is the average posterior variance deviation across the replications. Cover is the actual posterior coverage, where the nominal level is set to 95\% for the experiments.} }
	\label{tab:reg_d20_misspecified}
\end{table}

\subsection{Additional Results Poisson Regression}\label{app:sup_results_poiss}

This section contains complete numerical results for the case of $d_\theta=10$ and $d_\theta=20$ covariates for the Poisson regression model considered in Section \ref{sec:poiss} of the main text. Please refer to Section \ref{sec:poiss} for full details on the numerical experiments. 
\begin{table}[H]
	\hspace*{-1.25cm} 			
	\centering
	{\footnotesize
		\begin{tabular}{rrrrrrrrrrrrr}
			\hline\hline
			& \multicolumn{3}{c}{ACP} & \multicolumn{3}{c}{SB} & \multicolumn{3}{c}{GB}& \multicolumn{3}{c}{PostCorr}\\\hline\hline
			& Bias & Var &Cover& Bias & Var &Cover&Bias & Var &Cover&Bias & Var &Cover\\ \hline
			$\theta_1$         &        -0.0144  &             0.055   &  0.92     &           -0.0100      &          0.037  &    0.88      &      -0.0100    &        0.055   &   0.93 & -0.0009&  0.0073 & 0.93\\
			$\theta_2$        &          0.0136   &            0.032   &  0.96     &            0.01319       &         0.022    &  0.94        &     0.0132        &    0.032     & 0.97 &0.0013&   0.0055  & 0.97 \\
			$\theta_3$       &         -0.0045   &            0.032   &  0.95     &           -0.0041     &           0.021   &   0.88      &      -0.0040     &       0.032   &   0.95& -0.0004&    0.0055 & 0.95\\
			$\theta_4$       &         -0.0016   &            0.032    & 0.95      &          -0.0022      &          0.022    &  0.88      &      -0.0022      &      0.032    &  0.98&  -0.0002&   0.0055 & 0.96  \\
			$\theta_5$       &          -0.0040  &             0.031   &  0.93     &           -0.0040    &            0.021  &    0.89     &       -0.0040     &       0.032   &   0.95&  -0.0004&   0.0054& 0.93 \\
			$\theta_6$      &            0.0059   &            0.032  &   0.96     &            0.0060     &           0.021  &    0.92      &       0.0060       &     0.032   &   0.96&  0.0005&  0.0055  & 0.95\\
			$\theta_7$       &          0.0008    &           0.031   &  0.94       &          0.0087       &         0.021    &  0.91        &     0.0080          &  0.032    &  0.94&  0.0001&   0.0055 & 0.96 \\
			$\theta_8$      &            0.0074   &            0.031  &   0.97       &          0.0075      &          0.021    &  0.93       &      0.0075         &   0.031    &  0.96&  0.0007&   0.0055& 0.96 \\
			$\theta_9$      &            0.0026   &            0.032  &   0.98      &           0.0027     &           0.022   &   0.95     &        0.0026        &    0.032   &   0.97&  0.0003&   0.0056& 0.97  \\
			$\theta_{10}$    &             0.0003   &            0.032&     0.95      &           0.0004    &            0.021   &   0.86     &        0.0003        &    0.032  &    0.94& 0.0001&  0.0056& 0.91\\\hline\hline
		\end{tabular}
		\captionsetup{width=1.1\linewidth}
		\caption{Results in the Poisson regression model for standard Bayes (SB), ACP, {\color{black}the generalized Bayes approach of \cite{agnoletto2023bayesian} (GB) based on the quasi-likelihood,  and the Gaussian posterior correction (PostCorr). The true value of the displayed parameter is $(\theta_1,\theta_2,\theta_3)=(3.5,0.5,-0.5,0,\dots,0)$, where for both $d_\theta=10$ and $d_\theta=20$, only the first three coefficients are non-zero. See Table \ref{tab:reg} for a detailed description of the table entries.}{\color{black} Bias has been rounded to three decimal places}.}
		\label{tab:poiss_d10_correct}
	}
\end{table}

\begin{table}[H]
	\hspace*{-1.25cm} 			
	\centering
	{\footnotesize
		\begin{tabular}{rrrrrrrrrrrrr}
			\hline\hline
			& \multicolumn{3}{c}{ACP} & \multicolumn{3}{c}{SB} & \multicolumn{3}{c}{GB}& \multicolumn{3}{c}{PostCorr}\\\hline\hline
			& Bias & Var &Cover& Bias & Var &Cover&Bias & Var &Cover&Bias & Var &Cover\\ \hline
			$\theta_1$      &          -0.0157           &    0.056&     0.95       &         -0.0124  &              0.037  &    0.90  &          -0.0112         &   0.056     &   0.97& -0.0014& 0.0076& 0.87 \\
			$\theta_{2}$     &           -0.0050         &      0.033  &   0.94     &           -0.0048    &            0.022 &     0.88    &        -0.0048      &      0.034   &   0.95& 0.0004 & 0.0060& 0.94    \\
			$\theta_{3}$      &           0.0020        &       0.033&     0.96   &              0.0025        &        0.022   &   0.93        &     0.0025      &      0.033 &     0.96&  0.0003& 0.0060&  0.93   \\
			$\theta_{4} $     &           0.0060      &         0.033&     0.96                & 0.0056          &      0.022   &   0.92          &   0.0053   &         0.034   &   0.96&  0.0005& 0.0060&  0.93   \\
			$\theta_{5}$       &          0.0066    &           0.033&     0.94              &   0.0066            &    0.022  &    0.89 &            0.0066           & 0.033     & 0.94&  0.0012&  0.0060&     0.98   \\
			$\theta_{6}$       &          0.0026  &             0.033 &    1.00            &     0.0028   &             0.022   &   0.94    &         0.0029          &  0.033 &     0.99&  -0.0003&    0.0060&   0.89  \\
			$\theta_{7}$        &         0.0064 &              0.033  &   0.98         &        0.0063    &            0.022   &   0.93     &        0.0063       &     0.033  &    0.98&  -0.0003&    0.0060&  0.94   \\
			$\theta_{8}$       &         -0.0033 &              0.032  &   0.97      &          -0.0033     &           0.022  &    0.89      &      -0.0036   &         0.033  &    0.97&  0.0002& 0.0060&  0.91   \\
			$\theta_{9}$        &         0.0059&               0.032    & 0.94     &            0.0062         &       0.022    &  0.91           &  0.0061    &        0.033  &    0.94&  -0.0001&    0.0060& 0.86    \\
			$\theta_{10}$       &        -0.0060          &     0.033&     0.96   &             -0.0061  &              0.022   &   0.92 &           -0.0059           & 0.033 &     0.95&  -0.0003&    0.0060&   0.89  \\
			$\theta_{11}$         &      -0.0068        &       0.033&     0.98  &              -0.0067     &           0.022   &   0.95     &       -0.0068          &  0.033     & 0.97&  0.0001&     0.0060&   0.90  \\
			$\theta_{12}$    &           -0.0080      &         0.033 &    0.98   &             -0.0076     &           0.022&      0.97      &      -0.0084      &      0.033&      0.98&  -0.0006&     0.0059&  0.87 \\
			$\theta_{13}$     &           0.0041     &          0.033 &    0.99 &                0.0043       &         0.022     & 0.95         &    0.00418      &      0.033    &  0.98&  -0.0014&     0.0059&  0.93 \\
			$\theta_{14}$      &          0.0127     &          0.033  &   0.98                & 0.0126           &     0.022    &  0.95             &0.01281      &      0.033   &   0.98&  0.0014&    0.0059&    0.92   \\
			$\theta_{15}$       &         0.0065   &            0.033 &    0.93             &    0.0073            &    0.022    &  0.88  &           0.0077  &          0.033   &   0.92&  0.0011&    0.0059&    0.96   \\
			$\theta_{16}$        &       -0.0023 &              0.032 &    0.95          &      -0.0022              &  0.022  &    0.90    &        -0.0022          &  0.033    &  0.94&  0.0007&    0.0059&     0.95  \\
			$\theta_{17}$          &      0.0046 &              0.033  &   0.95        &         0.0046   &             0.022&      0.92       &      0.0047         &   0.033     & 0.95&  0.0013&    0.0059&      0.92 \\
			$\theta_{18}$           &    -0.0080  &             0.033 &    0.95     &           -0.0080    &            0.022      &0.89        &    -0.0078     &       0.033     & 0.94&  -0.0010&    0.0059 &     0.89  \\
			$\theta_{19}$   &            -0.0036 &              0.033 &    0.96    &            -0.0023      &          0.022    &  0.95          &  -0.0023  &          0.033      &0.96&  0.0016&   0.0059&        0.93 \\
			$\theta_{20} $   &            0.0045&               0.033  &   0.98    &             0.0044         &       0.022  &    0.94             &0.0045   &         0.033  &    0.98&	0.0001	&  0.0060&  0.92 \\\hline\hline
		\end{tabular}
		\caption{Accuracy results in the Poisson regression model for exact Bayes, ACP, and generalized Bayes (GenBayes) based on the quasi-likelihood.  Bias is the bias of the posterior mean across the replications and has been multiplied by 1000 for readability. Var is the average posterior variance deviation across the replications and has been multiplied by 100 for readability. Cover is the actual posterior coverage, where the nominal level is set to 95\% for the experiments.}}
	\label{tab:poiss_d20_correct}
\end{table}

\newpage

\section{Additional Examples}This section applies the ACP to several examples that have been treated in the generalized Bayesian literature. 

\subsection{Additional Non-unique Identification Example }\label{example:non_unique_tukey}
{This section illustrates the theoretical behavior of the ACP derived in Theorem \ref{thm:multi}. We assume that the observed data $y$ is generated from a normal distribution with unknown mean $\theta$ and known variance $\sigma^2$, which we fix at the estimated median absolute deviation for simplicity. However, we also believe the data is contaminated with outliers, and so we follow \cite{jewson2021general} and conduct generalized Bayesian inference on $\theta$ using Tukey's bi-weight loss: for $u_i=(y_i-\theta)/\sigma$, with $\sigma=\sqrt{\text{mad}(y_1,\dots,y_n)}$,
	$$
	\rho(u_i)
	=\begin{cases}
		\displaystyle\frac{\kappa^2}{6}\Bigl\{1-\bigl(1-\tfrac{u_i^2}{\kappa^2}\bigr)^3\Bigr\}
		&|u_i|\le \kappa\\
		\displaystyle\frac{\kappa^2}{6}&|u_i|>\kappa
	\end{cases}.
		$$
		Tukey's loss depends on the unknown mean $\theta$, and the hyper-parameter $\kappa$, which controls the tails of the loss. The value of $\kappa$ artificially deflates extreme data points, such as outliers, to ensure that these points do not unduly influence inferences. Interestingly, for certain data and depending on the value of $\kappa$, Tukey's loss admits multiple minimizers: if the observed data is a mixture of two well-separated groups then, for certain choices of $\kappa$, Tukey's loss produces multiple global minima that depend on the nature of the mixture.
		
		In this example, we conduct inference on $\theta$ in the assumed model and observe $n=500$ independent draws from a two-component Gaussian mixture model:  $0.5 \cdot\mathcal{N}(1,0.25) + 0.5\cdot \mathcal{N}(-1,0.25)$. Fixing $\kappa=1$, one can show that Assumption \ref{ass:set} is satisfied at three distinct points: $\theta=0$, $\theta=-1$ and $\theta=1$. Each of these minima are meaningful and interpretable: $\theta=\pm1$ correspond to the means of the mixture components, and $\theta=0$ corresponds to the population mean of the observed data.
		
		We compare the ability of the ACP and a Gaussian posterior correction to the generalized posterior - based on Tukey's loss - to produce calibrated inference. We specifically measure the ability of these methods to accurately capture the three population minimizers within their credible regions, which we calculate as a 95\% highest-density region of the posterior. Figure \ref{fig:tukey_posteriors} plots the ACP and the Gaussian posterior correction (PostCorr) for a representative sample of data. The results show that the posterior correction may not contain any of the three modes, while the ACP places nearly equal weight on all modes. This same behavior is observed in repeated samples: over 100 repeated samples, the ACP credible set for $\Theta_\star=\{-1,0,1\}$ exhibits 98\% frequentist coverage, while the PostCorr set contains $\Theta_\star$ in just 5\% of the replications. Further, the marginal coverage of PostCorr for the three individual points in $\Theta_\star$ is just 7\%, 31\% and 8\%, respectively. 
	}

	\begin{figure}[!htp]
		\centering
		\includegraphics[width=100mm,height=50mm]{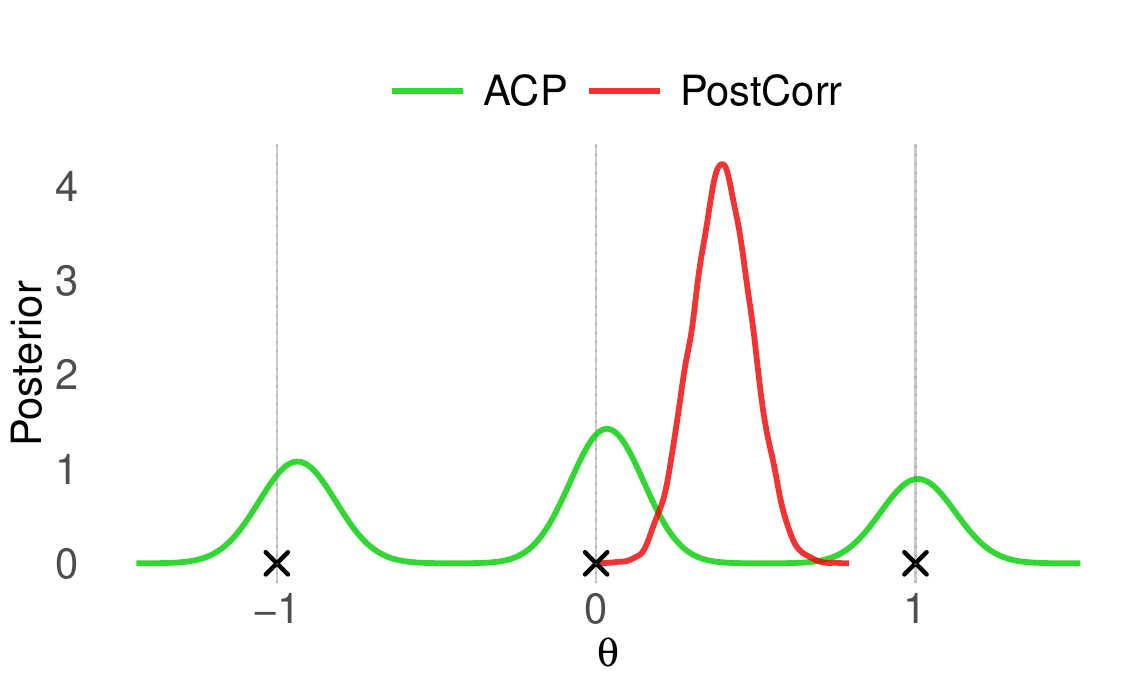}
		\captionsetup{width=1.1\linewidth}
		\caption{Univariate posterior distributions for the mean parameter under Tukey's loss for ACP (green) and the Gaussian posterior correction (PostCorr) (red) when $\Theta_\star=\{-1,0,1\}$ (marked by crosses on the $x$-axis).}
		\label{fig:tukey_posteriors}
	\end{figure}
	

	\subsection{Robust Location Inference}\label{ex:med}
	Consider observing a sequence $y_1,\dots,y_n$ from the model 
	$$
	y_i=\theta+\epsilon_i,\quad i=1,\dots,u,\quad \epsilon_i\stackrel{iid}{\sim}F(\cdot)
	$$where $F(\cdot)$ is unknown, the unknown parameter $\theta$ has prior density $\pi(\theta)$ and is independent of $\epsilon$. Our goal is posterior inference on $\theta$, and for reasons of robustness we follow \cite{doksum1990consistent}, and consider posterior inference based on the sample median $T_n=\text{med}(y_1,\dots,y_n)$, which is known to be robust to outliers in the data; for a related approach see \cite{lewis2021bayesian}. When $F$ has density $f$, and $n$ is odd, \cite{doksum1990consistent} show that the posterior $\pi(\theta\mid T_n)$ is
	$$
	\pi(\theta\mid T_n)  = \pi(\theta)\exp\left\{\frac{1}{2}(n-1)\log F(T_n-\theta)\{1-F(T_n-\theta)\}+\log f_\theta(T_n)(1-F(T_n-\theta)))\right\}.
	$$
	
	Given the above form of the posterior, \cite{miller2021asymptotic} suggests using generalized Bayesian inference for $\theta$ with the simpler loss function 
	$
	\mathcal{D}_n(\theta)=-\frac{1}{2}\log F(T_n-\theta)\{1-F(T_n-\theta)\}
	$. While \cite{miller2021asymptotic} demonstrates that such a posterior is well-behaved in large samples, the resulting posterior does not have correct coverage even when the model is correctly specified. It is simple to show that the same result, i.e., inaccurate posterior coverage, applies if one uses the original posterior $\pi(\theta\mid T_n)$ proposed in \cite{doksum1990consistent}, or the more general approach suggested in \cite{lewis2021bayesian}. 
	
	In contrast, if we conduct inference using the ACP based on the loss $\mathcal{D}_n(\theta)=\ell_n(\theta)$, then the resulting posterior is well-behaved and correctly quantifies uncertainty. The ACP requires calculating the gradient of $\mathcal{D}_n(\theta)$ with respect to $\theta$, given by
	$$
	\overline{m}_n(\theta)=\frac{1}{2n}\frac{f(T_n-\theta)}{F(T_n-\theta)}-\frac{1}{2}\frac{f(T_n-\theta)}{1-F(T_n-\theta)},
	$$ and the variance of $\sqrt{n}\overline{m}_n(\theta)$, for any $\theta\in\Theta$. For $W_n(\theta)$ denoting an estimator of $\text{var}\{\sqrt{n}\overline{m}_n(\theta)\}$,  the ACP is
	$$
	\pi(\theta\mid \mathcal{Q}_{n})\propto \pi(\theta) \exp\left[-\frac{1}{2n}\left\{\frac{1}{2}\frac{f(T_n-\theta)}{F(T_n-\theta)}-\frac{1}{2}\frac{f(T_n-\theta)}{1-F(T_n-\theta)}\right\}^2 /W_n(\theta)\right].
	$$
	
	A simple estimator for $\text{var}\{\sqrt{n}\overline{m}_n(\theta)\}$ can be obtained using bootstrap replications of the median, denoted as $\{T_n^{(b)}:b=1,\cdots,B\}$. Such an estimator is extremely fast since $\{T_n^{(b)}:b=1,\cdots,B\}$ are only generated once outside of the sampling algorithm. For any value of $\theta$, we can then take $W_n(\theta)$ as the sample variance of the observations $\{\overline{m}_n(\theta){|}_{T_n=T_n^{(b)}}:b=1,\dots,B\}$, which only requires evaluating the closed-form gradient $B$ times and then taking the sample variance of the evaluations.\footnote{Such an approach is similar to the estimating equations bootstrap, see \citealp{hu2000estimating} and \citealp{chatterjee2005generalized} for a discussion.} Consequently, the steps necessary to calculate $W_n(\theta)$ are precisely the same as forming any sample variance estimator.

	We now compare the uncertainty quantification produced using the generalized and ACPs in correctly and misspecified models. In both cases, we assume $F(\cdot)$ is a standard Gaussian distribution. In the first experiment, referred to as DGP1, the observed data is generated from a Gaussian distribution with mean $\theta=1$, and variance 4. In the misspecified regime, referred to as DGP2, we generate observed data from a mixed Gaussian distribution with parameterization
	$
	0.9N(\theta=1,4)+0.1N(0,1).
	$ 
	
	In the first case the true median is unity, while in the second case the true median of the data is approximately $0.84$; {this value is found numerically by inverting the CDF of the mixture distribution.} We simulate 1000 observed datasets from both DGPs and compare the results of generalized Bayes and that based on the ACP. Table \ref{tab:med} compares the posterior means, variances and coverage of the generalized Bayes and the ACP procedures. The results demonstrate that the generalized posterior proposed in \cite{miller2021asymptotic} does not produce reliable coverage for the true median, while the coverage of the ACP is again close to the nominal level. 
	
	\begin{table}[H]
		\centering
		{\footnotesize
			\begin{tabular}{lrrrrrr}
				\hline\hline
				& \multicolumn{3}{c}{ACP} & \multicolumn{3}{c}{GenBayes} \\
				\textbf{DGP1}  & Bias & Var &Cover& Bias & Var &Cover \\ \hline
				$\theta$ &   -0.0031  &  0.0654 &   0.9420  & -0.0031  &  0.0152   & 0.6700\\
				\textbf{DGP2}  & Bias & Var & Cover& Bias & Var &Cover \\ \hline
				$\theta$ &     -0.0115  &  0.0607  &  0.9360  & -0.0118  &  0.0157 &   0.7160
				\\
				\hline\hline
		\end{tabular}}
		\caption{Posterior accuracy results for the median using the generalized posterior (GenBayes) and ACP.  Bias is the bias of the posterior mean across the replications. Var is the average posterior variance deviation across the replications. Cover is the posterior coverage (95\% nominal coverage). For GenBayes we set $\omega=1$ in all experiments.}
		\label{tab:med}
	\end{table}

	\subsection{Approximate Inference in Time Series Models}\label{sec:whittle}
	
	This example shows how the ACP can provide calibrated Bayesian inferences when the assumed data generating process is correct, but an approximate likelihood is used to speed-up computation.  Furthermore, this example highlights how automatic differentiation can be used to calculate the scores that are not easily accessible analytically.
	
	Let $\{X_{t}\}_{t=1}^{n}$ be a discretely observed, zero-mean, random variable generated according to an  autoregressive fractionally integrated moving average (ARFIMA$(2,d,1)$) model 
	\begin{align*}
		(1-\phi_1L-\phi_2L^2)^d X_t = (1-\vartheta_1L)\epsilon_t,\quad \theta=(\phi_1,\phi_2,\vartheta_1,d)^\top, 
	\end{align*}
	where $L$ is the lag operator, and $\epsilon_t \sim \mathcal{N}(0,\sigma^2)$. We consider observed data generated from the ARFIMA($2,d,1$) model with true parameter $\theta_\star = (\phi_{1\star}, \phi_{2\star}, \vartheta_{1,\star}, d_\star)^\top = (0.45, 0.1, -0.4, 0.4)^\top$ and we set $n = 15,000$.  
	
	For simplicity, we impose independent uniform priors on the components of $\theta$, i.e., for $\mathcal{U}(a,b)$ denoting a uniform random variable with support $(a,b)$, $\phi_1 \sim \mathcal{U}(-1,1)$, $\phi_2 \sim \mathcal{U}(-1,1)$, $\vartheta_1 \sim \mathcal{U}(-1,1)$ and $d \sim \mathcal{U}(-0.5,0.5)$.  \citet{bon2021accelerating} impose stationarity constraints on the parameters, but given the large sample size investigated here, the prior has minimal influence.

	The likelihood function of the ARFIMA model for large $n$ is computationally intensive \citep{Chan1998}.  Consequently, we approximate the likelihood by the Whittle likelihood \citep{whittle1953estimation} to form an approximate posterior; see, also, \citet{salomone2020spectral} and \citet{bon2021accelerating} for a similar approach. We compare the approximate posterior and the ACP based on the Whittle likelihood, and show that the ACP delivers reliable uncertainty quantification, while the approximate posterior has poor uncertainty quantification.

	The Whittle likelihood operates with the data and covariance function in the frequency domain, as opposed to the time domain, and is based on the spectral density of the model, and the periodogram. For a frequency $\omega\in[-\pi,\pi]$, the spectral density of the ARFIMA($2,d,1$) for $d\in(-0.5,0.5)$ is
	$$
	f_\theta(\omega)=\frac{\sigma^2}{2 \pi} \frac{|1-\vartheta_1 \exp (i \omega)|^2}{\left|\left\{1-\phi_1 \exp (i \omega)-\phi_2 \exp (2 i \omega)\right\}\{1-\exp (i \omega)\}^d\right|^2}.
	$$The periodigram of the data at the frequencies $\{\omega_{k}: k = -\lceil n/2 \rceil+1,\ldots, \lfloor n/2 \rfloor\}$ is 
	$$
	\mathcal{F}(\omega_{k}) = \frac{|J(\omega_{k})|^2}{n},\quad 	J(\omega_{k}) = \frac{1}{\sqrt{2\pi}}\sum_{t = 1}^{n}X_{t}\exp(-i \omega_{k}t),\quad  \omega_{k} = \frac{2\pi k}{n},
	$$where $	J(\omega_{k})$ is the  discrete Fourier transform (DFT) of the time series. The Whittle log-likelihood for zero-mean data is then (\citealp{whittle1953estimation})
	$$
	\ell_{\text{whittle}}(\theta) = -\sum_{k = -\lceil n/2 \rceil+1}^{\lfloor n/2 \rfloor}\left(\log f_{\theta}(\omega_{k}) + \frac{\mathcal{F}(\omega_{k})}{f_{\theta}(\omega_{k})}\right).
	$$
	In practice, it is only necessary to calculate the summation over about half of the Fourier frequencies, $\omega_k$, due to symmetry about $\omega_{0}=0$ and since $f_{\theta}(\omega_{0}) = 0$ for centred data.
	
	The periodogram can be calculated in $\mathcal{O}(n \log n)$ time, and only needs to be calculated once per dataset. Given the periodogram, the cost of each subsequent likelihood evaluation is $\mathcal{O}(n)$, compared to the usual likelihood cost for time series (with dense precision matrix) which is $\mathcal{O}(n^2)$.
	
	The Whittle likelihood effectively treats $\mathcal{F}(\omega_{k})$ for each $\omega_{k}$ as being an independent exponential random variate with mean $f_{\theta}(\omega)$.  Given the assumed independence, it makes the ACP easy to apply in principle.  However, the score for each component of $\ell_{\text{whittle}}(\theta)$ is not easy to obtain analytically.  Herein, we use automatic differentiation in \texttt{Julia} \citep{Julia-2017}, specifically the \texttt{ForwardDiff.jl} library \citep{RevelsLubinPapamarkou2016} to compute the scores $\overline{m}_n(\theta)$.
	
	We use the \texttt{AdaptiveMCMC.jl} library \citep{vihola2014ergonomic} to draw samples from the posterior under the Whittle approximation and the ACP approach.   So long as the model is correctly specified, and under appropriate regularity conditions and restrictions on the true parameters, point-estimates obtained using the Whittle likelihood converge asymptotically to the true value generating the data  (see, e.g., \citealp{fox1986large}, and \citealp{dahlhaus1989efficient}).  Thus, given the large sample size considered here, the pseudo-true parameter and the true parameter coincide.  Table \ref{tab:whittle_results} gives accuracy measures based on 100 independent datasets (we use the \texttt{arfima} package \citep{Veenstra2012} in \texttt{R} \citep{R2022} to simulate datasets).  The results show that the Whittle approximation has low bias and MSE, but produces posterior approximations that are over-concentrated, as the poor coverage rates show.  In contrast, the ACP approximations exhibit substantially more accurate uncertainty quantification, but still having biases that are similar to those based on the Whittle likelihood.
	
	\begin{table}[!htp]
		\centering
		{\footnotesize	
			\begin{tabular}{lcccc}
				\hline
				Method & MSE & Bias & St.\ Dev.\ & Coverage (90\%)\\ 
				\hline
				\hline
				\multicolumn{5}{c}{$\phi_1$}  \\
				\hline  
				Whittle & 0.0028  & -0.0055  & 0.030 & 72  \\ 
				ACP &  0.0038 & -0.0024  & 0.044  & 93  \\ 
				\hline
				\multicolumn{5}{c}{$\phi_2$}  \\
				\hline  
				Whittle & 0.0013   & 0.00031  & 0.021  & 74  \\ 
				ACP & 0.0017  & -0.00038  & 0.029  & 91  \\ 
				\hline
				\multicolumn{5}{c}{$\vartheta_1$}  \\
				\hline  
				Whittle & 0.0021  & -0.0011  & 0.026  & 73  \\ 
				ACP & 0.0027  & 0.00024  & 0.037  & 89  \\ 
				\hline
				\multicolumn{5}{c}{$d$}  \\
				\hline  
				Whittle & 0.00078  & 0.0027  &  0.016 & 72  \\ 
				ACP & 0.0012  & 0.00061 & 0.025  & 88  \\ 
				\hline
			\end{tabular}
		}
		\caption{Results for the Whittle example using 100 independent datasets simulated with true parameter value $\theta = (\phi_1, \phi_2, \vartheta_1, d)^\top = (0.45, 0.1, -0.4, 0.4)^\top$. MSE denotes mean squared error, Bias the average bias, St.Dev. the average standard deviation, and Coverage is the actual coverage rate of the credible sets. }
		\label{tab:whittle_results}
	\end{table}
	
	In Figure \ref{fig:arfima_posteriors} we compare the approximate posteriors for the Whittle likelihood and ACP for a single dataset, the same one used in \citet{bon2022bayesian} generated with the same true parameter considered here.  We also compare the approximate posteriors with the true posterior generated in \citet{bon2022bayesian}.  It is evident that the posterior approximations based on the Whittle likelihood are overconcentrated compared to the true posterior.  In contrast, the ACPs are quite close to the true posteriors for this particular dataset.

	\begin{figure}[!htp]
		\centering
		\includegraphics[scale=0.7]{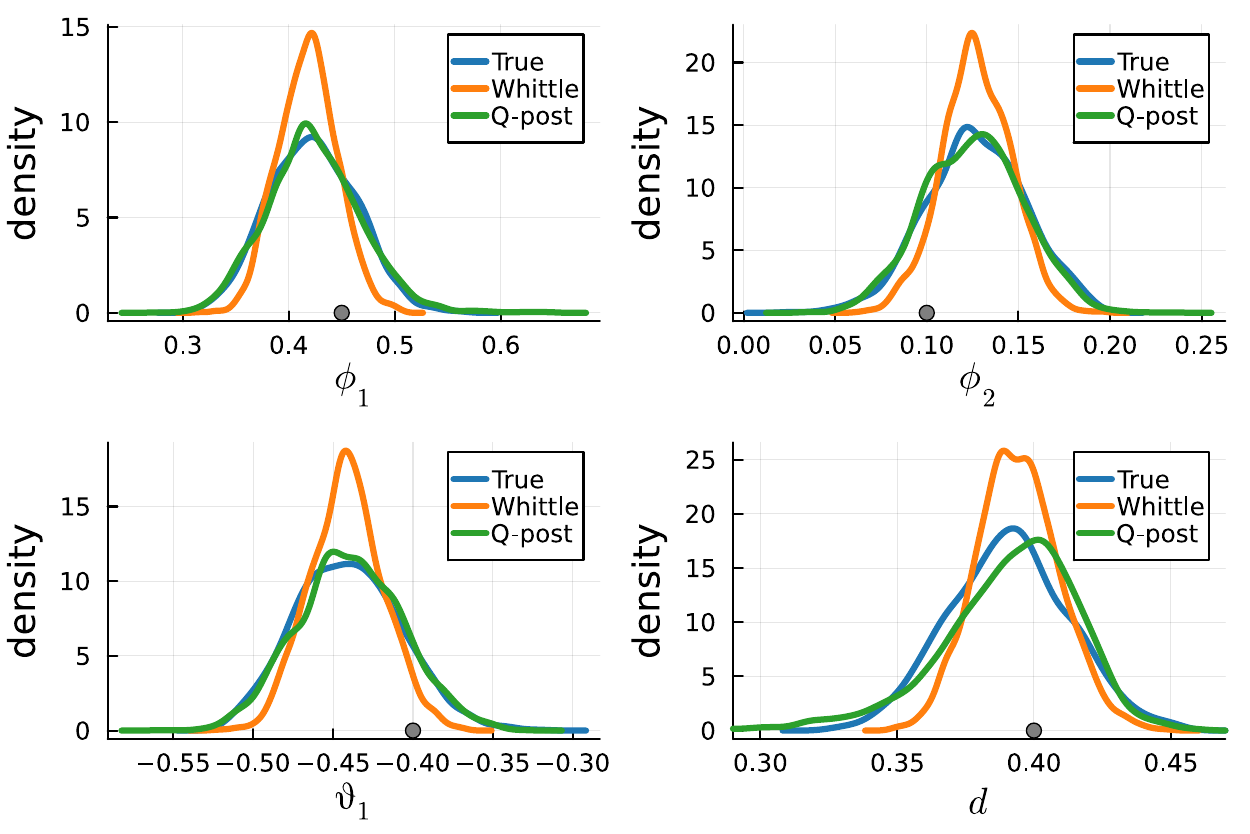}
		\caption{Univariate densities estimates of approximations to the posterior distribution for a single dataset generated from an ARFIMA with true parameter value $\theta = (\phi_1, \phi_2, \vartheta_1, d)^\top = (0.45, 0.1, -0.4, 0.4)^\top$.  Shown are posterior approximations based on the Whittle likelihood (orange), ACP (green) and the exact likelihood (blue).}
		\label{fig:arfima_posteriors}
		
	\end{figure}

	On a laptop computer with an Intel i7-12800H processor (3.70 GHz clock speed) and 32GB RAM the approximate timing for a single likelihood evaluation with the Whittle likelihood, Q-score and true likelihood are 0.001, 0.01 and 0.17 seconds.  The true likelihood uses the \texttt{arfima} package in \texttt{R}.

	\section{Proofs of Main Results}
	This section contains proofs of the stated results in the main paper. We first establish the following additional notation used throughout the remainder.  For two (possibly random) sequences $a_n,b_n$ we say that  $a_n\lesssim b_n$ if for some $n'$ large enough, and all $n\ge n'$, there exists a $C>0$ such that $a_n\le C b_n$ (almost surely); while we write  $a_n\asymp b_n$ if $a_n\lesssim b_n$ and $b_n\lesssim a_n$ (almost surely). Throughout, $\|\cdot\|$ denotes the Euclidean norm for vectors or a convenient matrix norm for matrices.
	
	\subsection{Lemmas }\label{app:Lemmas}
	
	The following lemmas are used to prove the main results. To simplify the statement and derivation of results, redefine $Q_n(\theta)$ as 
	$$
	Q_n(\theta)=-\frac{n}{2}\overline{m}_n(\theta)^\top W_n(\theta)^{-1}\overline{m}_n(\theta). 
	$$While this definition is not as easily interpretable as the $Q_n(\theta)$ defined in equation \eqref{eq:genpost}, it simplifies our manipulations. 
	
	\begin{lemma}\label{lemma:freq}
		Under Assumptions \ref{ass:infeasible}-\ref{ass:cons}, the following are satisfied for some $\theta_\star\in\Theta_\star$, and $\theta_n\in\arg_{\theta\in\Theta}\{0=\overline{m}_n(\theta)\}$. 
		
		\noindent1. $\|\theta_n-\theta_\star\|=O_p(n^{-1/2})$.  
		
		\noindent2. For any $\delta_n=o(1)$, $T_n=\{\theta\in\Theta,\theta_\star\in\Theta_\star:\|\theta-\theta_\star\|\leq \delta_n\}$, and $t={\sqrt{n^{}}}(\theta-\theta_\star),\; \theta\in T_n$,
		$$
		Q_n(\theta_\star+t/{\sqrt{n^{}}})-Q_n(\theta_\star)=t^{\top}\mathcal{H}_n(\theta_\star)W_n(\theta_\star)^{-1}\sqrt{n}\overline{m}_n(\theta_\star)-\frac{1}{2}t^{\top}\Delta(\theta_\star)t+R_n(\theta_\star+t/{\sqrt{n^{}}}), 
		$$for a remainder term $R_n(\theta)$ satisfying 
		$$
		\sup_{\|\theta-\theta_\star\|\le\delta_n}\frac{R_n(\theta)}{1+\|\sqrt{n}(\theta-\theta_\star)\|^2}=o_p(1).
		$$
	\end{lemma}
	\begin{proof}

		\noindent\textbf{Proof of 1.}
		From the consistency of $\theta_n$ there exists some positive $\delta_n=o(1)$	such that $P_{0}^{(n)}\left\{\|\theta_n-\theta_\star\|\ge\delta_n\right\}=o(1)$. By Assumption \ref{ass:infeasible}(ii), we have  
		\begin{flalign}\label{eq:res1}
			\sup_{\|\theta-\theta_\star\|\le\delta_n}\|\overline{m}_n(\theta)-\overline{m}(\theta)-\overline{m}_n(\theta_\star)\|=o_p(n^{-1/2})
		\end{flalign}
		under $P^{(n)}_0$. Then,	with $P^{(n)}_0$ - probability converging to one for the sequence $\delta_n$, we have
		\begin{flalign*}
			\|\overline{m}_n(\theta)-m(\theta_n)-\overline{m}_n(\theta_\star)\|&\leq o_p(n^{-1/2})\\\|\overline{m}_n(\theta)-m(\theta_n)-\overline{m}_n(\theta_\star)\|&\geq \|m(\theta_n)\|-\|\overline{m}_n(\theta)\|-\|\overline{m}_n(\theta_\star)\|.
		\end{flalign*}Rearranging terms, and applying Assumption \ref{ass:infeasible}(iv),
		\begin{flalign*}
			\|m(\theta_n)\|&\leq o_p(n^{-1/2})+\|\overline{m}_n(\theta_\star)\|\{1+o_p(1)\}=O_p(n^{-1/2}).
		\end{flalign*}Because $m(\theta)$  differentiable, $\mathcal{H}_\star=\mathcal{H}(\theta_\star)$ is positive-definite - Assumption \ref{ass:infeasible}(iii) - there exists $C>0$ such that $
		C\|\theta_n-\theta_\star\|\leq	\|m(\theta_n)\|\leq O_p(n^{-1/2}).
		$	\bigskip 
		
		\noindent\textbf{Proof of 2.} On the set $T_n$, the result follows from the following expansion of $Q_n(\theta)=-\frac{n}{2}\overline{m}_n(\theta)^\top W_n(\theta)^{-1}\overline{m}_n(\theta)$ around $\theta_\star$:
		\begin{flalign*}
			Q_n(\theta)-Q_n(\theta_\star)={\sqrt{n^{}}}(\theta-\theta_\star)'\mathcal{H}(\theta_\star)^\top W(\theta_\star)^{-1}{\sqrt{n^{}}}\overline{m}_n(\theta_\star)-{\sqrt{n^{}}}(\theta-\theta_\star)'\Delta(\theta_\star){\sqrt{n^{}}}(\theta-\theta_\star)/2+R_n(\theta);
		\end{flalign*} where $R_n(\theta)=R_{1n}(\theta)+R_{3n}(\theta)+R_{3n}(\theta)$, and
		\begin{flalign*}
			R_{1 n}(\theta) =&-n\left\{\overline{m}_n\left(\theta_\star\right)^{\top} W(\theta)^{-1} \mathcal{H}\left(\theta_\star\right)\left(\theta-\theta_\star\right)-\frac{1}{2}\left(\theta-\theta_\star\right)^{\top} \mathcal{H}\left(\theta_\star\right)^{\prime} W(\theta)^{-1} \mathcal{H}\left(\theta_\star\right)\left(\theta-\theta_\star\right)\right\}
			\\
			&+\left\{\frac{1}{2} \overline{m}_n\left(\theta\right)^\top W_n(\theta)^{-1}\overline{m}_n(\theta)+\frac{1}{2} \overline{m}_n\left(\theta_\star\right)^{\top} W_n(\theta)^{-1} m\left(\theta_\star\right)\right\} \\
			R_{2 n}(\theta)=&-n\left [\frac{1}{2} \overline{m}_n\left(\theta_\star\right)^{\top}\left\{W_n\left(\theta_\star\right)^{-1}-W_n(\theta)^{-1}\right\} \overline{m}_n\left(\theta_\star\right)\right] \\
			R_{3 n}(\theta) =&n\left[\overline{m}_n\left(\theta_\star\right)^{\top}\left\{W(\theta)^{-1}-W\left(\theta_\star\right)^{-1}\right\} \mathcal{H}\left(\theta_\star\right)\left(\theta-\theta_\star\right)\right]\\
			&-\left[\frac{1}{2}\left(\theta-\theta_\star\right)^{\prime} \mathcal{H}\left(\theta_\star\right)^{\prime}\left\{W(\theta)^{-1}-W\left(\theta_\star\right)^{-1}\right\} \mathcal{H}\left(\theta_\star\right)\left(\theta-\theta_\star\right)\right].
		\end{flalign*}
		The stated condition on $R_n(\theta)$ then directly follows by Lemma \ref{lem:remain}. 
	\end{proof}

	\begin{lemma}\label{bound} Under Assumptions \ref{ass:infeasible}-\ref{ass:cons}, with probability converging to one under $P^{(n)}_0$,
		$$
		0\le {\pi(\theta_n\mid \mathcal{Q}_n)}/{{n^{{d_\theta}/2}}}\le {1}/\{{{(2\pi)^{d_\theta}|\Delta|^{}}}\}^{1/2}.
		$$
	\end{lemma}
	\begin{proof}
		The proof proceeds via a similar argument to that used to prove Lemma 1 in \cite{frazier2021synthetic}. Let $\gamma_n=o(1)$ with $\gamma_n{\sqrt{n^{}}}\rightarrow\infty$, and let $\mathcal{N}_{\gamma_n}=\{\theta\in\Theta:\|\theta-\theta_n\|\le\gamma_n\}$. Rewrite the posterior $\pi(\theta\mid \mathcal{Q}_n)$ as 
		\begin{flalign*}
			\pi(\theta\mid \mathcal{Q}_n)&=\pi(\theta_n\mid \mathcal{Q}_n)\frac{M_n(\theta)^{}\pi(\theta)}{M_n(\theta_n)\pi(\theta_n)}\exp\{Q_n(\theta)-Q_n(\theta_n)\}\\&=\pi(\theta_n\mid \mathcal{Q}_n)\exp\{Q_n(\theta)-Q_n(\theta_n)\}+\pi(\theta_n\mid \mathcal{Q}_n)\left\{\frac{M_n(\theta)^{}\pi(\theta)}{M_n(\theta_n)\pi(\theta_n)}-1\right\}\exp\{Q_n(\theta)-Q_n(\theta_n)\},
		\end{flalign*}where we remind the reader that $M_n(\theta)=|W_n(\theta)|^{-1/2}$. From Lemma \ref{lemma:freq}(1), for any $\gamma_n=o(1)$ such that $\gamma_n\sqrt{n}\rightarrow+\infty$, and $ \Omega_n:=\|\theta_n-\theta_\star\|\le\gamma_n$ $P^{(n)}_0\left\{\Omega_n\right\}=1+o(1)$. On $\Omega_n$, by Assumptions \ref{ass:weight} and \ref{ass:prior}, and the compactness of $\{\theta\in\Theta:\|\theta-\theta_n\|\le\gamma_n\}$,  	
		$$
		\left\{\frac{M_n(\theta)^{}\pi(\theta)}{M_n(\theta_n)\pi(\theta_n)}-1\right\}\le \sup_{\|\theta-\theta_n\|\le\gamma_n}\left|\left\{\frac{M_n(\theta)^{}\pi(\theta)-M_n(\theta_n)\pi(\theta_n)}{M_n(\theta_n)\pi(\theta_n)}\right\}\right|=o_p(1). 
		$$Hence, over $\mathcal{N}_{\gamma_n}$, 
		\begin{flalign*}
			\pi(\theta\mid \mathcal{Q}_n)&=\pi(\theta_n\mid \mathcal{Q}_n)\{1+o_p(1)\}\exp\{Q_n(\theta)-Q_n(\theta_n)\},
		\end{flalign*}and 
		\begin{flalign}
			\int_{\mathcal{N}_{\gamma_n}}\pi(\theta\mid \mathcal{Q}_n)\dt \theta&=\pi(\theta_n\mid \mathcal{Q}_n)\{1+o_p(1)\}\int_{\mathcal{N}_{\gamma_n}}\exp\left\{Q_n(\theta)-Q_n(\theta_n)\right\}\dt \theta\nonumber\\&=\frac{\pi(\theta_n\mid \mathcal{Q}_n)}{{n^{d_{\theta}/2}}}\{1+o_p(1)\}\int_{\|t\|\le\gamma_n{\sqrt{n^{}}}}\exp\{Q_n(\theta_n+t/{\sqrt{n^{}}})-Q_n(\theta_n)\}\dt t,\label{eq:simple_post}
		\end{flalign}where the second line follows from the change of variables $t=\sqrt{n}(\theta-\theta_n)$. 
		
		From the expansion in Lemma \ref{lemma:freq}, over the set $\mathcal{N}_{\gamma_n}$, for $\Delta(\theta)=\mathcal{H}(\theta)\mathcal{I}(\theta)^{-1}\mathcal{H}(\theta)$, 
		\begin{flalign*}
			Q_n(\theta_n+t/{\sqrt{n^{}}})-Q_n(\theta_n)&=-\frac{1}{2}t^{\top}\Delta(\theta_n)t+O(\|t\|^2\gamma_n )
			=-\frac{1}{2}t^{\top}\Delta(\theta_n)t+o(1).
		\end{flalign*} Rewrite the above as
		\begin{flalign}
			Q_n(\theta_n+t/{\sqrt{n^{}}})-Q_n(\theta_n)&=-\frac{1}{2}t^{\top}\left(I+V_n\right)\Delta(\theta_\star)^{}t+o(1),\label{eq1:rep1}
		\end{flalign}where $V_n=
		\left[\Delta(\theta_n)-\Delta(\theta_\star)^{}\right]\Delta(\theta_\star)^{-1}$. By Assumptions \ref{ass:infeasible}(iii) and \ref{ass:weight}, for some $C>0$,
		$
		\|V_n\|\le C\|\Delta(\theta_\star)^{-1}\|\|\theta_n-\theta_{\star}\|.
		$
		Define $A_n=C \Delta(\theta_\star)^{-1} \|\theta_n-\theta_{\star}\|$, and conclude that $A_n$ is positive semi-definite with maximal eigenvalue
		$\lambda_{\text{max}}(A_n)=C\|\theta_n-\theta_\star\|\lambda_{\text{min}}\{\Delta(\theta_\star)\}\ge0$. Therefore, over $\|t\|\le\gamma_n$, 
		\begin{equation*}
			-t^{\top}A_nt\le t^{\top}V_nt\le t^{\top}A_nt;
		\end{equation*}applying the above into \eqref{eq1:rep1} yields 
		\begin{flalign*}
			-t^{\top}\Delta(\theta_\star)t/2-t^{\top}A_n\Delta(\theta_\star)t/2
			\leq Q_n(\theta_n+t/{\sqrt{n^{}}})-Q_n(\theta_n)&\leq -t^{\top}\Delta(\theta_\star)^{}t/2+t^{\top}A_n\Delta(\theta_\star)^{}t/2.
		\end{flalign*}
		
		Let $M^{\pm}_n=I\pm\Delta(\theta_\star)^{-1/2}A_n\Delta(\theta_\star)^{1/2}$, and rewrite the above as
		\begin{flalign*}
			-\frac{1}{2}t^{\top}\Delta(\theta_\star)^{1/2}M^+_n\Delta(\theta_\star)^{1/2}t\le Q_n(\theta_n+t/{\sqrt{n^{}}})-Q_n(\theta_n)&\leq -\frac{1}{2}t^{\top}\Delta(\theta_\star)^{1/2}M^-_n\Delta(\theta_\star)^{1/2}t.
		\end{flalign*}
		For $\|\theta_n-\theta_\star\|$ small enough, i.e., $n$ large enough, $\Delta(\theta_\star)^{1/2}M^\pm_n\Delta(\theta_\star)^{1/2}$ is positive-definite (with probability converging to one). Thus, for $n$ large we can bound the posterior probability over $\mathcal{N}_{\gamma_n}$ as
		\begin{flalign*}
			\int_{\mathcal{N}_{\gamma_n}}\pi(\theta\mid \mathcal{Q}_n)\dt \theta
			&\leq\frac{\pi(\theta_n\mid \mathcal{Q}_n)}{{n^{{d_\theta}/2}}}\left|\Delta(\theta_\star)^{1/2}M_n^-\Delta(\theta_\star)^{1/2}\right|^{-1/2} \int_{T^-_n}\exp(-x^{\top}x/2)\dt x \\
			\int_{\mathcal{N}_{\gamma_n}}\pi(\theta\mid \mathcal{Q}_n)\dt \theta&\geq\frac{\pi(\theta_n\mid \mathcal{Q}_n)}{{n^{{d_\theta}/2}}}\left|\Delta(\theta_\star)^{1/2}M_n^+\Delta(\theta_\star)^{1/2}\right|^{-1/2}\int_{T^+_n}\exp(-x^{\top}x/2)\dt x
		\end{flalign*}where $$T^{-}_n=\left\{x:\|x\|\leq \frac{\gamma_n{\sqrt{n^{}}}}{\lambda_{\text{min}}\{\Delta(\theta_\star)^{1/2}M_n^-\Delta(\theta_\star)^{1/2}\}^{1/2}}\right\},\; T^{+}_n=\left\{x:\|x\|\leq \frac{\gamma_n{\sqrt{n^{}}}}{\lambda_{\text{max}}\{\Delta(\theta_\star)^{1/2}M_n^+\Delta(\theta_\star)^{1/2}\}^{1/2}}\right\},$$	and
		where we have used the fact that for any positive semi-definite matrix $M$ and $\gamma>0$
		$$
		\{x:\|x\|\le\gamma/{\lambda_{\text{max}}(M)}^{1/2}\}\subseteq\{x:x^\top Mx\le\gamma\}\subseteq	\{x:\|x\|\le\gamma/{\lambda_{\text{min}}(M)}^{1/2}\}.
		$$
		Under the restriction that $\gamma_n {\sqrt{n^{}}}\rightarrow\infty$, $T_n^{+}$ and $T_n^{-}$ both converge to $\mathbb{R}^{d_\theta}$ and $\int_{T^{\pm}_n}\exp(-x^\top x/2)\dt x\rightarrow {(2\pi)}^{d_\theta/2}$ as $n\rightarrow\infty$. Hence, with probability converging to one,
		\begin{flalign*}
			|M_n^+|^{1/2}\int_{\mathcal{N}_{\gamma_n}}\pi_n(\theta\mid \mathcal{Q}_n)\dt \theta\le	\frac{\pi(\theta_n\mid \mathcal{Q}_n)}{{n^{{d_\theta}/2}}}(2\pi)^{d_\theta/2}|\Delta^{-1}|^{1/2}&\le |M_n^-|^{1/2}\int_{\mathcal{N}_{\gamma_n}}\pi(\theta\mid \mathcal{Q}_n)\dt \theta.
		\end{flalign*}Since $|M_n^\pm|\rightarrow1$, $|\Delta|>0$ and $0\le\int_{\mathcal{N}_{\gamma_n}}\pi(\theta\mid \mathcal{Q}_n)\dt \theta\le1$, with probability converging to one, 
		$$
		0\le\frac{\pi(\theta_n\mid \mathcal{Q}_n)}{{n^{{d_\theta}/2}}}\le 1/\{(2\pi)^{d_\theta}|\Delta^{-1}|\}^{1/2}.
		$$
	\end{proof}
	
	The following result is used in the proof of Lemma \ref{thm:two} and is a consequence of Proposition 1 in \cite{chernozhukov2003mcmc}.
	\begin{lemma}\label{lem:remain} Under the assumptions of Lemma \ref{lemma:freq}, and for $R_n(\theta)$ as defined in the proof of Lemma \ref{lemma:freq}, for each $\epsilon>0$ there exists a sufficiently small $\delta>0$ and $h>0$ large enough, such that
		$$
		\limsup_{n\rightarrow+\infty}P^{(n)}_0\left[\sup_{h/\sqrt{n}\le \|\theta-\theta_0\|\le\delta}\frac{|R_n(\theta)|}{1+{n}\|\theta-\theta_0\|^2}>\epsilon\right]<\epsilon
		$$	and
		$$
		\limsup_{n\rightarrow+\infty}P^{(n)}_0\left[\sup_{ \|\theta-\theta_0\|\le h/\sqrt{n}}{|R_n(\theta)|}>\epsilon\right]=0.
		$$	
	\end{lemma}
	\begin{proof}[Proof]
		The result is a specific case of Proposition~1 in \cite{chernozhukov2003mcmc}. Therefore, it is
		only necessary to verify that their sufficient conditions are satisfied in our context.
		
		Assumptions (i)-(iii) in their result follow
		directly from Assumptions \ref{ass:infeasible} and \ref{ass:weight}, and the normality of $\sqrt{n}\overline{m}_n(\theta_\star)$ in 
		Assumption~\ref{ass:infeasible}. Their Assumption~(iv) is stated as follows: for any $\epsilon>0$, there is a $\delta>0$ such that
		$$
		\limsup_{n\rightarrow+\infty}P^{(n)}_0\left\{\sup_{\|\theta-\theta'\|\le\delta}\frac{\sqrt{n}\|
			\{\overline{m}_n(\theta)-\overline{m}_n(\theta')\}-\{\mathbb{E}\left[\overline{m}_n(\theta)\right]-\mathbb{E}
			\left[\overline{m}_n(\theta')\right]\}\|}{1+n\|\theta-\theta'\|}>\epsilon\right\}<\epsilon . 
		$$
		In our context, this condition is satisfied by Assumption \ref{ass:infeasible}(iv). Hence, the result follows. 
	\end{proof}

	\subsection{Proofs of the Main Results}
	We first prove Theorem \ref{thm:multi} as it simplifies the proof of Theorem  \ref{thm:two}.
	\begin{proof}[Proof of Theorem \ref{thm:multi}]
		The proof of Theorem \ref{thm:multi} follows similar arguments to that of Theorem 1 in \cite{frazier2021synthetic}. Let  $C_\star\in[0,1]$ be such that $\lim_n n^{-d_\theta/2}{\pi(\theta_n\mid\mathcal{Q}_n)}(2\pi)^{d_\theta/2}|\Delta(\theta_\star)^{-1}|^{1/2}=C_\star$,  and define $\Pi_\star(\gamma):=C_\star\int_{\|t\|\le\gamma}\exp\left(-t^\top \Delta(\theta_\star)^{}t/2\right)\dt t$. The result follows if $J_n=|\int_{\|t\|\le \gamma}{\pi}(t\mid \mathcal{Q}_n)\dt t-\Pi_\star(\gamma)|=o_p(1).$
		
		Rewrite the exact posterior $\pi(\theta\mid \mathcal{Q}_n)$ as
		\begin{flalign}
			\pi(\theta\mid \mathcal{Q}_n)&= \pi(\theta_n\mid \mathcal{Q}_n)\frac{M_n(\theta)\pi(\theta)}{M_n(\theta_n)\pi(\theta_n)}\exp\left\{Q_n(\theta)-Q_n(\theta_n)\right\},\label{eq:Clev_post}
		\end{flalign}for any $\theta_n$ as in Lemma \ref{lemma:freq}. 
		For any $\gamma>0$, let $\gamma_n=\gamma/{\sqrt{n^{}}}$, with $\gamma_n=o(1)$, and define $\mathcal{N}_{\gamma_n}=\{\theta\in\Theta:\|\theta-\theta_n\|\le\gamma_n\}$. Plugging equation \eqref{eq:simple_post} in the proof of Lemma \ref{bound} into equation \eqref{eq:Clev_post}, we have
		\begin{equation}
			\int_{\mathcal{N}_{\gamma_n}}\pi(\theta\mid \mathcal{Q}_n)\dt\theta=\pi(\theta_n\mid \mathcal{Q}_n)\{1+o_p(1)\}\int_{\mathcal{N}_{\gamma_n}}\exp\left\{Q_n(\theta)-Q_n(\theta_n)\right\}\dt\theta\label{eq:approx1}.
		\end{equation}
		
		Now, consider the change of variables $\theta\mapsto t={\sqrt{n^{}}}(\theta-\theta_n)$, and the posterior $\pi(t\mid S_ n)=\pi(\theta_n+t/{\sqrt{n^{}}}\mid \mathcal{Q}_n)/{n^{{d_\theta}/2}}$. Noting that $\mathcal{N}_{\gamma_n}$ can be written as $\{t:\|t\|\le\gamma\}$, we have
		\begin{flalign}
			\int_{\mathcal{N}_{\gamma_n}}\pi(\theta\mid \mathcal{Q}_n)\dt\theta&=\int_{\|t\|\le\gamma}\pi(t\mid \mathcal{Q}_n)\dt t\nonumber\\&=\frac{\pi(\theta_n\mid \mathcal{Q}_n)}{{n^{{d_\theta}/2}}}\{1+o_p(1)\}\int_{\|t\|\le\gamma}\exp\left\{Q_n(\theta_n+t/{\sqrt{n^{}}})-Q_n(\theta_n)\right\}\dt t\nonumber\\&\le C_{\star}\{1+o_p(1)\}\int_{\|t\|\le\gamma}\exp\left\{Q_n(\theta_n+t/{\sqrt{n^{}}})-Q_n(\theta_n)\right\}\dt t\label{eq:xxx};
		\end{flalign} the second equality follows from equation \eqref{eq:approx1} and the change of variables, and the last equation follows by Lemma \ref{bound} (with $P^{(n)}_0$-probability converging to one).   
		Applying equation \eqref{eq:xxx}, we see that, up to an $o_p(1)$ term,
		\begin{flalign*}
			J_{n}&\lesssim \left|\int_{\|t\|\le\gamma}\left[\exp\left\{Q_n(\theta_n+{t}/{{\sqrt{n^{}}}})-Q_n(\theta_n)\right\}-\exp\left\{-t^{\top}\Delta(\theta_\star)^{}t/2\right\}\right]\dt t\right|;
		\end{flalign*} the result follows if the right hand side of the above converges to zero in $P^{(n)}_0$-probability.

		Using the same argument as in the proof of Lemma \ref{bound}, for $A_n=C \Delta(\theta_\star)^{-1} \|\theta_n-\theta_{\star}\|$, for some $C>0$, we have  
		\begin{flalign*}
			-t^{\top}\Delta(\theta_\star)^{}t/2-t^{\top}A_n\Delta(\theta_\star)^{}t/2
			\leq Q_n(\theta_n+t/{\sqrt{n^{}}})-Q_n(\theta_n)&\leq -t^{\top}\Delta(\theta_\star)^{}t/2+t^{\top}A_n\Delta(\theta_\star)^{}t/2.
		\end{flalign*}
		However, using the definition of $A_n$, over $\|t\|\le\gamma$, the above equation simplifies to
		\begin{flalign}\label{eq:bound}
			-t^{\top}\Delta(\theta_\star)^{}t/2-C\|\theta_n-\theta_\star\|
			\leq Q_n(\theta_n+t/{\sqrt{n^{}}})-Q_n(\theta_n)&\leq -t^{\top}\Delta(\theta_\star)^{}t/2+C\|\theta_n-\theta_\star\|.
		\end{flalign}
		
		Using equations \eqref{eq:xxx} and \eqref{eq:bound}, and for $n$ large enough, we can bound the posterior over $\mathcal{N}_{\gamma}$ above and below as
		\begin{flalign*}
			&\frac{\pi(\theta_n\mid \mathcal{Q}_n)}{{n^{{d_\theta}/2}}}\int_{\mathcal{N}_{\gamma}}\exp\left\{Q_n(\theta_n+t/{\sqrt{n^{}}})-Q_n(\theta_n)\right\}\dt t \\&\leq\frac{\pi(\theta_n\mid \mathcal{Q}_n)}{{n^{{d_\theta}/2}}} \exp(C\|\theta_n-\theta_\star\|)\int_{\mathcal{N}_{\gamma_n}}\exp\left\{-t^{\top}\Delta(\theta_\star)^{}t/2\right\}\dt t\\&\geq \frac{\pi(\theta_n\mid \mathcal{Q}_n)}{{n^{{d_\theta}/2}}}\exp(-C\|\theta_n-\theta_\star\|)\int_{\mathcal{N}_{\gamma_n}}\exp\left\{-t^{\top}\Delta(\theta_\star)^{}t/2\right\}\dt t.
		\end{flalign*}
		As $n\rightarrow\infty$, it follows from Lemma \ref{bound}, and the dominated convergence theorem  that
		\begin{flalign*}
			\Pi_\star(\gamma)\{1+o_p(1)\}&\le \frac{\pi(\theta_n\mid \mathcal{Q}_n)}{{n^{{d_\theta}/2}}}\int_{\mathcal{N}_{\gamma}}\exp\left\{Q_n(\theta_n+t/{\sqrt{n^{}}})-Q_n(\theta_n)\right\}\dt t
			\\
			\Pi_\star(\gamma)\{1+o_p(1)\}&\ge \frac{\pi(\theta_n\mid \mathcal{Q}_n)}{{n^{{d_\theta}/2}}}\int_{\mathcal{N}_{\gamma}}\exp\left\{Q_n(\theta_n+t/{\sqrt{n^{}}})-Q_n(\theta_n)\right\}\dt t.	
		\end{flalign*}Hence, $J_n\rightarrow0$ in $P^{(n)}_0$-probability.

	\end{proof}
	
	\begin{proof}[Proof of Theorem \ref{thm:two}] Since $\theta_\star$ is assumed to be unique, we write $\Delta=\Delta(\theta_\star)$.  For some $\delta_n$ such that $\delta_n=o(1)$ and ${\sqrt{n^{}}}\delta_n\rightarrow\infty$,  split the region of integration into $\mathcal{T}_n=\mathcal{T}_{1n}\cup \mathcal{T}_{2n}$, where $\mathcal{T}_{1n}=\{h\le\|t\|\le \delta_n{\sqrt{n^{}}}\}$ and $\mathcal{T}_{2n}=\{\|t\|\ge\delta_n{\sqrt{n^{}}}\}$, for some arbitrary $h>0$. We consider the integral over each region separately. 
		
		\noindent\textbf{Region $\mathcal{T}_{1n}$.}  Rewrite the  posterior as	
		\begin{flalign*}
			\pi(\theta\mid \mathcal{Q}_n)=\frac{\pi(\theta_n\mid \mathcal{Q}_n)|M_n(\theta_n)|^{1/2}}{\pi(\theta_n)}|M_n(\theta)|^{-1/2}\exp\{Q_n(\theta)-Q_n(\theta_n)\}\pi(\theta).
		\end{flalign*}Note that, from the definition of $\theta_n$ and Assumption \ref{ass:infeasible}(iv), $\|\sqrt{n}(\theta_n-\theta_\star)\|=O_p(1)$ by Lemma \ref{lemma:freq}. Under the uniqueness of $\theta_\star$,   Lemma \ref{bound} implies
		\begin{flalign}\label{eq:star1}
			\pi(\theta_n\mid \mathcal{Q}_n)/{n^{{d_\theta}/2}}=C_{\star}+o_p(1)=C_{\star}=1/\{(2\pi)^{d_\theta}|\Delta^{-1}|\}^{1/2}+o_p(1).
		\end{flalign}Over $\{\|t\|\le \delta_n{\sqrt{n^{}}}\}$, equation \eqref{eq:star1} and similar arguments to those used in the proof of Theorem \ref{thm:multi} to obtain equation \eqref{eq:xxx}
		yield
		\begin{flalign*}
			\int_{\|\theta-\theta_n\|\le\delta_n}\pi(\theta\mid \mathcal{Q}_n)\dt\theta&=\frac{\pi(\theta_n\mid \mathcal{Q}_n)}{{n^{{d_\theta}/2}}}\{1+o_p(1)\}\int_{\|t\|\le\delta_n{\sqrt{n^{}}}}\pi(t\mid \mathcal{Q}_n)\dt t\nonumber\\&=C_{\star}\{1+o_p(1)\}\int_{\mathcal{T}_{1n}}\exp\{Q_n(\theta_n+t/{\sqrt{n^{}}})-Q_n(\theta_n)\}\dt t,
		\end{flalign*} and, up to an $o_p(1)$ term,
		\begin{flalign}
			&\int_{\mathcal{T}_{1n}} \|t\||\pi(t\mid \mathcal{Q}_n)-N\{t;0,\Delta^{-1}\}|\dt t\nonumber\\&=C_{\star}\int_{\mathcal{T}_{1n}}\|t\||\exp\{Q_n(\theta_n+{t}/{{\sqrt{n^{}}}})-Q_n(\theta_n)\}-\exp(-t^\top\Delta^{}t/2)|\dt t.\label{eq:star4}
		\end{flalign}Over $\mathcal{T}_{1n}$, arguments similar to those used in Theorem \ref{thm:two} to obtain equation \eqref{eq1:rep1} yield
		\begin{equation}
			Q_n(\theta_n+t/{\sqrt{n^{}}})-Q_n(\theta_n)\leq -\frac{1}{2}t^{\top}\left(I-A_n\right)\Delta^{} t+o_p(1)\label{eq:star3},
		\end{equation}
		for $A_n$ as defined in that proof.
		
		For some $0<h<\infty$, further split $\mathcal{T}_{1n}=\{\|t\|\le h\}\cup\{h\le \|t\|\le\delta_n{\sqrt{n^{}}}\}$, and consider first the set $\{\|t\|\le h\}$. Now, recall that $A_n\rightarrow0$ as $n\rightarrow\infty$, so that  over $\|t\|\le h$ equation \eqref{eq:star3} implies
		$$
		Q_n(\theta_n+t/{\sqrt{n^{}}})-Q_n(\theta_n)=-\frac{1}{2}t^\top\Delta^{}t+o_p(1),
		$$and it follows that the integral in \eqref{eq:star4} is $o_p(1)$ over $\|t\|\le h$. 
		
		Over $\{h\le \|t\|\le \delta_n{\sqrt{n^{}}}\}$, 
		$$
		\int_{h<\|t\|\le {\sqrt{n^{}}} \delta_n}\|t\| N\{t;0,\Delta^{-1}\}\dt t
		$$ can be made arbitrarily small by taking $h$ large enough and $\delta_n$ small enough, so that, applying \eqref{eq:star3}, it suffices to show that for any $\varepsilon>0$ there exists an $h$ and $\delta_n$ such that, for some $n$ large enough,  
		$$
		P^{(n)}_0\left[\int_{h<\|t\|\le {\sqrt{n^{}}} \delta_n}\|t\| \exp\{-t^{\top}\left(I-A_n\right)\Delta^{} t/2\}\dt t<\varepsilon\right]\ge 1-\varepsilon.
		$$
		
		However, for $h'$ large enough, and all $h>h'$, on the set $h<\|t\|\le \delta_n{\sqrt{n^{}}}$, 
		$$
		\|t\| \exp(-t^\top \Delta^{}t/2)=O(1/h). 
		$$Therefore, for any $\varepsilon>0$, there is an $h$ large enough and a $\delta_n$ small enough such that
		$$
		\int_{h<\|t\|\le {\sqrt{n^{}}} \delta_n}\|t\|\exp(-t^{\top}\Delta^{}t/2)\dt t<\varepsilon.
		$$Since $A_n\rightarrow 0$, we can conclude that for some $n$ large enough, with $P^{(n)}_0$-probability at least $1-\varepsilon$,
		$$
		\int_{M<\|t\|\le {\sqrt{n^{}}} \delta}\|t\| \exp\{-t^{\top}\left(I-A_n\right)\Delta^{}t/2\}\dt t<\varepsilon.
		$$
		
		\noindent\textbf{Region $\mathcal{T}_{2n}$.} Again, 
		$
		\int_{\mathcal{T}_{2n}}\|t\|N\{t;0,\Delta^{-1}\}\dt t 
		$ can be made arbitrarily small by taking $\delta_n {\sqrt{n^{}}}$ large enough, and it remains to show that $\int_{\mathcal{T}_{2n}}\|t\|\pi(t\mid \mathcal{Q}_n)=o_p(1)$. 
		
		Applying, in-turn, the expression for the exact posterior in \eqref{eq:Clev_post}, Assumptions \ref{ass:weight}, \ref{ass:prior}, and Lemma \ref{bound}, 
		\begin{flalign*}
			&\int_{\mathcal{T}_{2n}} \|t\|\pi(t\mid \mathcal{Q}_n)\dt t\le \\&C \int_{\mathcal{T}_{2n}}{\|t\|M_n(\theta_n+t/n^{1/2})\pi(\theta_n+t/{\sqrt{n^{}}})}\exp\{Q_n(\theta_n+t/{\sqrt{n^{}}})-Q_n(\theta_n)\}\dt t,
		\end{flalign*}for some $C>0$ with probability converging to one. Using the change of variables $t\mapsto\theta$, the integral on the right hand side of the inequality is bounded by
		\begin{flalign}\label{eq:terms1}
			C\left\{1+o_p(1)\right\}{n^{(d_{\theta}+2)/2}}\int_{\|\theta-\theta_\star\|>\delta}\|\theta-\theta_\star\|\pi(\theta)M_n(\theta)\exp\{Q_n(\theta)-Q_n(\theta_n)\}\dt\theta,
		\end{flalign}where the $o_p(1)$ term follows from the triangle inequality and consistency of $\theta_n$ for $\theta_\star$ (Lemma \ref{lemma:freq}). 
		For any $\delta>0$, and $Q(\theta)=-m(\theta)^\top W(\theta)^{-1}m(\theta)/2$, under Assumptions \ref{ass:infeasible} and \ref{ass:weight},
		\begin{flalign*}
			\sup_{\|\theta-\theta_\star\|>\delta}n^{-1}\left\{Q_n(\theta)-Q_n(\theta_n)\right\}&
			\leq \sup_{\|\theta-\theta_\star\|>\delta}\{Q(\theta)-Q(\theta_\star)\}+o_p(1).
		\end{flalign*}
		From Assumption \ref{ass:weight}(iii),  for any $\delta>0$ there exists some $\epsilon>0$ such that 
		$
		\sup_{\|\theta-\theta_\star\|>\delta}\{Q(\theta)-Q(\theta_\star)\}\le -\epsilon. 
		$ Therefore, for any $\delta>0$, 
		$$
		\lim_{n\rightarrow\infty} P^{(n)}_0\left[\sup_{\|\theta-\theta_\star\|>\delta}\exp\{Q_n(\theta)-Q_n(\theta_n)\}\le \exp(-\epsilon n)\right]=1. 
		$$	Use the definition $M_n(\theta)=| W_n(\theta)^{-1}|^{1/2}$, and the above to  obtain
		\begin{align*}
			\int_{\mathcal{T}_{2n}} \|t\|\pi(t\mid \mathcal{Q}_n)\dt t & \le C\{1+o_p(1)\}n^{(d_\theta+2)/2}\int_{\|\theta-\theta_{0}\|\ge \delta }M_n(\theta)\|\theta-\theta_\star\|^{}\pi\left(\theta\right)\exp\{Q_n(\theta)-Q_n(\theta_\star)\}\dt \theta\\&\leq C \exp\left(-\epsilon n\right)n^{(d_\theta+2)/2}\int_{\|\theta-\theta_\star\|\ge \delta }M_n(\theta)\|\theta-\theta_\star\|^{}\pi\left(\theta\right)\dt \theta\\&\leq C\left\{\exp\left(-\epsilon n\right)n^{(d_\theta+2)/2}\right\}
			\\&=o_p(1), 
		\end{align*}where 
		the second inequality follows from the moment hypothesis in Assumption~\ref{ass:prior}.
	\end{proof}

	\begin{proof}[Proof of Corollary \ref{corr:one}]
		For $\vartheta=\sqrt{n}(\theta-\theta_n)$, consider the change of variables $\theta=\theta_n+\vartheta/\sqrt{n},$
		\begin{flalign*}
			\overline\theta=\int_\Theta \theta\pi(\theta\mid \mathcal{Q}_n)\dt\theta=\int_{\mathcal{T}_n}(\theta_n+\vartheta/\sqrt{n})\pi(\vartheta\mid \mathcal{Q}_n)\dt\vartheta 
		\end{flalign*}
		so that $$\sqrt{n}(\overline\theta-\theta_n)=\int_{\mathcal{T}_n} \vartheta \pi(\theta\mid \mathcal{Q}_n)\dt \vartheta.$$ When $\theta_\star$ is unique, standard show that $\theta_n:=\theta_\star-\Delta(\theta_\star)^{-1}\mathcal{H}_\star^\top W(\theta_\star)^{-1}\overline{m}_n(\theta_\star)$. Hence, for $Z_n=\Delta^{-1}\mathcal{H}_\star^\top W(\theta_\star)^{-1}\sqrt{n}\overline{m}_n(\theta_\star)$, we have
		$$
		\sqrt{n}(\overline\theta-\theta_n)=\sqrt{n}(\overline\theta-\theta_\star)-Z_n=\int_{\mathcal{T}_n} \vartheta [\pi(\theta\mid \mathcal{Q}_n)-N\{\vartheta;0,\Delta(\theta_\star)^{-1}\}]\dt \vartheta+\int_{\mathcal{T}_n}\vartheta N\{\vartheta;0,\Delta(\theta_\star)^{-1}\}\dt\vartheta .
		$$

		Since $\mathcal{T}_n\rightarrow\mathbb{R}^{d_\theta}$ as $n\rightarrow\infty$, $\int_{\mathcal{T}_n}\vartheta N\{\vartheta;0,\Delta(\theta_\star)^{-1}\}\dt\vartheta=o(1)$, and by Theorem \ref{thm:two},
		$$
		\left\|\int_{\mathcal{T}_n} \vartheta\pi(\vartheta\mid \mathcal{Q}_n)\dt \vartheta\right\|\le \int_{\mathcal{T}_n} \|\vartheta\||\pi(\vartheta\mid \mathcal{Q}_n)-N\{\vartheta;0,\Delta(\theta_\star)^{-1}\}|\dt \vartheta+o(1)=o_p(1).
		$$Therefore, we have that $\|\sqrt{n}(\overline\theta-\theta_\star)-Z_n\|=o_p(1)$. By Assumption \ref{ass:infeasible}, 
		$$
		Z_n\Rightarrow N\{0,\Delta(\theta_\star)^{-1}\mathcal{H}_\star W(\theta_\star)^{-1}\mathcal{I}_\star W(\theta_\star)^{-1}\mathcal{H}_\star\Delta(\theta_\star)^{-1}\}
		$$ and the stated result follows. 
	\end{proof}

	\begin{proof}[Proof of Theorem \ref{thm:new}]
		By Assumption 4, for any $k,k'\in\{1,\dots,K\}$ with $k\ne k'$, there exists an $\epsilon>0$ such that, $\|\theta_{k,\star}-\theta_{k',\star}\|>\epsilon$. Since $\|\theta_{k,n}-\theta_{k,\star}\|=o_p(1)$, there exists an $n_k$ large enough so that, for all $n\ge n_k$, 
		$$
		\left\{\theta:\|\theta-\theta_{k,n}\|\le\epsilon\right\}\cap \left\{\theta:\|\theta-\theta_{k',n}\|\le\epsilon\right\}=\{\emptyset\},
		$$with $P_0^{(n)}$-probability converging to one. Now, for $n_k$ large enough such that $r_{n_k}/\sqrt{n_k} \le \epsilon$, and all $n\ge n_k'$, we have that 
		\begin{flalign*}
			\left\{\theta\in\Theta:\|\sqrt{n}(\theta-\theta_{k,n})\|\le r_n\right\}&\subset \left\{\theta:\|\theta-\theta_{k,n}\|\le\epsilon\right\}
		\end{flalign*}Theorem \ref{thm:multi} and the dominated convergence theorem them imply that 
		$$
		\left|\Pi\left\{\|t_k\|\le r_n\mid\mathcal{Q}_n\right\}-C_{k,\star}\int N\left\{t_k;0,\Delta(\theta_{k,\star})^{-1}\right\}\dt t_k\right|=o_p(1),
		$$which implies that $|\Pi\left\{\|t_k\|\le r_n\mid\mathcal{Q}\right\}-C_{k,\star}|=o_p(1).$
		
		Consequently, for $t_k=\sqrt{n}(\theta-\theta_{k,n})$, and $\gamma>0$, Theorem \ref{thm:multi} and the above imply: 
		$$
		\left|\frac{\Pi\left\{\|t_k\|\le \gamma\mid\mathcal{Q}_n\right\}}{\Pi\left\{\|t_k\|\le r_n\mid\mathcal{Q}_n\right\}}-\int_{\|t_k\|\le\gamma}N\{t_k;0,\Delta(\theta_{k,\star})^{-1}\}\dt t_k\right|=o_p(1).
		$$Hence, for $\alpha\in[0,1]$, and each $\gamma_{k,n}(\alpha)$, we have 
		\begin{flalign}
			\gamma_{k,n}(\alpha)&=\inf\left\{\gamma>0:\frac{\Pi\left\{\|t_k\|\le \gamma\mid\mathcal{Q}_n\right\}}{\Pi\left\{\|t_k\|\le r_n\mid\mathcal{Q}_n\right\}}\ge1-\alpha\right\}\nonumber\\&=\underbrace{\inf\left\{\gamma>0:\int_{\|t_k\|\le\gamma}N\{t_k;0,\Delta(\theta_{k,\star})^{-1}\}\dt t_k\ge 1-\alpha\right\}}_{:=\gamma_k(\alpha)}+o_p(1)\label{eq:gamma}	.
		\end{flalign}
		
		Now, note that 
		$$\theta_{k,\star}\in\left\{\|t_k\|\le \gamma_{k,n}(\alpha)\right\}\iff \|\theta_{k,n}-\theta_{k,\star}\|\le\gamma_{k,n}(\alpha).$$ Let, $Z_k\sim N(0,\Delta(\theta_{k,\star})^{-1})$, and applying Lemma \ref{lemma:freq}.(i) and Assumption 1.(iv) at $\theta_{k,\star}$ to deduce that  
		$
		\sqrt{n}(\theta_{k,n}-\theta_{k,\star})\Rightarrow Z_k
		$; since $\gamma_{k,n}(\alpha)=\gamma_k(\alpha)+o_p(1)$ by \eqref{eq:gamma}, we then see that  
		\begin{flalign}
			P_0^{(n)}\left(\theta_{k,\star}\in\{\|t_k\|\le \gamma_{k,n}\}\right)=&P_0^{(n)}\left(\theta_{k,\star}\in\{\|\sqrt{n}(\theta_{k,n}-\theta_{k,\star})\|\le \gamma_{k,n}\}\right)\nonumber\\=&\text{Pr}\left\{\|Z_k\|\le \gamma_k(\alpha)\right\} +o(1) \nonumber\\=&1-\alpha	+o(1)\label{eq:new_gamma}
		\end{flalign}

		Recalling $t_k=\sqrt{n}(\theta-\theta_{k,n})$, and the definition of $\mathcal{C}^{(n)}_\Pi(1-\alpha)=\cup_{k=1}^{K}\left\{\|t_k\|\le \gamma_{k,n}(\alpha/K)\right\}$, we have via the union bound
		\begin{flalign*}
			P_0^{(n)}	\left\{\Theta_\star\subseteq \mathcal{C}^{(n)}_\Pi(1-\alpha)\right\}&=P_0^{(n)}	\left\{\bigcap_{k=1}^{K}\left\{\theta_{k,\star}\in \left\{\|t_k\|\le \gamma_{k,n}(\alpha/K)\right\}\right\}\right\}\\&=1-P_0^{(n)}	\left\{\bigcup_{k=1}^{K}\left\{\theta_{k,\star}\not\in \left\{\|t_k\|\le \gamma_{k,n}(\alpha/K)\right\}\right\}\right\}
			\\&\ge 1-\sum_{k=1}^{K}P_0^{(n)}	\left\{\theta_{k,\star}\not\in \left\{\|t_k\|\le \gamma_{k,n}(\alpha/K)\right\}\right\}.
		\end{flalign*}From \eqref{eq:new_gamma}, $$\limsup_{n\rightarrow+\infty}P_0^{(n)}	\left\{\theta_{k,\star}\not\in \left\{\|t_k\|\le \gamma_{k,n}(\alpha/K)\right\}\right\}\le \alpha/K,$$ which implies that 
		\begin{flalign*}
			\liminf_{n\rightarrow+\infty}	P_0^{(n)}	\left\{\Theta_\star\subseteq \mathcal{C}^{(n)}_\Pi(1-\alpha)\right\}\ge 1-\alpha.
		\end{flalign*}as stated. 
	\end{proof}

	\section{Exponential Family Models in Natural Form}\label{sec:conjugacy}
	The following example shows that if the model is of the exponential family, then with conjugate priors the ACP has a closed-form. 
	
	Suppose $Y_1,\dots,Y_n\stackrel{iid}{\sim}p_\t(y)=\exp\{\eta(\theta)^\top S(y)-A(\t)\}h(y),
	$ where $S:\mathcal{Y}\rightarrow\mathbb{R}^{d_\t}$ is a vector of sufficient statistics,  $h:\mathcal{Y}\rightarrow \mathbb{R}$ a reference measure or density on the sample space $\mathcal{Y}$, and $A:\Theta\rightarrow\mathbb{R}$ the log-partition function (see \citealp{lehmann2006theory}, Section 1.5 for further details). Then, the joint density $p_\t^{(n)}(y)$ takes the form 
	$$
	p_\t^{(n)}(y)=\exp\left\{\eta(\theta)^\top \sum_{i=1}^{n}S(y_i)-nA(\t)\right\}\prod_{i=1}^{n}h(y_i),
	$$where $
	A(\theta)=\log \left[\int \exp \left\{\eta(\theta)^\top S(x)\right\} h(x) \dt \mu(x)\right],
	$ and $\mu(x)$ is the Lebesgue measure.
	
	Conducting inference on the natural parameter $\eta=\eta(\theta)$ is simplified by noting that if a conjugate prior is placed on $\eta$, then its ACP has a closed-form expression. In particular, 
	the above model has average scores $\overline{m}_n(\eta)=\nabla_\eta A(\eta)-Q_n$, where $S_n=n^{-1}\sum_{i=1}^{n}S(y_i)$. Since $A(\eta)$ is non-random and the variance of $\overline{m}_n(\eta)$ can be estimated using the sample variance
	$$
	W_n:=\frac{1}{n}\sum_{i=1}^{n}\left\{S(y_i)-S_n\right\} \left\{S(y_i)-S_n\right\} ^\top,
	$$ which does not depend on $\eta$. More generally, any consistent estimator of $\text{Cov}\{S(y_i)\}$ can be used. One can then consider inference on $\eta$ using $\pi(\eta\mid \mathcal{Q}_n)\propto \exp\{-nQ_n(\eta)\}\pi(\eta)$, where
	$
	Q_n(\eta)=\frac{1}{2}\left\{\nabla_\eta A(\eta)-S_n\right\}^\top W_n^{-1}\{\nabla_\eta A(\eta)-S_n\}.
	$
	
	Define the mean parameter $\mu$ by the function $\mu=g(\eta)=\nabla_\eta A(\eta)$. In regular models the function $g(\eta)$ exists and is invertible for all $\eta$. The parameter $\mu=\mu(\eta)$ is referred to as the mean parameterization of the model and satisfies $\mu=\E_{Y\sim P^{(n)}_\t }[Q_n]$. The form of the ACP for $\eta$ then follows by finding the ACP for $\mu$, and invoking a change of variables. 
	
	\begin{lemma}\label{lem:conjugate}
		Suppose that $\mathcal{Y}=\mathbb{R}^d$. If the (transformed) prior beliefs for the mean parameter $\mu=g(\eta)$ is $\pi(\mu)\propto \exp\{-\frac{1}{2}(\mu-\mu_0)^\top W_0^{-1}(\mu-\mu_0)\}$, then the ACP for $\eta$ is
		$
		\pi(\eta\mid \mathcal{Q}_n)=N\{b_n;g(\eta),\Sigma^{-1}_n\}|\nabla_\eta^2 A(\eta)|
		$, where 
		\begin{flalign*}
			\Sigma^{-1}_n&=n^{-1}W_0\left[n^{-1}W_n+W_0\right]^{-1}W_n,\\ b_n&=W_0\left[n^{-1}W_n+W_0\right]^{-1}S_n+W_0\left[n^{-1}W_n+W_0\right]^{-1}\frac{\mu_0}{n}.	
		\end{flalign*}
	\end{lemma}
	
	Lemma \ref{lem:conjugate} demonstrates that the ACP for the natural parameters $\eta$
	is Gaussian if the prior for $\mu$ is Gaussian. Interestingly, calculating $|\nabla_\eta^2 A(\eta)|$ can be avoided by first sampling $\widetilde\mu\sim N\{\mu;b_n,\Sigma_n^{-1}\}$, and then (numerically) inverting the equation $\widetilde\mu=g(\eta)$ to obtain the draw $\widetilde\eta$. The latter is feasible when $g(\eta)=\nabla_\eta A(\eta)$ can be reliably calculated. 
	

	\begin{proof}[Proof of Lemma \ref{lem:conjugate}]
		The ACP under the mean parameterization $\mu=g(\eta)=\nabla_\eta A(\eta)$, with prior $\pi(\mu)\propto\exp\{-\frac{1}{2}(\mu-\mu_0)^\top W_0^{-1}(\mu-\mu_0)\}$, where $\mu_0$ and $W_0$ are known hyper-parameters. Define $\bar{S}_n=n^{-1}Q_n$, and note
		$$
		n^{-1}m_n(\eta)=g(\eta)-\bar{S}_n=\mu-\bar{S}_n=n^{-1}m_n(\mu).
		$$Hence, writing $\bar{S}_n=n^{-1}Q_n$, the ACP for $\mu$ is 
		$$
		\pi(\mu\mid \mathcal{Q}_n)\propto \exp\left\{-\frac{n}{2}\left(\mu-\bar{S}_n\right)^\top W_n^{-1}\left(\mu-\bar{S}_n\right)\right\}\exp\left\{-\frac{1}{2}(\mu-\mu_0)^\top W_0^{-1}(\mu-\mu_0)\right\}.
		$$We now show that $\pi(\mu\mid \mathcal{Q}_n)=N\{\mu;b_n,\Sigma_n^{-1}\}$. With some algebra, 
		\begin{flalign*}
			&\exp\left\{-\frac{1}{2}\mu^\top\left[(W_n/n)^{-1}+W_0^{-1}\right]^{}\mu-\mu^\top\left[(W_n/n)^{-1}\bar{S}_n+\mu_0\right]\right\}	
			\\&=\exp\left\{-\frac{1}{2}\mu^\top\left[(W_n/n)^{-1}+W_0^{-1}\right]^{}\mu-\mu^\top\left[(W_n/n)^{-1}+W_0^{-1}\right]\left[(W_n/n)^{-1}+W_0^{-1}\right]^{-1}\left[(W_n/n)^{-1}\bar{S}_n+\mu_0\right]\right\}\\&\propto\exp\bigg{\{}-\frac{1}{2}\left(\mu-\Sigma_n^{-1}\left[(W_n/n)^{-1}\bar{S}_n+\mu_0\right]\right)^\top \Sigma_n \left(\mu-\Sigma_n^{-1}\left[(W_n/n)^{-1}\bar{S}_n+\mu_0\right]\right)\bigg{\}}\\&=\exp\left\{-\frac{1}{2}(\mu-b_n)^\top\Sigma_n(\mu-b_n)\right\}
		\end{flalign*}
		so that
		\begin{flalign*}
			\Sigma_n^{-1}&=\left[(W_n/n)^{-1}+W_0^{-1}\right]^{-1}=W_0\left[n^{-1}W_n+W_0\right]^{-1}W_n\frac{1}{n},
		\end{flalign*}
		where the second equality follows from the Woodbury identity, and 
		\begin{flalign*}
			b_n&=\Sigma_n^{-1}\left[(W_n/n)^{-1}\bar{S}_n+\mu_0\right]\\&=W_0\left[n^{-1}W_n+W_0\right]^{-1}\bar{S}_n+W_0\left[n^{-1}W_n+W_0\right]^{-1}W_n\frac{\mu_0}{n}\mu_0.
		\end{flalign*}
		Hence,  $\pi(\mu\mid \mathcal{Q}_n)=N\{\mu;b_n,\Sigma_n^{-1}\}$. For a regular exponential family, the parameter change from $\mu\mapsto \eta=g^{-1}(\mu)$ exists if the model is identifiable (in $\eta$). A change of variables $\mu\mapsto \eta$ then implies 
		\begin{flalign*}
			\pi(\eta\mid \mathcal{Q}_n)=\pi\{g(\eta)\mid \mathcal{Q}_n\}|\nabla_\eta g(\eta)|=N\{g(\eta);b_n,\Sigma_n^{-1}\}|\nabla_\eta^2 A(\eta)|,
		\end{flalign*}where the second equality follows since $g(\eta)=\nabla_\eta A(\eta)$. 
	\end{proof}

	{
		\spacingset{1.25} 
		\footnotesize
		\bibliographystyle{apalike}
		\bibliography{library}
	}

\end{document}